\documentclass[acmsmall,screen]{acmart}

\usepackage{complexity-mod}
\usepackage{totcount}
\usepackage{newfile}
\usepackage[capitalise]{cleveref}
\usepackage{thm-restate}

\usepackage{subcaption}
\usepackage{listings}
\lstdefinelanguage{pseudo}{morekeywords={init,with,or,if,then,else,fi,and,not,while,do,od,distinct,
    case, goto,local,algorithm, function, for, each, times, from, to,
    variables, procedure, recursive, return},
  morecomment=[l]{//}, morecomment=[s]{/*}{*/},
  mathescape=true,escapechar={@},
  basicstyle=\sffamily\small,
  commentstyle=\itshape\rmfamily\small,
  keywordstyle=\sffamily\bfseries\small
}
\usepackage{mathtools}
\usepackage{amsfonts,amsmath}

\usepackage{etoolbox}  
\usepackage{enumitem} 

\usepackage{stmaryrd}
\usepackage{bm}

\usepackage{hyperref}
\usepackage{environ}
\usepackage{csvsimple}
\usepackage{longtable} 
\usepackage[normalem]{ulem} 

\usepackage{fourier-orns} 
\usepackage{wasysym} 
\usepackage{marvosym} 
\usepackage{bbding} 

\usepackage{etoolbox}  
\usepackage{enumitem} 

\usepackage{enumerate}
\usepackage{mathrsfs}
\usepackage{mathtools}
\usepackage{appendix}
\usepackage{comment}
\usepackage{tikz}
\usepackage{array}
\usetikzlibrary{calc}
\usepackage{multicol}

\usetikzlibrary{arrows,shapes,snakes,automata,backgrounds,petri,positioning}
\usetikzlibrary{fit,backgrounds} 
\definecolor{processblue}{cmyk}{0.96,0,0,0}

\makeatletter
\newcommand{\dashact}[2][]{\ext@arrow 0359\rightarrowfill@@{#1}{#2}}
\def\rightarrowfill@@{\arrowfill@@\relax\relbar\dashrightarrow}
\def\arrowfill@@#1#2#3#4{%
	$\m@th\thickmuskip0mu\medmuskip\thickmuskip\thinmuskip\thickmuskip
	\relax#4#1
	\xleaders\hbox{$#4#2$}\hfill
	#3$%
}
\makeatother

\newcommand{\nn}{\mathbb{N}}

\newcommand{\zn}{\mathbb{Z}}

\newcommand{\qnz}{\mathbb{Q}_{+}}

\newcommand{\be}{\begin{enumerate}}
\newcommand{\ee}{\end{enumerate}}
\newcommand{\bc}{\begin{center}}
\newcommand{\ec}{\end{center}}

\newcommand{\bi}{\begin{itemize}}
\newcommand{\ei}{\end{itemize}}

\newcommand{\act}{\xrightarrow}
\newcommand{\Act}{\xRightarrow}

\newcommand{\fin}{q_f}

\newcommand{\node}{\mathsf{n}}

\usepackage[textsize=scriptsize,,disable]{todonotes}
\usepackage{hyperref}

\newcommand{\mach}{\mathcal{M}}

\newcommand{\bx}{\mathbf{x}}

\newcommand{\bz}{\mathbf{z}}
\newcommand{\bu}{\mathbf{u}}
\newcommand{\bv}{\mathbf{v}}
\newcommand{\bw}{\mathbf{w}}

\newcommand{\init}{q_{in}}

\newcommand\slice[2]{#1{\raise-.5ex\hbox{\ensuremath|}}_{#2}}

\newcommand{\cpp}[1]{#1{\footnotesize{\texttt{++}}}}
\newcommand{\cmm}[1]{#1{\footnotesize{\texttt{-\hspace{0.1pt}-}}}}
\newcommand{\cpptwo}[1]{#1{\footnotesize{\texttt{+=2}}}}
\newcommand{\cppfour}[1]{#1{\footnotesize{\texttt{+=4}}}}
\newcommand{\cmmtwo}[1]{#1{\footnotesize{\texttt{-\hspace{0.1pt}=2}}}}

\newcommand{\inc}{\mathbf{inc}}

\newcommand{\double}{\mathbf{double}}
\newcommand{\nop}{\mathbf{nop}}

\newcommand{\sta}[1]{\mathsf{state}(#1)}
\newcommand{\valu}[1]{\mathsf{val}(#1)}

\mathchardef\mhyphen="2D

\newcommand{\N}{\mathbb{N}}                    
\newcommand{\Z}{\mathbb{Z}}                    
\newcommand{\Q}{\mathbb{Q}}                    
\renewcommand{\vec}[1]{\bm{#1}}                

\newcommand{\supp}[1]{{\llbracket#1\rrbracket}}  
\renewcommand{\vec}[1]{\bm{#1}}                

\newcommand{\val}{\mathrm{val}}                

\usepackage{thm-restate}

\newcommand{\CGVAS}{\mathsf{CGVAS}}

\newcommand{\CRVASZ}{\mathsf{CRVASZ}}

\newcommand{\VAS}{\mathsf{VAS}}
\newcommand{\CVAS}{\mathsf{CVAS}}

\newcommand{\rl}{\mathsf{RL}}
\newcommand{\PVASS}{\mathsf{PVASS}}
\newcommand{\QVASS}{\qnz\mhyphen\mathsf{VASS}}
\newcommand{\QPVASS}{\qnz\mhyphen\mathsf{PVASS}}
\newcommand{\TCM}{\mathsf{2CM}^{\cdot 2,+1}_{\rl}} 
\newcommand{\IVASSRL}{[0,1]\mhyphen\mathsf{VASS}_{\rl}^{0?}}
\newcommand{\QVASSRL}{\qnz\mhyphen\mathsf{VASS}_{\rl}}

\newcommand{\CVASS}{\QVASS}
\newcommand{\MCM}{\TCM}
\newcommand{\NCCM}{\IVASSRL}

\def\cC{\mathcal{C}}
\def\cG{\mathcal{G}}
\def\cR{\mathcal{R}}
\def\cV{\mathcal{V}}
\def\cA{\mathcal{A}}
\def\cT{\mathcal{T}}
\def\cM{\mathcal{M}}
\def\rat{\mathbb{Q}}
\def\ratplus{\mathbb{Q}_+}
\def\integ{\mathbb{Z}}
\def\nat{\mathbb{N}}

\def\ratmonoid{((\mathbb{Q}^d)^*,\cdot,\varepsilon)}
\def\ratM{M_{\mathrm{rat}}}
\def\Mstar{M^{\bigstar}}
\newcommand{\addstar}[1]{{#1}^{\bigstar}}

\newcommand{\deriv}[1]{\xRightarrow{#1}} 
\newcommand{\derivstar}{\xRightarrow{*}} 
\newcommand{\step}[2]{\xrightarrow{#2}_{#1}} 
\newcommand{\mvec}[1]{\mathbf{#1}}
\def\pifwd{\pi_{\mathsf{fwd}}}
\def\pibwd{\pi_{\mathsf{bwd}}}
\def\extright{\upharpoonright}
\def\extleft{\upharpoonleft}

\def\ctrl{\mathrm{ctrl}}

\newcommand{\NEXPTIME}{\mathsf{NEXPTIME}}

\makeatletter

\newcommand{\eqxrightarrow}[2]{%
  \mathop{%
    \vtop{%
      \m@th 
      \offinterlineskip 
      \ialign{%
        \hfil##\hfil\cr
        \rightarrowfill\cr
        \hphantom{$\scriptstyle\mskip8mu{#2}\mskip8mu$}\cr
        \vrule height0pt width 1.5em\cr
        $\scriptscriptstyle {#1}$\cr
      }%
    }%
  }\limits^{#2}%
}
\makeatother

\newcommand{\Parikh}{\Psi}

\newcommand{\bala}[1]{\todo[color=blue!30]{\small #1}}

\newcommand{\gzin}[1]{\todo[color=green!30,inline]{\small GZ: #1}}
\newcommand{\balain}[1][]{\todo[color=blue!30,inline]{#1}}
\newcommand{\rst}[1]{\todo[color=yellow]{\small Ram: #1}}
\newcommand{\ramen}[1]{\todo[color=yellow!30, inline]{\small Ram: #1}}

\setcopyright{acmcopyright}
\copyrightyear{2018}
\acmYear{2018}
\acmDOI{XXXXXXX.XXXXXXX}

\begin{document}
	
	\title{Reachability in Continuous Pushdown VASS}

	\author{A. R. Balasubramanian}
	\email{bayikudi@mpi-sws.org}
	\affiliation{
		\institution{Max Planck Institute for Software Systems (MPI-SWS)}
		\city{Kaiserslautern}
		\country{Germany}}
	\orcid{https://orcid.org/0000-0002-7258-5445}
	\authornote{A part of the work was done when this author was at Technical University of Munich (TUM).}
	
	\author{Rupak Majumdar}
	\email{rupak@mpi-sws.org}
	\affiliation{
		\institution{Max Planck Institute for Software Systems (MPI-SWS)}
		\city{Kaiserslautern}
		\country{Germany}}
	\orcid{https://orcid.org/0000-0003-2136-0542}
	
	\author{Ramanathan S. Thinniyam}
	\email{ramanathan.s.thinniyam@it.uu.se}
	\affiliation{
		\institution{Uppsala University}
		\city{Uppsala}
		\country{Sweden}}
	\orcid{https://orcid.org/0000-0002-9926-0931}
	\authornote{A part of the work was done when this author was at Max Planck Institute for Software Systems (MPI-SWS).}
	
	\author{Georg Zetzsche}
	\email{georg@mpi-sws.org}
	\affiliation{
		\institution{Max Planck Institute for Software Systems (MPI-SWS)}
		\city{Kaiserslautern}
		\country{Germany}}
	\orcid{https://orcid.org/0000-0002-6421-4388}
	
	\renewcommand{\shortauthors}{A. R. Balasubramanian et al.}
	
	\begin{abstract}
		Pushdown Vector Addition Systems with States (PVASS) consist of finitely many control states, a pushdown stack, and a set of counters that can be incremented and decremented, but not tested for zero. 
		Whether the reachability problem is decidable for PVASS is a long-standing open problem.
		
		We consider \emph{continuous PVASS}, which are PVASS with a continuous semantics.
		This means, the counter values are rational numbers and whenever a vector is added to the current counter values, this vector is first scaled with an arbitrarily chosen rational factor between zero and one.
		
		We show that reachability in continuous PVASS is $\NEXPTIME$-complete.
		Our result is unusually robust: Reachability can be decided in $\NEXPTIME$ even if all numbers are specified in binary. 
		On the other hand, $\NEXPTIME$-hardness already holds for coverability, in fixed dimension, for bounded stack, and even if all numbers are specified in unary.
	\end{abstract}

	\ccsdesc[500]{Theory of computation~Models of computation}
	
	\keywords{Vector addition systems, Pushdown automata, Reachability, Decidability, Complexity}

	
	\maketitle

\section{Introduction} 
\label{sec:introduction}

\def\VASS{\mathrm{VASS}}
\def\PVASS{\mathrm{PVASS}}
\def\PDA{\mathrm{PDA}}

Pushdown $\VASS$ ($\PVASS$) is a model of computation which combines a Pushdown Automaton ($\PDA$) and a Vector Addition System with States ($\VASS$) by using both a stack and 
counters. 
Since $\PDA$s naturally model sequential computation with recursion \cite{RHS95,alurAnalysisRecursiveState} 
and $\VASS$s naturally model concurrency \cite{karpParallelProgramSchemata1969}, 
the combination of the two is an expressive modelling paradigm. 
For instance, $\PVASS$ can be used to model recursive programs with unbounded data domains \cite{atigApproximatingPetriNet2011}, beyond the capability of $\PDA$ alone.
They can also model context-bounded analysis of multi-threaded programs \cite{ABQ09}, even when one thread can have arbitrarily many context switches. 
In program analysis, one-dimensional PVASS models certain pointer analysis problems \cite{Pavlogiannis,LiReps}.

Many verification problems can then be reduced to the \emph{reachability} problem \cite{hack1976decidability}, where 
one asks, given a $\PVASS$ $\mathcal{M}$ and two of its configurations $c_0$ and $c_f$, whether there is a run of $\mathcal{M}$ that starts at $c_0$ and ends at $c_f$. 
Unfortunately, in spite of our understanding of the reachability problem in the context of $\PDA$ \cite{Yannakakis90,RHS95,alurAnalysisRecursiveState,chaudhuriSubcubicAlgorithmsRecursive}  and 
$\VASS$ models \cite{DBLP:conf/lics/LerouxS19,leroux2022reachability,czerwinskiReachabilityVectorAddition2022}, the decidability of $\PVASS$ reachability remains open \cite{schmitzCoverabilityUndecidableOneDimensional2019a,englertLowerBoundCoverability2021,ganardiReachabilityBidirectedPushdown2022}, even with just one stack and one $\VASS$ counter. 
Hence a natural approach to take is to consider \emph{overapproximations} of the system behaviour, both to build up theoretical understanding and to approximate verification questions in practice.
The approach we take in this paper is to \emph{approximate} the model via continuous semantics.
A continuous semantics gives an over-approximation of the possible behaviors, so if the relaxed program cannot reach a location, neither can the original one.

By \emph{continuous} semantics of $\VASS$, we mean that a transition is allowed to be fired \emph{fractionally}, allowing the addition or removal of a \emph{rational fraction} 
of tokens from a counter.
This model was first introduced by David and Alla in the context of Petri nets \cite{david1987continuous}.
Continuous $\VASS$ ($\CVASS$) were studied by \citet{blondinLogicsContinuousReachability2017}, who showed reachability and coverability are $\NP$-complete.
Approximation via continuous semantics has allowed the application of SMT solvers and the development of state-of-the-art 
solvers from an empirical perspective \cite{blondinApproachingCoverabilityProblem2016} to the coverability problem for $\VASS$. 
More generally, relaxing integer-valued programs to continuous-valued programs is a common approximation in invariant synthesis \cite{SrivastavaGulwaniFoster,ColonSankaranarayananSipmaCAV03}.

In addition to its use as an overapproximation, the continuous semantics
captures the exact behavior in systems where update vectors represent change
rates (to real variables such as temperature, energy consumption) per time
unit. Here, fractional addition corresponds to executing steps that take at
most one time unit. For example, $\CVASS$ are constant-rate multi-mode
systems~\cite{AlurTW12} where each action takes at most one time unit.
Continuous PVASS can then be seen as recursive programs with such constant-rate
dynamics.
\balain[Few points for intro that we discussed before: 
For checkers that use Farka's lemma style arguments, our result gives expressive power; 
over-approx. for context-bounded programs where one process is not context-bounded; 
over-approx for interleaved Dyck reachabiility. 
In addition, do we exhibit a program for which our approximation is able to prove unreachability?]

\smallskip
\noindent\textbf{Contribution} 
We study $\PVASS$ with continuous semantics (denoted by $\qnz\mhyphen\mathsf{PVASS}$), where we allow fractional transitions on counters, but retain the discrete nature of the stack. 
Hence, a \emph{configuration} is a tuple $(q,\gamma,\bv)$ where
$q$ is the control-state, $\gamma$ is the stack content and $\bv$
represents the counter values.
We show that reachability is decidable, and we provide a comprehensive complexity landscape for reachability, coverability, and state reachability. The \emph{reachability problem} asks for given configurations $(q,\gamma,\bv)$ and $(q',\gamma',\bv')$, whether from $(q,\gamma,\bv)$, the system can reach $(q',\gamma',\bv')$.
The \emph{coverability problem} asks for given configurations $(q,\gamma,\bv)$ and $(q',\gamma',\bv')$, whether from 
$(q,\gamma,\bv)$, the system can reach a configuration of the form $(q',\gamma',\bv'')$ where $\bv'' \ge \bv'$). Moreover, \emph{state reachability} asks for a given configuration $(q,\gamma,\bv)$ and a state $q'$, whether from $(q,\gamma,\bv)$, the system can reach a configuration of the form
$(q',\gamma',\bv')$ for some $\gamma'$ and $\bv'$.

Our main result is the following: 
\begin{theorem}\label{main-result}
	Reachability in $\qnz\mhyphen\mathsf{PVASS}$ is $\NEXPTIME$-complete.
\end{theorem}
The $\NEXPTIME$-completeness is unusually robust. 
Specifically, the complexity is not affected by (i)~whether we consider reachability or coverability, or (ii)~whether the number of counters is part of the input or fixed, or (iii)~whether counter values (in updates and configurations) are encoded in unary or binary. 
This is summarized in the following stronger version:
\begin{theorem}\label{main-result-finer}
	Reachability in $\qnz\mhyphen\mathsf{PVASS}$, with binary encoded
	numbers, is in $\NEXPTIME$. Moreover, $\NEXPTIME$-hardness already
	holds for coverability in $85$-dimensional
	$\qnz\mhyphen\mathsf{PVASS}$, with unary encoded numbers, and bounded stack-height.
\end{theorem}
%
%
Further, if we allow the configurations to be encoded in binary, then hardness
already holds for coverability in 13-dimensional $\qnz\mhyphen\mathsf{PVASS}$. 

Our result is in stark contrast to reachability problems in classical VASS:
It is well-known that there, coverability is $\EXPSPACE$-complete~\cite{CovPetrinets}, whereas general reachability is Ackermann-complete~\cite{leroux2022reachability,czerwinskiReachabilityVectorAddition2022}. 
Furthermore, fixing the dimension brings the complexity down to primitive-recursive~\cite{DBLP:conf/lics/LerouxS19} (or from $\EXPSPACE$ to $\PSPACE$ in the case of coverability \cite{RosierY86}). 

Another surprising aspect is that for continuous PVASS, the coverability problem and the state reachability problem do not have the same complexity. We also show:
\begin{theorem}\label{main-result-state-reachability}
	The state reachability problem for $\qnz\mhyphen\mathsf{PVASS}$ is $\NP$-complete.
\end{theorem}
This is also in contrast to the situation in PVASS: There is a simple reduction from the reachability problem to the coverability problem~\cite{lerouxCoverabilityProblemPushdown2015}, and from there to state reachability. Thus, the three problems are polynomial-time inter-reducible for PVASS.


Our results are based on a number of novel technical ingredients.
Especially for our lower bounds, we show a number of subtle constructions that enable us to encode discrete computations of
bounded runs of counter machines in the continuous semantics.

\smallskip
\noindent\textbf{Ingredient I: Upper bound via rational arithmetic} We prove the
$\NEXPTIME$ upper bound by observing that a characterization of runs in a cyclic $\CVASS$ (meaning: the initial state is also the only final one) by
\citet{blondinLogicsContinuousReachability2017} still holds in a more general setting of cyclic $\qnz\mhyphen\mathsf{PVASS}$.
We apply this observation by combining several (known) techniques. As is standard in the analysis of PVASS~\cite{lerouxCoverabilityProblemPushdown2015,englertLowerBoundCoverability2021}, we view runs as derivations in a suitable grammar. As usual, one can then decompose each derivation tree into an acyclic part and 
``pump derivations'' of the form $A\derivstar uAv$ for some non-terminal $A$. Such pumps, in turn, can 
by simulated by a cyclic $\qnz\mhyphen\mathsf{PVASS}$. Here, to simulate $A\derivstar uAv$, one applies $u$ as is and one applies $v$ in reverse on a separate set of counters. This idea of simulating grammar derivations by applying ``the left part forward'' and ``the right part backward'' is a recurring theme in the literature on context-free grammars (see, e.g.~\cite{rosenberg1967machine,Berstel1979,DBLP:journals/pacmpl/BaumannGMTZ23,DBLP:conf/mfcs/LohreyRZ22,DBLP:conf/popl/RepsTP16}) and has been applied to PVASS by \citet[Section~5]{lerouxCoverabilityProblemPushdown2015}.

As a consequence, reachability can be decided by guessing an exponential-size
formula of Existential Linear Rational Arithmetic (ELRA). 
Since satisfiability for ELRA is in $\NP$, this yields an $\NEXPTIME$ upper bound.  

\smallskip
\noindent\textbf{Ingredient II: High precision and zero tests}
For our lower bound, we reduce from reachability in machines with two counters, which can only be doubled and incremented. A run in these machines is accepted if it is of a particular length (given in binary) and has both counters equal in the end. We call such machines $\MCM$, the $\RL$ stands for ``run length''. 
This problem is $\NEXPTIME$-hard by a reduction from a variant of the Post Correspondence Problem where the word pair is restricted to have a specified length (given in binary)~\cite{AiswaryaMS22}.

We give the desired reduction from $\MCM$, through a series of intricate constructions that ``control'' the fractional firings.
We go through an intermediate machine model called $\NCCM$. An $\NCCM$
is just like a $\CVASS$, except that the counter values are
constrained to be in the interval $[0,1]$ and we allow zero tests.
Further, we only consider runs up to a particular length (given in binary),
as indicated by the $\RL$ subscript. When reducing from $\MCM$, we are 
confronted with two challenges: First, $\NCCM$ cannot store numbers beyond 1 and second, $\NCCM$ cannot natively double numbers. The key idea
here is that, since we only consider runs of a $\MCM$ up to a given length $m$, the counter values of a $\MCM$ are bounded by $2^m$ along any run.
Hence, instead of storing the counter values of a $\MCM$ exactly, we instead use \emph{exponential precision}.
We encode a number $n\in\N$ by $\tfrac{n}{2^m}\in[0,1]$.
Since then all the values are in $[0,1]$, we can double the counter values in a $\NCCM$ by forcing the firing 
fraction of the $\NCCM$ to be a particular value. 
The firing fraction is controlled, in turn, by means of the zero tests. 

\smallskip
\noindent\textbf{Ingredient III: Constructing precise numbers}
In order to simulate increments of the $\MCM$ in our $\NCCM$, we need to be able to add $\tfrac{1}{2^m}$ to a counter. To this end, we present a $\NCCM$ gadget of polynomial size that produces the (exponentially precise) number $\frac{1}{2^m}$ in a given counter. 
The idea is to start with $1$ in the counter and halve it $m$ times. 
The trick is to use an additional counter that goes up by $1/m$ for each halving step. 
Checking this counter to be $1$ in the end ensures that exactly $m$ halvings have been performed.

\smallskip
\noindent\textbf{Ingredient IV: Zero tests via run length}
We then reduce $\NCCM$ to $\QVASSRL$, which are $\QVASS$ with a run-length constraint. Here, we need to (i)~make sure that the counter values remain in $[0,1]$ and (ii)~simulate zero tests. We achieve both by introducing a \emph{complement counter} $\bar{x}$ for each counter $x$, where it is guaranteed that $x+\bar{x}=1$ at all times. This means that instead of checking $x=0$, we can check $\bar{x}=1$ by subtracting $1$ from $\bar{x}$. However, this does not suffice---we need to ensure that the firing fraction is exactly $1$ in these steps. Here, the key idea is, whenever we check $\bar{x} = 1$, we also increment at the same time (and thus with the same firing fraction), a separate counter called the \emph{controlling counter}, which in the end must equal $Z$, the total number of zero tests. This exploits the fact that every run attempts the same pre-determined number of zero tests due to the run-length constraint.
If the controlling counter reaches the value $Z$ at the very end,
then we are assured that every zero-test along the run was indeed
performed correctly.

Finally, we reduce from $\QVASSRL$ to $\QPVASS$ by using the pushdown stack to count from zero up to a number specified in binary. This employs a standard trick for encoding a binary number on the stack, where the least significant bit is on top. We further show that the 
final $\QPVASS$ that we construct has bounded stack-height, 13 counters, and also that the
target configuration can be reached from the source configuration
if and only if the target can be covered from the source.
This proves that coverability is hard even for a constant number of counters.

\smallskip
\noindent\textbf{Ingredient V: Unary encodings}\label{ingre:v} The above reduction
produces instances of $\QPVASS$ where the 
configurations are encoded in binary. Proving hardness for
unary encodings requires more ideas. First, by using a trick akin to exponential precision from above, we show that
hardness of coverability in $\QPVASS$ holds already when all the values 
of the given configurations are less than 1. Next,
by reusing the doubling and the halving gadgets from Ingredients~II and~III, we show that for any fraction $p/2^k$ where $p$ is given in binary,
there exists an \emph{amplifier}, i.e., there is a $\QVASS$ of polynomial size in $\log(p)$ and $k$, which starting from an unary-encoded configuration
is able to put the value $p/2^k$ in a given counter. We then simply plug in
a collection of these amplifiers before and after our original $\QPVASS$ to get the desired result.

\smallskip
\noindent\textbf{Related work.}
There have been several attempts to study restrictions or relaxations of the $\PVASS$ reachability problem.
For example, reachability is decidable when one is allowed to simultaneously test the first $i$ counters of a $\VASS$ for zero for any $i$ \cite{reinhardt2008reachability}; this model can be seen as a special case of $\PVASS$ \cite{atigApproximatingPetriNet2011}. Furthermore, the coverability problem in one-dimensional $\PVASS$ is decidable~\cite{lerouxCoverabilityProblemPushdown2015} and $\PSPACE$-hard~\cite{englertLowerBoundCoverability2021}.
Reachability is decidable for \emph{bidirected} $\PVASS$ \cite{ganardiReachabilityBidirectedPushdown2022}, although the best known upper bound is Ackermann time (primitive recursive time in fixed dimension).
Our work is in the same spirit.
The continuous semantics reduces the complexity of reachability from Ackermann-complete for $\VASS$ to $\NP$-complete \cite{blondinLogicsContinuousReachability2017} (and even to $\P$ for Petri nets \cite{fraca2015complexity}).
Our results show that the presence of a stack retains decidability, but allows exponentially more computational power.

\smallskip
All missing proofs can be found in the appendix of this paper.
\section{Preliminaries} 
\label{sec:preliminaries}

We write $\rat$ for the set of rationals and $\ratplus$ for the set of nonnegative rationals. Vectors over $\rat$ (or $\ratplus$) are written in bold ($\mvec u, \mvec v$ etc.) and are represented as a pair of natural numbers (numerator and denominator) for each rational.
Throughout this paper, all numbers will be encoded in binary, unless stated otherwise. Note that, this means that, each rational number
is a pair of natural numbers, with both of them encoded in binary.  

\smallskip\noindent\textbf{Machine models.} 
A $d$-dimensional \emph{Continuous Vector Addition System with States} ($d$-$\CVASS$ or simply $\CVASS$) $\mach = (Q,T,\Delta)$ consists of  
a finite set $Q$ of states, 
a finite set 
$T \subseteq \zn^d$ of transitions, and 
a finite set $\Delta \subseteq Q \times T \times Q$ of rules. 
We will on occasion consider an infinite $\CVASS$ where $T$ continues to be finite, but $Q$ and $\Delta$ are infinite. 

A \emph{configuration} of $\mach$ is a tuple $C=(q,\bv)$ where $q$ is a state
and $\bv \in \qnz^d$ is a valuation of the counters. We use the notations $\sta{C}, \valu{C}, C(i)$ to denote $q, \bv, \bv(i)$ respectively. Let $I = (q,t,q') \in \Delta$ be a rule and let $\alpha \in (0,1]$ be the \emph{firing fraction}. A \emph{step} from a configuration $C$ to another
configuration $C'$ by means of the pair $(\alpha,I)$ (denoted by $C \act{\alpha I} C')$ is possible if{}f $\sta{C} = q, \sta{C'} = q'$ and $\valu{C'} = \valu{C} + \alpha t$. A run of $\mach$ is a finite sequence of steps $C_0 \act{\alpha_1 I_1} C_1\act{\alpha_2 I_2}\ldots \act{\alpha_n I_n} C_n$, where $\alpha_1 I_1 \ldots \alpha_n I_n$ is called a \emph{firing sequence}, and we say $C_n$ is \emph{reachable} from $C_0$ (written $C_0 \act{\alpha_1 I_1, \alpha_2 I_2, \dots, \alpha_n I_n} C_n$ or $C_0 \act{*} C_n$). 

We assume the reader is familiar with context-free grammars and give basic definitions and notation (see, e.g., \cite{sipserIntroductionTheoryComputation2012}). 
A \emph{context-free grammar} $\cG=(S,N,\Sigma,P)$ consists of a finite set of nonterminals $N$, a starting nonterminal $S$, a finite alphabet $\Sigma$ and a finite set of production rules $P \subseteq N \times (N \cup \Sigma)^*$. We will assume that $\cG$ is in Chomsky Normal Form.
As usual $\derivstar$ is the reflexive, transitive closure of $\deriv{}$. 
A word $w \in \Sigma^*$ belongs to the language $L(\cG)$ of the grammar iff $S \smash{\derivstar} w$.

\newcommand{\stack}{\mathsf{stack}}
A \emph{Continuous Pushdown $\VASS$} ($\QPVASS$) is a $\CVASS$ additionally equipped with a stack. 
Formally, it is a tuple $\mach=(Q,\Gamma,T,\Delta)$ where $Q$ is a finite set of states, $\Gamma$ is a finite stack alphabet, 
$T \subseteq \zn^d \times (\Gamma \cup \bar{\Gamma} \cup \varepsilon)$ is a finite set of transitions, and 
$\Delta \subseteq Q \times T \times Q$ is a finite set of rules. 
A configuration $C=(q,w,\bv)$ of $\mach$ contains additionally the stack $w \in \Gamma^*$ and we write $w=\stack(C)$. 
A \emph{step} $C \act{\alpha I} C'$ using rule $I=(q,a,t,q')$ is possible if{}f $\sta{C} = q$, $\sta{C'} = q'$, $\valu{C'} = \valu{C} + \alpha t$,
and one of the following holds: (1) $a \in \Gamma$ and $\stack(C')=a\;\stack(C)$, (2) $a \in \bar{\Gamma}$ and $\stack(C)=a\;\stack(C')$, or 
(3) $a=\varepsilon$ and $\stack(C)=\stack(C')$. 
The notion of run, firing sequence, and reachability is defined as for $\CVASS$. 
In some cases, we will need to extend the notion of step to allow vectors in $\rat^d$ in a configuration rather 
than just $\ratplus^d$. 
We then explicitly specify this in the form of an underscript: $\step{\ratplus}{}$ or $\step{\rat}{}$ to make it clear.




\smallskip\noindent\textbf{Decision Problems.} The reachability problem for $\QPVASS$ is defined as follows:
\begin{quote}
	Given a $\QPVASS$ $\mach$ and two of its configurations $C_0,C_1$, is $C_1$ reachable from $C_0$?
\end{quote}
The coverability problem for $\QPVASS$ is defined as follows:
\begin{quote}
	Given a $\QPVASS$ $\mach$ and two of its configurations $C_0,C_1=(q_1,w_1,\bv_1)$, does there exist a configuration $C'=(q_1,w_1,\bv_1')$ with $\bv_1' \geq \bv_1$ such that $C'$ reachable from $C_0$?
\end{quote}
The state reachability problem is defined as follows:
\begin{quote}
	Given a $\QPVASS$ $\mach$, a configuration $C_0$ and a state $q$, does there exist a configuration $C_1$ with $\sta{C_1}=q$ that is reachable from $C_0$?
\end{quote}

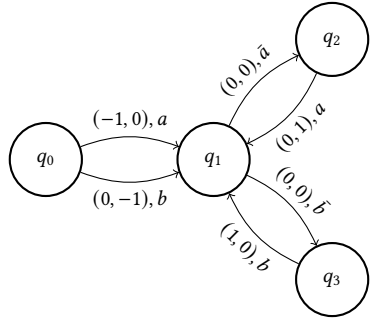
\begin{figure}
	\begin{center}
		\tikzstyle{node}=[circle,draw=black,thick,minimum size=12mm,inner sep=0.75mm,font=\normalsize]
		\tikzstyle{edgelabelabove}=[sloped, above, align= center]
		\tikzstyle{edgelabelbelow}=[sloped, below, align= center]
		\begin{tikzpicture}[->,node distance = 1.25cm,scale=0.8, every node/.style={scale=0.8}]
			\node[node] (q0) {$q_0$};
			\node[node, right = of q0] (q1) {$q_1$};
			\node[node, above right = of q1] (q2) {$q_2$};
			\node[node, below right = of q1] (q3) {$q_3$};
			
			\draw(q0) edge[bend left = 20, edgelabelabove] node[above]{$(-1,0), a$} (q1);
			\draw(q0) edge[bend right = 20, edgelabelabove] node[below]{$(0,-1), b$} (q1);
			\draw(q1) edge[bend left = 20, edgelabelabove] node[above]{$(0,0), \bar{a}$} (q2);
			\draw(q2) edge[bend left = 20, edgelabelabove] node[below]{$(0,1), a$} (q1);
			\draw(q1) edge[bend left = 20, edgelabelabove] node[above]{$(0,0), \bar{b}$} (q3);
			\draw(q3) edge[bend left = 20, edgelabelabove] node[below]{$(1,0), b$} (q1);
		\end{tikzpicture}
	\end{center}
	\caption{An example $\QPVASS$ with 2 counters. Here $a$ and $b$ are the stack symbols.}
	\label{fig:example-cpvass}
\end{figure} 

\begin{example}
	Let us consider the $\QPVASS$ from Figure~\ref{fig:example-cpvass}, which we shall denote by $\mach$.
	It has 2 counters and stack symbols $a$ and $b$. Recall that a label $a$ represents a push of $a$ and $\bar{a}$ represents a pop of $a$.
	There are only two outgoing rules from the state $q_0$: the first rule $r_1$ decrements the first counter by 1, does not modify the second counter and pushes $a$ onto the stack and the second rule $r_2$ decrements the second counter by 1, does not 
	modify the first counter and pushes $b$ onto the stack. 
	Hence, starting from the configuration $(q_0,\varepsilon,(0,0))$ it is not possible to reach a configuration whose state is $q_1$. This implies that the input $(\mach,(q_0,\varepsilon,(0,0))),q_1)$ is 
	a negative instance of the state reachability problem.
	
	On the other hand, starting from $(q_0,\varepsilon,(1,1))$, by firing $r_1$ with fraction $0.5$, 
	we can reach $(q_1,a,(0.5,1))$. This means that $(\mach,(q_0,\varepsilon,(1,1)),q_1)$ is a positive
	instance of the state reachability problem. Moreover, this also means that
	$(\mach,(q_0,\varepsilon,(1,1)),(q_1,a,(0.5,0.5)))$
	is a positive instance of the coverability problem. 
	
	However, the input $(\mach,(q_0,\varepsilon,(1,1)),(q_1,a,(1,1)))$ is a negative instance
	of the coverability problem. To see this, suppose for the sake of contradiction, a run exists
	between $(q_0,\varepsilon,(1,1))$ and $(q_1,a,(n_1,n_2))$ for some $n_1 \ge 1, n_2 \ge 1$.
	The first step of this run has
	to fire either $r_1$ or $r_2$ by some non-zero fraction $\alpha$. Suppose $r_1$ is fired. (The argument is similar for the other case). 
	Then $a$ gets pushed onto the stack and the value of the first counter becomes $1-\alpha$. 
	From that point onwards, the only rules that can be fired are the ones going in and out of the state $q_2$, both of which 
	do not increment the first counter. Hence, the first counter will have $1-\alpha$ as its value throughout the run, which leads to a contradiction. 
	
	Finally, note that starting from $(q_0,\varepsilon,(1.1,0.6))$, it is possible to 
	reach $(q_1,a,(1,1))$: first, fire $r_1$ with fraction $0.1$, then fire the incoming rule
	to $q_2$ (which pops $a$) with fraction 1  and then fire the outgoing rule from $q_2$ (which pushes $a$) with fraction $0.4$. Hence, $(\mach,(q_0,\varepsilon,(1.1,0.6)),(q_1,a,(1,1)))$ is a positive
	instance of the reachability problem.
\end{example}


\section{Upper Bound for Reachability}\label{sec:upper-bound}

\tikzset{
	nonterminal/.style={fill,circle,inner sep=0, minimum size=3pt},
	terminal/.style={inner sep=3pt,yshift=-0.4cm},
	epath/.style={edge from parent/.style={very thick,draw}}, 
	npath/.style={edge from parent/.style={draw,thin}},
}

\newcommand{\theight}{0.9cm}
\newcommand{\twidth}{0.9cm}
\pgfdeclareshape{ltree}{
	\nodeparts{}
	\anchor{north}{\pgfpoint{0cm}{0cm}}
	\anchor{center}{\pgfpoint{0cm}{0cm}}
	\behindbackgroundpath{
		\path [draw,thin,fill=blue!30] (0,0) -- (0,-\theight) -- (-\twidth,-\theight) -- cycle;
	}
}
\pgfdeclareshape{rtree}{
	\nodeparts{}
	\anchor{north}{\pgfpoint{0cm}{0cm}}
	\anchor{center}{\pgfpoint{0cm}{0cm}}
	\behindbackgroundpath{
		\path [draw,thin,fill=blue!30] (0,0) -- (0,-\theight) -- (\twidth,-\theight) -- cycle;
	}
}

\newcommand{\ftheight}{4cm}
\newcommand{\ftwidth}{3cm}





We first prove the $\NEXPTIME$ upper bound in \cref{main-result-finer}.
To this end, we first use a standard language-theoretic translation to slightly rephrase the reachability problem in $\QPVASS$ (this slight change in viewpoint is also taken in other work on $\PVASS$~\cite{lerouxCoverabilityProblemPushdown2015, englertLowerBoundCoverability2021}).
Observe that when we are given two configurations $C_0$ and $C_1$ of a $\QPVASS$, we want to know whether there exists a sequence $w\in(\Z^d)^*$ of update vectors such that (i)~there exists a sequence $\sigma$ of transitions that applies $w$, such that $\sigma$ is a valid run from $C_0$ to $C_1$ in the pushdown automaton underlying the $\QPVASS$ (thus ignoring the counter updates) and (ii)~there exist firing fractions for each vector in $w$ such that adding the resulting update vectors will lead from $\mvec u$ to $\mvec v$, where $\mvec u,\mvec v$ are the vectors in the configurations $C_0$ and $C_1$.
Now observe that the set of words $w$ as in condition (i) are a context-free language. Therefore, we can phrase the reachability problem in $\QPVASS$ by asking for a word in a context-free language that satisfies condition~(ii). 

Let us make condition~(ii) precise.
Let $\Sigma\subseteq\Z^d$ be the finite set of vectors that appear as transition labels in our $\QPVASS$.
Given two configurations $\mvec u, \mvec v \in \qnz^d$ and a word $w = w_1w_2\dots w_n\in \Sigma^*$ with each $w_i \in \Sigma$, 
we say that $\mvec u\step{\ratplus}{w}\mvec v$ iff 
there exist $\alpha_1, \dots, \alpha_n$ such that $\mvec u \step{\ratplus}{\alpha_1 w_1, \alpha_2 w_2, \dots, \alpha_n w_n} \mvec v$. 
Similarly, given a language $L \subseteq \Sigma^*$, we say that
$\mvec u \step{\ratplus}{L} \mvec v$ iff $\mvec u \step{\ratplus}{w} \mvec v$ for some $w \in L$.
By our observation above, the reachability problem in $\QPVASS$ is equivalent to the following problem:  
\begin{description}
	\item[Given] A set of vectors $\Sigma \subseteq \Z^d$, a context-free language $L\subseteq \Sigma^*$ and $\mvec u,\mvec v\in\ratplus^d$.
	\item[Question] Does $\mvec u\step{\ratplus}{L}\mvec v$?
\end{description}

We solve this problem using results about the existential fragment of the first-order theory of $(\rat, +, <)$, which we call Existential Linear Rational Arithmetic (ELRA).
Our algorithm constructs an ELRA formula for the following relation $R_L$.
\begin{definition}
	\label{defn:reachRelation}
	The \emph{reachability relation $R_L$} corresponding to a context-free language $L\subseteq(\Z^d)^*$ is given by  $R_L=\{(\mvec u,\mvec v)\in\ratplus^d\times\ratplus^d\mid \mvec u\step{\ratplus}{L}\mvec v\}$.
\end{definition}

The following definition of computing a formula using a non-deterministic algorithm is inspired by the definition of \emph{leaf language} from complexity theory \cite{papaComplexityBook2007}.
We say that one can \emph{construct an ELRA formula in $\NEXPTIME$ (resp. $\NP$) for a
	relation $R\subseteq\qnz^n$} if there is a non-deterministic exponential (resp. polynomial)
time-bounded Turing machine such that every accepting path of the machine computes
an ELRA formula such that if $\varphi_1,\ldots,\varphi_m$ are the produced
formulae, then their disjunction $\bigvee_{i=1}^m \varphi_i$ defines the
relation $R$. Here, a formula $\phi$ is said to define a relation $R\subseteq\qnz^n$ if for every $n$-tuple $\vec u \in \qnz^n$, we have $R(\vec u)$ holds iff $\phi(\vec u)$ is true of the rational numbers.

\begin{proposition} \label{prop:ELRAformula}
	Given a context-free language $L\subseteq(\Z^d)^*$ one can construct in $\NEXPTIME$ an ELRA formula for the relation $R_L$.
\end{proposition} 

Since the truth problem for ELRA formulae can be solved in $\NP$~\cite{sontagRealAdditionPolynomial1985}, the $\NEXPTIME$ upper bound follows from \cref{prop:ELRAformula}:
Our algorithm would first non-deterministically compute a disjunct $\varphi$ of the ELRA formula for $R_L$ and then check the truth of $\varphi$ in $\NP$ in the size of $\varphi$. 
This is a non-deterministic algorithm that runs in exponential time.

Therefore, the remainder of this section is devoted to proving \cref{prop:ELRAformula}.
The key difficulty lies in understanding the reachability relation along \emph{pumps}, which are derivations of the form $A\derivstar wAw'$ for some non-terminal $A$.


\begin{definition}
	\label{defn:cyclicReachRelationA}
	Let $\cG$ be a context-free grammar over $\Z^d$ and $A$ a non-terminal in $\cG$.
	The \emph{pump reachability relation} is defined as
	\[ P_A=\left\{(\mvec u,\mvec v, \mvec{u}', \mvec{v}')\in\Q_+^d\times\Q_+^d\times\Q_+^d\times\Q_+^d \mid \exists w,w'\colon A\derivstar wAw',~\mvec{u}\step{\ratplus}{w}\mvec{v},~\mvec{u}'\step{\ratplus}{w'}\mvec{v}'\right\}. \]
\end{definition}

\begin{theorem}\label{thm:ratArithCycleFormula}
	Given a context-free grammar $\cG$ over $\Z^d$ and a non-terminal $A$ in $\cG$, one can compute in $\NEXPTIME$ an ELRA formula for the relation $P_A$.
	
	%
\end{theorem}

Before proving \cref{thm:ratArithCycleFormula}, we first show how \cref{prop:ELRAformula} follows from \cref{thm:ratArithCycleFormula}.

Let $\cG$ be a grammar for the language $L$ in \cref{prop:ELRAformula}. Consider an arbitrary derivation tree $\cT$ of $\cG$. We say a derivation tree is \emph{pumpfree} if along every path of the tree, every nonterminal occurs at most once. Clearly, an arbitrary derivation tree $\cT$ can be  obtained from a pumpfree tree by inserting ``pumping'' derivations of the form $A \derivstar wAw'$. 
Since every pumpfree tree is exponentially bounded in size (since its depth is bounded by the number of nonterminals $|N|$), there can only be exponentially many such pumps that need to be inserted for any given nonterminal $A$. 

A pump $A \derivstar vAx$ on a nonterminal $A$ in an arbitrary derivation tree $\cT$ can be replaced by additional terminal letters called \emph{pump letters}  $(A,\mvec n)$ and $\overline{(A,\mvec n)}$ as shown in \cref{fig:insertCycles}. Let the two occurrences of $A$ be the first and last occurences of $A$ along a path. Here $\mvec n \in \{ 0,1\}^*$ is a vector denoting the node which is labelled by the first $A$. Note that we assume that the grammar is in Chomsky Normal Form and hence nodes in the derivation tree can be identified in this manner since the trees are binary trees.
The tree $\cT'$ contains four children at $\mvec n$, with the first and fourth being pump letters and the second and third being labelled by the nonterminals $B,C$ occurring in the production $A \deriv{} BC$. It could also be the case that the rule used is of the form $A \deriv{} a$, in which case there are only three letters: the middle letter being $a$ and the other two pump letters. Repeating this replacement procedure along each path, we finally obtain a pumpfree tree $\tilde{\cT}$ which does not have a repeated nonterminal along any path. Since $\tilde{\cT}$ is exponentially bounded in size, the number of pump terminals introduced is also exponentially bounded. In particular, every vector $\mvec n \in \{0,1 \}^h$ for $h \leq |N|$ where $N$ is the set of nonterminals of $\cG$.

\begin{figure}[t]
	\begin{center}
		\begin{tikzpicture}[level distance=0.7cm, sibling distance=0.8cm]
			\newcommand{\godown}{-0.5cm}
			\newcommand{\goleft}{-0.2cm}
			
			\path [draw,thin] (0,0) -- (-\ftwidth,-\ftheight) -- (\ftwidth,-\ftheight) -- cycle;
			\node at (-2cm,-1cm) {\large $\mathcal{T}=$};	
			
			\path[draw,thick] (0,0)--(\goleft,\godown)--($(-\goleft,2*\godown)$)--($(\goleft,3*\godown)$)--($(-\goleft,4*\godown)$)--($(\goleft,5*\godown)$);
			\coordinate (a1) at ($(-\goleft,2*\godown)$);
			
			\coordinate (a2) at ($(\goleft,5*\godown)$);
			\node (n1) at ($(a1)+(2,0)$) {node $\vec n$};
			\path[draw, ->] (a1)--(n1);
			
			\node at ($(-\goleft-0.2cm,2*\godown)$) {$A$};
			
			\node at ($(\goleft-0.1cm,5*\godown)$) {$A$};	
			\path[draw,thin] (a2)--($(-0.3*\ftwidth,-\ftheight)$)--($(0.3*\ftwidth,-\ftheight)$)--(a2);
			\path[draw,thin] (a1)--($(-0.7*\ftwidth,-\ftheight)$)--($(0.7*\ftwidth,-\ftheight)$)--(a1);
			\node at ($(-0.85*\ftwidth,-\ftheight-0.2cm)$) {$u$};	
			\node at ($(-0.5*\ftwidth,-\ftheight-0.2cm)$) {$v$};	
			\node at ($(0,-\ftheight-0.2cm)$) {$w$};	
			\node at ($(0.5*\ftwidth,-\ftheight-0.2cm)$) {$x$};	
			\node at ($(0.85*\ftwidth,-\ftheight-0.2cm)$) {$y$};	
			
			\newcommand{\ftheighttwo}{2cm}
			\begin{scope}[xshift=8cm,yshift=-1cm]
				\node at (2.5cm,-1cm) {\large $=\mathcal{T}'$};	
				\path [draw,thin] (0,0) -- (-\ftwidth,-\ftheighttwo) -- (\ftwidth,-\ftheighttwo) -- cycle;
				\path[draw,thick] (0,0)--(\goleft,\godown)--($(-\goleft,2*\godown)$);
				\coordinate (a1) at ($(-\goleft,2*\godown)$);
				\node (n2) at ($(a1)-(3,0)$) {node $\vec n$};
				\path[draw, ->] (a1)--(n2);
				
				\node at ($(-\goleft+0.1cm,2*\godown)$) {$A$};	
				\path[draw,thin] (a1)--($(-0.6*\ftwidth,-\ftheighttwo)$)--($(0.6*\ftwidth,-\ftheighttwo)$)--(a1);
				\node at ($(-0.85*\ftwidth,-\ftheighttwo-0.2cm)$) {$u$};	
				\node at ($(-0.6*\ftwidth,-\ftheighttwo-0.2cm)$) {$(A,n)$};	
				\node at ($(0,-\ftheighttwo-0.2cm)$) {$w$};	
				\node at ($(0.6*\ftwidth,-\ftheighttwo-0.25cm)$) {$\overline{(A,n)}$};	
				\node at ($(0.85*\ftwidth,-\ftheighttwo-0.2cm)$) {$y$};	
			\end{scope}
			
			
			

			\node at (current bounding box.center) {$\leadsto$};
		\end{tikzpicture}
	\end{center}
	\caption{Removal of a single cycle in an arbitrary derivation tree $\cT$ to get a tree $\cT'$ with additional leaves of the form $(A,n)$ and $\overline{(A,n)}$. }\label{fig:insertCycles}
\end{figure}
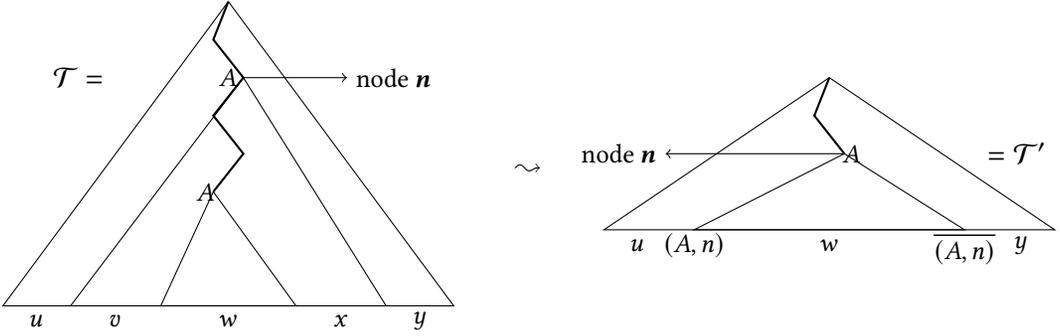

The algorithm guesses an exponential sized tree $\tilde{\cT}$ and verifies the consistency of node labels between parent and children nodes in the tree using the rule set $P$ of the grammar. It then  constructs a formula $\phi_{\tilde{\cT}}$ as follows. 
The formula $\phi_{\tilde{\cT}}$ contains variables for a sequence of fractions and vectors $\mvec{x}_0,\alpha_1,\mvec{x}_1 \ldots \alpha_l, \mvec{x}_l$ where $l$ is the number of leaf nodes in $\tilde{\cT}$. Let $\gamma_i$ be the label of the $i^{th}$ leaf node. The constructed formula is the conjunction of the following formulae $\phi_i$ for each leaf $i$: 
\begin{itemize}
	\item if $\gamma_i$ is a nonpump letter then $\phi_i := (\mvec{x}_i + \alpha_{i+1}\gamma_i=\mvec{x}_{i+1})$, else
	\item  $\gamma_i$ is a pump letter $(A,\mvec n)$, then $\mvec{x}_i,\mvec{x}_{i+1}$ are plugged into an instantiation of the formula obtained from \cref{thm:ratArithCycleFormula} for $A$, along with the corresponding vectors $\mvec{x}_j,\mvec{x}_{j+1}$ for the dual letter $\overline{(A,\mvec n)}=\gamma_j$ to give the formula $\phi_{i,j}$. In this case, $\phi_i=\phi_j=\phi_{i,j}$. 
\end{itemize}
The formula $\phi_{\tilde{\cT}}$ existentially quantifies over the variables $\vec x_1,\ldots, \vec x_{l-1}$ as well as the firing fractions $\alpha_1,\ldots, \alpha_l$, while $\vec x_0$ and $\vec x_l$ are free variables corresponding to $\vec u$ and $\vec v$ respectively. The final formula we want is $\bigvee_{\tilde{\cT}} \phi_{\tilde{\cT}}$.

\subsection{Capturing pump reachability relations} 
\label{ssec:cyclicReachCFG}
It remains to prove \cref{thm:ratArithCycleFormula}.
The key observation is that a characterization of~\citet{blondinLogicsContinuousReachability2017} 
of the existence of ``cyclic runs'' (i.e.\ ones that start and end in the same control state) in $\QVASS$ actually also applies to $\QVASS$ with infinitely many control states.
Thus, the first step is to translate the setting of pumps into that of cyclic $\QVASS$ with infinitely many control states.
It is more convenient for us to use algebraic terminology, so we will phrase this as a translation to the case of semigroups.
We say a language $K \subseteq \Sigma^*$ is a \emph{semigroup} if it is closed under concatenation, 
i.e., for any $u,v\in K$, we have $uv\in K$.
We will show that for our particular cyclic $\QVASS$, the characterization of Blondin and Haase allows us to build an exponential-sized ELRA formula.

\smallskip\noindent\textbf{Reduction to semigroups}
Let us first show that the pump reachability relations $P_A$ can be captured using context-free semigroups.
In this section, we often write letters $a$ in normal font, even though they are vectors in $\integ^d$. Vectors in $\ratplus^d$ are represented by bold font e.g. $\mvec{u}$.

The following lemma uses the idea of simulating grammar derivations by applying ``the left part forward'' and ``the right part backward'', which is a recurring theme in the literature on context-free grammars (see, e.g.~\cite{rosenberg1967machine,Berstel1979,DBLP:journals/pacmpl/BaumannGMTZ23,DBLP:conf/mfcs/LohreyRZ22,DBLP:conf/popl/RepsTP16}) and has been applied to PVASS by \citet[Section~5]{lerouxCoverabilityProblemPushdown2015}.
\begin{lemma}\label{lem:reduction-to-semigroup} 
	Given a grammar $\cG$ over $\Z^d$ and a non-terminal $A$
	in $\cG$, one can compute, in polynomial time, a context-free language
	$K\subseteq(\Z^{2d})^*$ such that (i)~$K$ is a semigroup and (ii)~for
	any $\mvec{u},\mvec{v},\mvec{u}',\mvec{v}'\in\Q^d$, we have
	$(\mvec{u},\mvec{v}')\step{\ratplus}{K} (\mvec{u}',\mvec{v})$ if and
	only if $(\mvec{u},\mvec{v},\mvec{u}',\mvec{v}')\in P_A$.  
\end{lemma} 
Suppose $(S,N,\Sigma,P)$ is a context-free grammar in Chomsky normal form, with $\Sigma \subseteq\Z^d$ and $A\in N$.  
The idea is to take a derivation tree for $A\derivstar wAw'$ and
consider the path from the root to the $A$ in the derived word, see
\cref{subtree-flipping} on the left. We transform the tree as follows. Each
subtree on the left of this path ($\ell_1$ and $\ell_2$ in the figure) is left
unchanged, except that each produced vector $ a\in\Z^d$ is padded so as to
obtain $( a,0,\ldots,0)\in\Z^{2d}$. In the figure, the resulting subtrees are
$\overrightarrow{\ell_1}$, $\overrightarrow{\ell_2}$. Each subtree on the right
($r_1$ and $r_2$ in the figure), however, is moved to the left side of the path
and it is \emph{reversed}, meaning in particular that the word produced by it
is reversed. Moreover, each vector $b\in\Z^d$ occurring at a leaf is turned
into $(0,\ldots,0,-b)\in\Z^{2d}$.

Then, every word generated by the new grammar is of
the form $\overrightarrow{x_1}\overleftarrow{y_n}\cdots
\overrightarrow{x_n}\overleftarrow{y_1}$, where $x_1\cdots x_n y_1\cdots y_n$ is the word produced by the original grammar.
Here, for a word $w\in(\Z^d)^*$, $\overrightarrow{w}$ is obtained from $w$ by
replacing each vector $ a \in\Z^d$ in $w$ by $( a,0,\ldots,0)\in\Z^{2d}$, and 
$\overleftarrow{w}$ is obtained from $w$ by reversing the word and replacing
each $ a\in\Z^d$ by $(0,\ldots,0,-a)\in\Z^{2d}$. Conversely, for every $A\derivstar
wAw'$, we can find a word $\overrightarrow{x_1}\overleftarrow{y_n}\cdots
\overrightarrow{x_n}\overleftarrow{y_1}$ in the new grammar such that
$x_1\cdots x_n=w$ and $y_1\cdots y_n=w'$. Thus, we clearly have
$(\mvec{u},\mvec{v}')\step{\ratplus}{K} (\mvec{u}',\mvec{v})$ if and only
if $(\mvec{u},\mvec{v},\mvec{u}',\mvec{v}')\in P_A$ for every $\mvec u,\mvec
v\in\ratplus^d$. 

Formally, in the new grammar for $K$, we have three copies of the non-terminals
in $N$, hence we have $N'=\overleftarrow{N}\cup\overrightarrow{N}\cup\hat{N}$,
where $\overleftarrow{N}=\{\overleftarrow{B}\mid B\in N\}$,
$\overrightarrow{N}=\{\overrightarrow{B}\mid B\in N\}$ and
$\hat{N}=\{\hat{B}\mid B\in N\}$ are disjoint copies of $N$. The productions in
the new grammar are as follows: For every production $B\to CD$ in $P$, we
include productions $\overrightarrow{B}\to
\overrightarrow{C}\overrightarrow{D}$, $\hat{B}\to \overrightarrow{C}\hat{D}$
and $\hat{B}\to\overleftarrow{D}\hat{C}$,
$\overleftarrow{B}\to\overleftarrow{D}\overleftarrow{C}$. Moreover, for every
$B\to  b$ with $ b\in\Z^d$, we include the productions $\overrightarrow{B}\to
( b,0,\ldots,0)\in \Z^{2d}$ and $\overleftarrow{B}\to
(0,\ldots,0,- b)\in\Z^{2d}$. Finally, we add $\hat{A}\to\varepsilon$ and set the start symbol to $\hat{A}$. Observe that the generated language is closed under concatenation. This is because at any point in any derivation, there is exactly one hat nonterminal symbol which always occurs as the last symbol and only $\hat{A}$ can be replaced by $\varepsilon$. This means whenever $\hat{A} \derivstar u$, it is the case that $\hat{A} \derivstar u \hat{A}$ and hence if $\hat{A} \derivstar v$ as well, then it is the case that $\hat{A} \derivstar uv$ as well as $\hat{A} \derivstar vu$.
This grammar achieves the required transformation.

	\begin{figure}[t]
		\begin{center}
			\begin{tikzpicture}[level distance=0.7cm, sibling distance=0.8cm]
				\node [nonterminal] (root original) {} 
				child[npath] { node[ltree] (l1) {}
				}
				child[epath] { node[nonterminal] {}
					child[epath] { node[nonterminal] {}
						child[npath] { node[ltree] (l2) {} 
						} 
						child[epath] { node[nonterminal] {}
							child[epath] { node[nonterminal] (last original) {}
							}
							child[npath] { node[rtree] (r2) {}
							}
						}
					}
					child[npath] { node[rtree] (r1) {}
					}
				};
				\node[right=0.05cm of root original] {$A$};	
				\node[right=0.05cm of last original] {$A$};
				
				\node at ($(l1)+(-0.3*\twidth,-0.7*\theight)$) {$\ell_1$};
				\node at ($(l2)+(-0.3*\twidth,-0.7*\theight)$) {$\ell_2$};
				\node at ($(r1)+(0.3*\twidth,-0.7*\theight)$) {$r_1$};
				\node at ($(r2)+(0.3*\twidth,-0.7*\theight)$) {$r_2$};
				
				\node[right=5cm of root original, nonterminal] (root new) {} 
				child[npath] { node[ltree] (nl1) {}
				}
				child[epath] { node[nonterminal] {}
					child[npath] { node[ltree] (nl2) {}
					}
					child[epath] { node[nonterminal] {}
						child[npath] { node[ltree] (nl3) {} 
						} 
						child[epath] { node[nonterminal] {}
							child[npath] { node[ltree] (nl4) {}
							}
							child[epath] { node[nonterminal] (last new) {}
							}
						}
					}
				};
				
				\node[right=0.05cm of root new] {$\hat{A}$};
				\node[right=0.05cm of last new] {$\hat{A}$};
				\node at ($(nl1)+(-0.3*\twidth,-0.7*\theight)$) {$\overrightarrow{\ell_1}$};
				\node at ($(nl2)+(-0.3*\twidth,-0.7*\theight)$) {$\overleftarrow{r_1}$};
				\node at ($(nl3)+(-0.3*\twidth,-0.7*\theight)$) {$\overrightarrow{\ell_2}$};
				\node at ($(nl4)+(-0.3*\twidth,-0.7*\theight)$) {$\overleftarrow{r_2}$};

				\node at (current bounding box.center) {$\leadsto$};
			\end{tikzpicture}
		\end{center}
		\vspace*{-0.1cm}
		\caption{Illustration of \cref{lem:reduction-to-semigroup}. A derivation of the original grammar (shown on the left) is transformed into a derivation of the new grammar (on the right).}\label{subtree-flipping}
		\vspace*{-0.4cm}
	\end{figure}
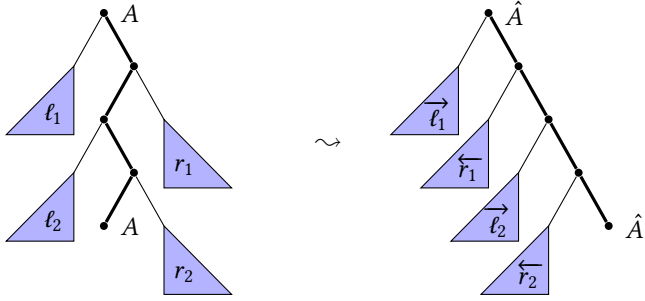

	\smallskip
	\noindent\textbf{Reduction to letter-uniform semigroup} 
	As a second step, we will further reduce the problem to the case where 
	n all runs, the letters (i.e.\ added vectors) appear uniformly (in some precise sense). The \emph{support sequence} of a word $w\in\Sigma$ is the tuple $(\Gamma,<)$ where $\Gamma \subseteq \Sigma$ is the subset of letters occuring in $w$ and $<$ is a total order on $\Gamma$ which corresponds to the order of first occurrence of the letters in $w$. For example the support sequence of $aacabbc$ consists of $\Gamma =\{ a,b,c\}$ and the linear ordering $a < c < b$.
	
	A context-free
	language $K\subseteq(\Z^d)^*$ is \emph{letter-uniform} if any two
	words in $K$ have the same support sequence. Let $\Sigma\subseteq\Z^d$
	be the set of letters occurring in $K$. Moreover, for every subset
	$\Gamma\subseteq\Sigma$ and a total order $<$ on $\Gamma=\{ \gamma_1, \ldots, \gamma_l\}$ given as $\gamma_1 < \gamma_2 \ldots < \gamma_l$, let $K_{(\Gamma,<)}=\{w \in K \mid \exists u_1,u_2,\ldots, u_l \; w=\gamma_1 u_1 \gamma_2 u_2 \ldots \gamma_l u_l \text{ where } u_i \in \{ \gamma_1,\ldots, \gamma_i\}^* \}$ denote the set of all words in $K$ with support sequence $(\Gamma,<)$.
	
	Then we can observe that each $K_{(\Gamma,<)}$ is
	letter-uniform and also a semigroup: for any two words $u,v \in K_{(\Gamma,<)}$, it is the case that the letters occurring in $uv$ and $vu$ are exactly $\Gamma$ and furthermore, the order of first occurrence of the letters from $\Gamma$ in the two words also corresponds to the total order $<$. Furthermore, $uv,vu \in K$ since $K$ is a semigroup. Hence both of these words also belong to $K_{(\Gamma,<)}$.
	Moreover, we have $\mvec
	u\step{\ratplus}{K}\mvec v$ if and only if there exists some 
	$\Gamma\subseteq\Sigma$ and total order $<$ on $\Gamma$ with $\mvec u\step{\ratplus}{K_{(\Gamma,<)}}\mvec v$.
	We shall prove the following:
	\begin{proposition}\label{letter-uniform-semigroup}
		Given a context-free letter-uniform semigroup $K\subseteq(\Z^d)^*$,
		we can in $\NEXPTIME$ construct an ELRA formula for the relation $R_K$.
	\end{proposition}
	Let us see how \cref{thm:ratArithCycleFormula} follows from \cref{letter-uniform-semigroup}.
	Given some nonterminal $A$ of a CFG, we want a formula for $P_A$. We first use \cref{lem:reduction-to-semigroup} to compute a context-free language $K$ such that $R_K$ captures $P_A$ (up to permuting some counters).
	Suppose $K\subseteq\Sigma^*$ for some $\Sigma\subseteq\Z^d$. For each subset
	$\Gamma\subseteq\Sigma$ and total order $<$ on $\Gamma$, consider the set $K_{(\Gamma,<)}$ as defined earlier. As we already observed,  (i)~each $K_{(\Gamma,<)}$ is a semigroup,
	(ii)~each $K_{(\Gamma,<)}$ is letter-uniform, and (iii)~$K$ is the union of all
	$K_{(\Gamma,<)}$. 
	Therefore, our construction proceeds as follows. We guess
	${(\Gamma,<)}$ and then apply \cref{letter-uniform-semigroup} to
	compute in $\NEXPTIME$ an existential $\FO(\rat,+,<)$ formula for $R_{K_{(\Gamma,<)}}$.
	Then, the disjunction of all resulting formulas clearly defines $R_K$. We note that given some $(\Gamma,<)$, a grammar for $K_{(\Gamma,<)}$ can be constructed from the grammar for $K$ in polynomial time. This is because we need to construct a grammar for the intersection of $K$ with the language of the regular expression given by $R := \gamma_1 \gamma_1^* \gamma_2 (\gamma_1+\gamma_2)^*\gamma_3\ldots\gamma_k(\gamma_1+\gamma_2\ldots \gamma_k)^*$, which only incurs a polynomial blowup. In fact, without the linear order and only the subset $\Gamma$, the same construction would lead to an exponential blowup since we would then have to remember all possible subsets of $\Gamma$ while reading a word.  
	
	\smallskip
	\noindent\textbf{Characterizing reachability by three runs}
	It remains to show \cref{letter-uniform-semigroup}. 
	The advantage of reducing our problem to the letter-uniform case is
	that we can employ a characterization of \citet{blondinLogicsContinuousReachability2017}
	about the existence of runs. 
	In the rest of this section, we will assume that the language $K$ comes with a corresponding support sequence $(\Gamma,<)$.
	
	The following \lcnamecref{three-runs} tells us that reachability along a
	letter-uniform semigroup can be characterized by the existence of three
	runs: One run that witnesses reachability under $\Q$-semantics, and two runs
	that witness ``admissibility'' in both directions. Here, the ``only if'' direction
	is trivial, because the run from $\mvec u$ to $\mvec v$ along $K$ is a run of
	all three types. For the converse, we use the fact that $K$ is a semigroup and
	letter-uniform to compose the three runs into a run under
	$\ratplus$-semantics.
	\begin{lemma}\label{three-runs}
		Let $K\subseteq(\Z^d)^*$ be a letter-uniform semigroup. Then we have $\mvec{u}\step{\ratplus}{K}\mvec{v}$ if and only if there are $\mvec{u}',\mvec{v}'\in\Q^d$ such that: 
		\begin{align*}
			\mvec{u}\step{\rat}{K}\mvec{v}, &&
			\mvec{u}\step{\ratplus}{K}\mvec{v}', &&
			\text{and} &&
			\mvec{u}'\step{\ratplus}{K}\mvec{v}.
		\end{align*}
	\end{lemma}
	
	\Cref{three-runs} is an extension of \cite[Proposition 4.5]{blondinLogicsContinuousReachability2017}. The only difference is that in \cite{blondinLogicsContinuousReachability2017}, $K$ is given by a non-deterministic finite automaton where one state is both the initial and the final state.
	The ``only if'' direction is trivial. 
	For the ``if'' direction, the proof in \cite{blondinLogicsContinuousReachability2017} takes the three runs and shows that a suitable concatenation of these runs, together with an appropriate choice of multiplicities, yields the desired run $\mvec u\step{\ratplus}{K}\mvec v$.
	Since $K$ is a semigroup, the same argument yields \cref{three-runs}.
	See Subsection~\ref{sec:appendix-upper-bound-three-runs} of the appendix for details.
	
	\Cref{three-runs} allows us to express the reachability relation along $K$:
	It tells us that we merely have to express existence of the three simpler types of runs. The first of the three runs is 
	reachability under $\rat$-semantics while the second and third are examples of \emph{admissible} runs under $\ratplus$-semantics. Thus we need to characterize these two types of runs.
	We will do this in the following two subsections.

	\smallskip
	\noindent\textbf{Characterizing reachability under $\rat$-semantics}
	We first show how to construct an ELRA formula for the $\rat$-reachability relation along a letter-uniform context-free $K$. 
	\begin{lemma}\label{q-reachability}
		Given a letter-uniform context-free language $K\subseteq(\Z^d)^*$, we can construct in exponential time an ELRA formula for the relation
		\[ R_K^{\Q} = \{(\mvec u,\mvec v)\in\Q_+^d\times\Q_+^d \mid \mvec u\step{\rat}{K} \mvec v\}.\]
	\end{lemma}
	Our proof relies on the following, which was shown in \cite[Proposition~B.4]{blondinLogicsContinuousReachability2017}:
	\begin{lemma}[\citet{blondinLogicsContinuousReachability2017}]\label{q-reachability-vass}
		Given an NFA
		$\cA$ over some alphabet $\Sigma\subseteq\Z^d$, one can in polynomial time construct an ELRA formula $\varphi$ such that for $\mvec u,\mvec v\in\ratplus^d$, we have $\varphi(\mvec u,\mvec v)$ if and only if $\mvec u\step{\rat}{L(\cA)}\mvec v$.
	\end{lemma}
	
	\begin{proof}[Proof of Lemma~\ref{q-reachability}]
		The key observation is that in the case of $\rat$-semantics, reachability along
		a word $w\in(\Z^d)^*$ does not depend on the exact order of the letters in $w$. Let $\Parikh(w) \in \mathbb{N}^{|\Sigma|}$ be the \emph{Parikh image} of $w$ i.e. $\Parikh(w)(a)$ for $a \in \Sigma$ denotes the number of times $a$ occurs in $w$.
		Formally, if $\Parikh(w)=\Parikh(w')$, then $\mvec u\step{\rat}{w}\mvec v$
		if and only if $\mvec u\step{\rat}{w'}\mvec v$. In particular, for
		languages $K,K'\subseteq(\Z^d)^*$, if $\Parikh(K)=\Parikh(K')$, then
		$R_{K}^{\rat}=R_{K'}^{\rat}$. We use this to reduce the case of context-free
		$K$ to the case of regular languages $K$.
		
		It is well known that given a context-free grammar, one can construct an NFA of
		exponential size such that the NFA accepts a language of the same Parikh image
		as the grammar. For example, a simple construction with a close-to-tight size
		bound can be found in~\cite{DBLP:journals/ipl/EsparzaGKL11}. Therefore, given
		$K$, we can construct an exponential-sized NFA $\cA$ such that
		$\Parikh(L(\cA))=\Parikh(K)$. 
		
		
		
		Observe that $\Parikh(L(\cA))=\Parikh(K)$ implies that for any $\mvec u,\mvec
		v\in\ratplus^d$, we have $\mvec u\step{\rat}{K}\mvec v$ if and only if
		$\mvec u\step{\rat}{L(\cA)}\mvec v$. 
		Therefore, we apply
		\cref{q-reachability-vass} to compute a formula $\varphi$ from $\cA$.
		Since $\cA$ is exponential in size, this computation takes exponential time
		and results in an exponential formula $\varphi$. Then, for $\mvec u,\mvec
		v\in\ratplus^d$, we have $\varphi(\mvec u,\mvec v)$ if and only if $\mvec
		u\step{\rat}{L(\cA)}\mvec v$, which is equivalent to $\mvec
		u\step{\rat}{K} \mvec v$.  
	\end{proof}
	
	We note that it is also possible to use a construction of~\cite{VermaSS05} to construct a formula for $R_K^{\Q}$ in polynomial time -- the reason we chose the current presentation is that applying the result from~\cite{VermaSS05} as a black box would result in a formula over \emph{mixed} integer-rational arithmetic (of polynomial size): 
	One would use integer variables to implement the construction from~\cite{VermaSS05} and then rational variables to account for continuous semantics. This would yield the same complexity bound in the end (existential mixed linear arithmetic is still in NP), but we preferred not to introduce another logic.

	\smallskip
	\noindent\textbf{Characterizing admissibility under $\ratplus$-semantics}
	Finally, we construct an ELRA formula for the set of vectors $\mvec u$ such that there exists a $\mvec v'\in\ratplus^d$ with $\mvec u\step{\ratplus}{K}\mvec v'$. 
	We call such vectors \emph{$K$-admissible} and denote the set of $K$-admissible vectors as $A_K$.
	The key observation is that $\mvec u$ is $K$-admissible if and only if the total order $<$ satisfies some simple properties.
	Intuitively, $\mvec u$ is \emph{$<$-admissible} if for each letter that decrements a counter, either (i)~that counter is positive in $\mvec u$ or (ii)~there is an earlier letter that increments this counter.
	For $ a \in \mathbb{Z}^d$, we denote by $\supp{a}$ (resp. $\supp{a}^+$ or $\supp{a}^-$) the subset of indices $i$ where $a(i) \neq 0$ (resp. $a(i) > 0$ or $a(i) < 0$).  
	Formally, $\mvec u$ is $<$-admissible if for each $a\in\Gamma$ and each $j\in\supp{a}^-$, we either have (i)~$\mvec u(j)>0$ or (ii)~there is a $b\in\Gamma$ with $b<a$ and $j\in \supp{b}^+$.
	We show the following:
	\begin{lemma}\label{lem:coverKcharacter}
		Let $K$ be letter-uniform and $K \neq \emptyset$.  Then $\mvec u\in A_K$ if and only if $\mvec u$ is $<$-admissible.
	\end{lemma}
	The ``if'' direction is easy. If $\vec u$ is not $<-admissible$, this means that there is some index $j$ and some letter $\gamma_i$ such that $\gamma_i$ decrements $j$, $\vec u(j)=0$ and $\gamma_k(j)\leq 0$ for all $k < i$. This means that starting from $\vec u$, on any word $w \in K$, we would go below $0$ in index $j$ when $\gamma_i$ first occurs in $w$. Hence $\vec u \not\in K$. 
	
	For the converse, suppose that $\mvec u$ is $<$-admissible and $w\in K$ be any word. Write $w=w_1\cdots w_n$ with $w_1,\ldots,w_n\in\Gamma$. For each $i\in\{0,\ldots,n\}$, we define
	\[ I_i = \supp{\vec u}^+\cup\bigcup_{\ell=1}^i\supp{w_\ell}^+,\]
	i.e.\ the set of components that are incremented at some point when firing the prefix $w_1\cdots w_i$.
	
	We will show the following for every $i$: There exists a $\mvec v_i$ such that $\mvec u\step{\ratplus}{w_1\cdots w_i}\mvec v_i$ such that $\mvec v_i \in \ratplus$ and $\mvec v_i(j) >0$ for $j \in I_i$.
	
	We proceed by induction on $i$. For $i=0$, the statement clearly holds, because $\mvec u$ is positive on all co-ordinates in $I_0=\supp{\mvec u}^+$ and $\mvec u \in \ratplus$. Now suppose there is a run $\mvec u\step{\ratplus}{w_1\cdots w_i}\mvec v_i$. We know that $\supp{w_{i+1}}^-\subseteq I_i$ and thus $\mvec v_i$ is positive on all indices where $w_{i+1}$ is negative. Let $\alpha_{i+1}= \min \{ -\frac{\mvec v_i(j)}{2w_{i+1}(j)} \mid j \in \supp{w_{i+1}}^-\}$. Clearly $\mvec v_i + \alpha_{i+1}w_{i+1}= \mvec v_{i+1} \in \ratplus$ and also,  $\mvec v_{i+1}(j)>0$ for all $j\in I_{i+1}$. 
	In particular this means that $\mvec u\step{\ratplus}{w}\mvec v_n$ and thus $\mvec u\in A_K$.

	\begin{lemma}\label{lem:constructCoverKFormulaInNP}
		Given a non-empty letter-uniform context-free language $K\subseteq(\Z^d)^*$, we can construct in polytime
		an ELRA formula for the relation $A_K$.
	\end{lemma}
	The formula $\phi_<$, where $\phi_<(\bu)$ for a vector $\bu \in \ratplus^d$ is true
	iff $\bu$ is $<$-admissible can be written as:
	$$\phi_<(\bx) = \bigwedge_{\gamma \in \Gamma} \bigwedge_{i \in \supp{\gamma}^-} \bx(i) > 0 \lor \bigvee_{\eta < \gamma} i \in \supp{\eta}^+$$

	\begin{proof}[Proof of \cref{letter-uniform-semigroup}]
		We are now ready to prove \cref{letter-uniform-semigroup}.
		By \cref{three-runs}, it suffices to show that there are formulae for the relations $R_K^{\Q}$ and $A_K$. These formulae have been obtained in \cref{q-reachability} and \cref{lem:constructCoverKFormulaInNP} respectively. 
	\end{proof}
	This concludes the proof that reachability in $\QPVASS$ is in $\NEXPTIME$.  
	
	\smallskip
	\noindent\textbf{State reachability} 
	The material in this section also allows us to derive \cref{main-result-state-reachability}. 
	Since state-reachability is $\NP$-hard already for $\QVASS$~\cite{blondinLogicsContinuousReachability2017}, we only have to show membership in $\NP$. 
	Using the language-theoretic translation from the beginning of this section, state reachability
	can be phrased as: Given a set of vectors $\Sigma \subseteq \Z^d$, a letter-uniform context-free language $K \subseteq \Sigma^*$ (which comes with an associated $(\Gamma,<)$) and $\bu \in \ratplus^d$,
	decide if $\bu$ is $K$-admissible. This is because, given an arbitrary context-free language $K$, we can see it as the disjoint union of (exponentially) many $K_{(\Gamma,<)}$ each of which is 
	letter-uniform. 
	Furthermore, we can construct each $K_{(\Gamma,<)}$ in polynomial time and check if they are non-empty. The $\NP$ upper bound then follows from \cref{lem:constructCoverKFormulaInNP}:
	We can guess $(\Gamma,<)$, construct $K_{(\Gamma,<)}$ and  \cref{lem:constructCoverKFormulaInNP} lets us build an ELRA formula for $A_{K_{(\Gamma,<)}}$. 
	Since the truth problem for ELRA is in $\NP$~\cite{sontagRealAdditionPolynomial1985}, we can then check whether $\bu$ satisfies the constructed formula for $A_K$.

\section{$\NEXPTIME$-hardness of $\MCM$}\label{sec:lower-bound-mcm-hardness}
We now move on to proving the $\NEXPTIME$-hardness of reachability in $\QPVASS$. 
As outlined in the introduction, we do this by a chain of reductions. Our reduction chain starts with the machine model $\MCM$.
Informally, a $\MCM$ has a finite-state control along with two counters, each of which can hold a non-negative integer.
A rule of the machine allows us to move from one state to another whilst either incrementing the value of a counter by 1 or doubling the value of a counter. 
The set of final configurations of such a machine will be given by a final state and \emph{an equality condition on the two counters}.

Formally, a $\MCM$ is a tuple $\mach = (Q,\init,\fin,\Delta)$ where $Q$ is a finite set of states, $\init, \fin \in Q$ are the initial and the final states respectively,
and $\Delta \subseteq Q \times \{\inc_0,\inc_1,\double_0, \double_1, \nop \} \times Q$ is a finite set of rules. A configuration of $\mathcal{M}$ is a triple
$(q,v_0,v_1)$ where $q \in Q$ is the current state of $\mach$ and $v_0, v_1 \in \nn$ are the current values of the two counters respectively. Let
$r = (q,t,q') \in \Delta$ be a rule.
A \emph{step} from a configuration $C = (p,v_0,v_1)$ to another configuration $C' = (p',v_0',v_1')$ by means of the rule $r$ (denoted by $C \act{r} C'$) is possible if and only if $p = q, p' = q',$ and 
	$$\text{If } t = \inc_i \text{ then } v'_{i} = v_{i}+1,  \ v'_{1-i} = v_{1-i} \qquad  \text{ If } t = \double_i \text{ then } v'_i = 2v_i, \ v'_{1-i} = v_{1-i}$$
	$$\text{If } t = \nop \text{ then } v'_0 = v_0, \ v'_1 = v_1$$

We then say that a configuration $C$ can reach another configuration $C'$ 
if $C'$ can be reached from $C$ by a sequence of steps.
The initial configuration of $\mach$ is $C_{init} := (\init,0,0)$. The set of \emph{final configurations} of $\mach$ is taken to be $\{(\fin,n,n) : n \in \nn\}$.

The reachability problem for $\MCM$ asks, given a $\MCM$ $\mach$ and a number $m$ \emph{in binary}, if the initial configuration can reach \emph{some final configuration} in exactly $m$ steps. 

\begin{example}
	Let us consider the $\MCM$ given in Figure~\ref{fig:example-mcm}, which we shall denote by $\mach$. The initial state is $q_0$ and the final state is $q_2$. 
	Note that the initial configuration $(q_0,0,0)$ can reach $(q_2,1,1)$ in exactly 2 steps. Hence if we set the length 
	of the run $m$ to 2, then the instance $\langle \mach, m \rangle$
	is a positive instance of the reachability problem for $\MCM$. On the other hand, for any other value of $m$, the instance
	$\langle \mach, m \rangle$ is a negative instance of the reachability problem for $\MCM$. 
	Indeed, first note that in this $\MCM$, each state has exactly one outgoing transition. Hence, there is exactly one run
	starting from $(q_0,0,0)$ and that run is as follows: First it reaches the configuration $(q_2,1,1)$ in exactly 2 steps.
	Then from there, it follows the following  (cyclical) pattern.
	\begin{multline*}
		(q_2,x,y) \act{} (q_3,2x,y) \act{} (q_4,2x+1,y) \act{} (q_5,2x+1,2y) \act{} (q_2,2x+1,2y) \act{} (q_3,4x+2,2y) \\
		(q_4,4x+3,2y) \act{} (q_5,4x+3,4y) \act{} (q_2,4x+3,4y) \dots
	\end{multline*}

	This pattern indicates that after the configuration $(q_2,1,1)$, 
	whenever the run reaches the state $q_2$, the first counter has an odd value, whereas the second counter
	has an even value. Hence, the run will never reach a final configuration and so $\langle \mach, m \rangle$ is a negative 
	instance of the reachability problem whenever $m \neq 2$.
\end{example}

\begin{figure}
	\begin{center}
		\tikzstyle{node}=[circle,draw=black,thick,minimum size=12mm,inner sep=0.75mm,font=\normalsize]
		\tikzstyle{edgelabelabove}=[sloped, above, align= center]
		\tikzstyle{edgelabelbelow}=[sloped, below, align= center]
		\begin{tikzpicture}[->,node distance = 1cm,scale=0.8, every node/.style={scale=0.8}]
			\node[node, initial, initial text = \text{}] (q0) {$q_0$};
			\node[node, right = of q0] (q1) {$q_1$};
			\node[node, accepting,right = of q1] (q2) {$q_2$};
			\node[node, right = of q2] (q3) {$q_3$};
			\node[node, right = of q3] (q4) {$q_4$};
			\node[node, right = of q4] (q5) {$q_5$};
			
			\draw(q0) edge[edgelabelabove] node[above]{$\inc_0$} (q1);
			\draw(q1) edge[edgelabelabove] node[above]{$\inc_1$} (q2);
			\draw(q2) edge[edgelabelabove] node[above]{$\double_0$} (q3);
			\draw(q3) edge[edgelabelabove] node[above]{$\inc_0$} (q4);
			\draw(q4) edge[edgelabelabove] node[above]{$\double_1$} (q5);
			\draw(q5) edge[edgelabelabove, bend left] node[above]{$\nop$} (q2);
		\end{tikzpicture}
	\end{center}
	\caption{An example $\MCM$. The initial state is $q_0$ and the final state is $q_2$.}
	\label{fig:example-mcm}
\end{figure}
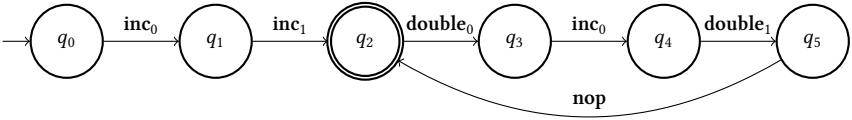

%
We shall prove the following:
\begin{restatable}{theorem}{mcm}\label{thm:mcm}
	The reachability problem for $\MCM$ is $\NEXPTIME$-hard.
\end{restatable}
\cref{thm:mcm} is shown using a bounded version of the classical Post Correspondence Problem (PCP). 
Recall that, in the PCP problem, we are given a set of pairs of words $(u_1,v_1),(u_2,v_2),\dots,(u_m,v_m)$ over a common alphabet $\Sigma$ and
we are asked to decide if there is a sequence of indices $i_1, i_2, \dots, i_k$ for some $k$ such that $u_{i_1} \cdot u_{i_2} \cdot \cdots \cdot u_{i_k} = v_{i_1} \cdot v_{i_2} \cdot \cdots \cdot v_{i_k}$. It is well-known that this problem is undecidable~\cite{sipserIntroductionTheoryComputation2012}.
For our purposes, we shall use a bounded version of PCP, called bounded PCP, defined as follows.
\begin{description}
	\item[Input:] A set of pairs of words $(u_1,v_1),(u_2,v_2),\dots,(u_m,v_m)$ over an alphabet $\Sigma$ such that none of the given words is the empty string, and a number $\ell$ encoded in binary.
	\item[Question:] Is there a sequence of indices $i_1, i_2, \dots, i_k$ such that $u_{i_1} \cdot u_{i_2} \cdot \cdots \cdot u_{i_k} = v_{i_1} \cdot v_{i_2} \cdot \cdots \cdot v_{i_k}$, and the length of $u_{i_1} \cdots \cdot u_{i_k}$ is exactly $\ell$. 
\end{description}

Note that this problem is decidable -- we simply have to guess a sequence of indices of length at most $\ell$ and check that
the resulting words from these indices satisfy the given property. In~\cite[Section 6.1]{AiswaryaMS22}, Bounded-PCP was shown to be $\NEXPTIME$-hard.
We now prove~\cref{thm:mcm} by giving a reduction from Bounded-PCP to the reachability problem for $\MCM$.

Let $(u_1,v_1),\dots,(u_m,v_m)$ be a set of pairs of words over a common
alphabet $\Sigma$ and let $\ell \in \nn$. Without loss of generality we assume that $|\Sigma| = 2^k$ for some $k \ge 1$, for instance, by adding at most twice as many dummy letters
as the size of the alphabet. 
With this assumption, there are two essential ideas behind this reduction, which we now briefly outline.

The first idea is as follows: Since the size of $\Sigma$ is $2^k$, we can identify $\Sigma$ with the set $\{0,1,\dots,2^k-1\}$, by mapping
each letter in $\Sigma$ to some unique number in $\{0,1,\dots,2^k-1\}$. 
This identification means that, any non-empty word $w$
represents a number $n$ in base $|\Sigma|$ in the most significant bit notation.  
In this way, to any number $n$ we can bijectively map a non-empty word $w$. 

The second idea is as follows: Assume that we have a word $w$ and its corresponding number $n$.
Suppose we are given another word $w'$ and we are asked to compute the number corresponding to the concatenated word $w \cdot w'$. 
We can do that as follows: Let $w' = w'_1,\dots,w'_j$ with each $w'_i$ being a letter. 
Construct the sequence of numbers $n_0, n_1, \dots, n_j$ given by $n_0 = n$ and $n_i = |\Sigma| \cdot n_{i-1} + w'_i$. 
Notice that each $n_i$ is essentially the representation of the string $w \cdot w'_1 \cdot w'_2 \cdots w'_i$ and so $n_j$ is the 
representation for the word $w \cdot w'$.

These two ideas essentially illustrate the reduction from the Bounded-PCP problem to the reachability problem for $\MCM$.  Given a Bounded-PCP instance $\langle (u_1,v_1),\dots(u_m,v_m), \ell \rangle$,  we construct a $\MCM$ as follows: 
Initially it starts at an initial state $\init$ with both of its counters set to 0. 
From here it executes a loop in the following manner: Suppose at some point, the machine is at state $\init$ with counter values $n_1$ and $n_2$ \rst{Overloads the strings of the PCP.} \bala{thanks for noticing this. changed now.}
corresponding to some strings $w_1$ and $w_2$ respectively. Then the machine picks some index between 1 and $k$ and then by the idea given in the previous paragraph, it updates the values of its counters to $n_1'$ and $n_2'$ corresponding to the strings $w_1 \cdot u_i$ and $w_2 \cdot v_i$, respectively and then comes back to the state $\init$. 

We can hard-code the rules in this machine so that whenever it has the representation for two strings $w, w'$ in its counters and it wants to 
compute the representation for $w \cdot u_i$ and $w' \cdot v_i$ for some $1 \le i \le k$, it takes \emph{exactly} $t$ steps for some $t$ which is polynomial in the size of the given Bounded-PCP instance.
Then clearly, reaching a configuration $(\init,z,z)$ for some number $z$ in the machine in exactly $t \ell$ steps is equivalent
to finding a sequence of indices $i_1,\dots,i_k$ such that $u_{i_1} \cdots u_{i_k} = v_{i_1} \cdots v_{i_k}$ and the length of $u_{i_1} \cdots u_{i_k}$ is exactly $\ell$. This completes the reduction.

\balain[As we had discussed, added proof for bounded MPCP]


\section{From $\MCM$ to $\NCCM$}\label{sec:lower-bound-mcm-nccm}
The next step in our reduction chain moves from $\MCM$ to $\NCCM$. 
Intuitively, a $\NCCM$ has a finite-state control along with some number of \emph{continuous counters}, each of which can only hold \emph{a fractional number} belonging to the interval $[0,1]$. A rule of such a machine allows us to move from one state to another whilst incrementing or decrementing some counters by \emph{some fractional number}.
Further a rule can also specify that \emph{the effect of firing that rule} makes some counters 0, thereby allowing us to perform zero-tests. 
Note that a $\NCCM$ is different from $\CVASS$ in two aspects: First, the counters of a $\NCCM$ can only hold numbers in $[0,1]$, whereas the counters of a $\CVASS$ can hold any rational number. Second, the counters of a $\NCCM$ can be tested for zero, which is not possible in a $\CVASS$.
We now proceed to formally define the model of a $\NCCM$.

More formally, a $d$-dimensional $\NCCM$ (or $d$-$\NCCM$ or simply $\NCCM$) is a tuple $\cC = (Q, T, \Delta)$ where $Q$ is a finite state of states, $T \subseteq \zn^d \times 2^{[d]}$ is a finite set of \emph{transitions} and $\Delta \subseteq Q \times T \times Q$ is a finite set of \emph{rules}. A configuration of $\cC$ is a tuple $C = (q,\bv)$ where $q \in Q$ is the current state of $\cC$ and $\bv \in [0,1]^d$
is the vector representing the current values of the counters of $\cC$. 
We use the notations $\sta{C}, \valu{C}, C(i)$ to denote $q, \bv, \bv(i)$,
respectively.
Let $I = (q,t,q') \in \Delta$ be a rule with $t = (r,s)$ and let $\alpha \in (0,1]$. A \emph{step} from a configuration $C$ to another
configuration $C'$ via the pair $(\alpha,I)$ (denoted by $C \act{\alpha I} C')$ is possible if and only if $\sta{C} = q, \sta{C'} = q'$ and
$$\valu{C'} = \valu{C} + \alpha r \quad  \text{ and }  \quad \valu{C'}(i) = 0 \text{ for all } i \in s $$ 

Note that we implicitly require that $\valu{C} + \alpha r \in [0,1]^d$ and also that the value obtained after firing $\alpha I$ is 0 on all the counters in the set $s$. 
We define the notions of firing sequences $\alpha_1 I_1, \dots, \alpha_n I_n$ and reachability between configurations $C \act{\alpha_1 I_1, \dots, \alpha_n I_n}~C'$
as for $\CVASS$, 
The reachability problem for $\NCCM$ asks, given a $\NCCM$ $\cC$, two configurations $c_{\mathit{init}}, c_{\mathit{fin}}$ and a number $m$ encoded \emph{in binary}, 
whether $c_{\mathit{init}}$ can reach $c_{\mathit{fin}}$ in exactly $m$ steps. We show that
\begin{theorem}
	The reachability problem for $\NCCM$ is $\NEXPTIME$-hard.
\end{theorem}

We prove this theorem by exhibiting a reduction from the reachability problem for $\MCM$.
Fix a $\MCM$ $\mach$ and a number $m$ in binary. 
Since the initial values of both the counters are 0, the largest value we can attain in any counter during a run of length $m$ is at most $2^m$ (in fact, the bound is $2^{m-1}$).
Hence, we shall implicitly assume that 
the set of configurations of $\mach$ that are under consideration are those where the counter values are bounded by $2^m$. 

\smallskip\noindent\textbf{Overview of the reduction.} We want to construct an $\NCCM$ $\cC$ that simulates $\mach$.
As already mentioned in the introduction, we use \emph{exponential precision} and represent a discrete counter value $n$ in a configuration of $\mach$ as the value $\frac{n}{2^m}$ in a 
continuous counter of $\cC$. 
Furthermore, we want to correctly simulate increment and doubling operations on $\mach$ which correspond to addition of $\frac{1}{2^m}$ and doubling in $\cC$ respectively. 
Since we do not control the fraction $\alpha$ in a rule, we have to overcome the following challenge:
\begin{itemize}
	\item[(C1)] How can we create gadgets which simulate addition of $\frac{1}{2^m}$ and doubling?
\end{itemize}

Towards solving this challenge, we use the following idea: Suppose we are in some configuration $C$ and suppose we want to make a step from $C$ by adding $\frac{1}{2^m}$ to a counter $c$. Assume that there are two other counters $st$ and $te$ whose values in $C$ are $\frac{1}{2^m}$ and $0$, respectively.
Suppose $I$ is a rule which decrements $st$ by 1, increments $c$ and $te$ both by 1 and then checks that
the value of $st$ (after firing $I$) is 0. Then, if $C \act{\alpha I} C'$ is a step, it must be that $\alpha$ is \emph{exactly}
$\frac{1}{2^m}$. This is because, by assumption, before firing this rule the value of $st$ was $\frac{1}{2^m}$ and
after firing this rule, the zero-test ensures that the value of $st$ is 0. Hence, the only possible 
value that $\alpha$ can take is $\frac{1}{2^m}$. Therefore, this rule allows us to add $\frac{1}{2^m}$ to the counter $c$.

However, note that after firing $I$, the values of $st$ and $te$ are reversed, i.e., the values of $st$ and $te$ are 0 and $\frac{1}{2^m}$, respectively.
This is undesirable, as we might once again want to use $st$ to simulate addition by $\frac{1}{2^m}$. 
Therefore, we 
add another rule $J$, which decrements $te$ by 1, increments $st$ by 1 and then checks that the value of $te$ (after firing $J$) is 0.
Then, a successful firing of the rule $J$ by some fraction $\beta$ means that $\beta = \frac{1}{2^m}$ (due to the same reasons as above)
and so this would mean that the values of $st$ and $te$ after firing $J$ would again become $\frac{1}{2^m}$ and $0$, respectively.
Hence, the counter $te$ essentially acts as a temporary holder of the value of $st$ and allows us to ``refill'' the value of $st$.

Generalizing this technique allows us to \emph{control the firing fraction} to perform doubling as well. 
However, this technique has a single obstacle, which we now address.



For this technique to work, we need a counter $st$ initially which stores the value $\frac{1}{2^m}$.
It might be tempting to simply declare that the value of $st$ in the initial configuration is $\frac{1}{2^m}$.
However, this cannot be done, because the number $m$ is given to us in binary and so the number of bits needed to write
down the number $\frac{1}{2^m}$ is exponential in the size of the given input $\langle \mach,m \rangle$, which 
would not give as a polynomial-time reduction.
This raises the following challenge as well:
\ramen{Wouldnt it be better to say "But this begs the question of, how do we create a value of $\frac{1}{2^m}$ in a counter in the first place?", and then briefly explain "write down in polynomial time". May be we could also already mention how we can retain the created $\frac{1}{2^m}$ by copying to another counter $st'$ which has value 0.}
\balain[Addressed both of these issues]
\begin{itemize}
	\item[(C2)] How can we create a value of $\frac{1}{2^m}$ in a continuous counter?
\end{itemize}
We show that challenge C2 can also be solved by our idea of controlling the firing fraction.  


\smallskip\noindent\textbf{Solving Challenge (C1). }  From the $\MCM$ $\mach = (Q,\init,\fin,\Delta)$, we construct a $\NCCM$ $\cC_0$ as follows. 
$\cC_0$ will have 4 counters $c_0$, $c_1$, $st$, $te$, i.e., it will be 4-dimensional.
Intuitively, each $c_i$ will store the value of one of the counters of $\mach$, $st$ will store the value $1/2^m$ that 
will be needed for simulating the addition operation, and
$te$ will be used to temporarily store the values of $c_0, c_1$ and $st$ at some points along a run.
A rule in $\cC_0$ consists of a vector $r \in \zn^4$ and a subset 
$s \subseteq \{1,2,3,4\}$. For ease of reading,  we write the vector $r$ as a sequence of increment or decrement operations $c \texttt{+=}n$ (or $c \texttt{-=}n$) whose intended meaning is that counter $c$ is incremented (or decremented) by $n$, followed by a sequence of zero-tests. 
For example, $t=(r, s)$ where $r=(1,0,0,-2)$ and $s=\{1,3\}$ is represented by $c_0\texttt{+=}1,te\texttt{-=}2;\ c_0=0?,st=0?$.

$\cC_0$ will have all the states of $\mach$ and in addition, for every rule $r$ of $\mach$, it will have a state $r_{mid}$. 
The set of rules of $\cC_0$ will be given as follows.
\begin{itemize}
	\item For the rule $r := (q,\inc_i,q')$ of $\mach$, $\cC_0$ will have the ``increment$(i)$'' gadget given in Figure~\ref{fig:case-one}.
	\item For the rule $r := (q,\double_i,q')$, $\cC_0$ will have the ``double$(i)$'' gadget given in Figure~\ref{fig:case-two}.
	\item For the rule $r := (q,\nop,q')$, $\cC_0$ will have the ``nop'' gadget given in Figure~\ref{fig:case-three}.
\end{itemize}

%
	%
	%
	%
	%
	%

\begin{figure}
	\centering
	\begin{subfigure}[b]{0.45\textwidth}
		\begin{center}
			\tikzstyle{node}=[circle,draw=black,thick,minimum size=12mm,inner sep=0.75mm,font=\normalsize]
			\tikzstyle{edgelabelabove}=[sloped, above, align= center]
			\tikzstyle{edgelabelbelow}=[sloped, below, align= center]
			\begin{tikzpicture}[->,node distance = 2cm,scale=0.8, every node/.style={scale=0.8}]
				\node[node] (q) {$q$};
				\node[node, right = of q] (r) {{\small $r_{mid}$}};
				\node[node, right = of r] (q') {$q'$};
				
				\draw(q) edge[edgelabelabove] node[above]{{\small $\cpp{c_i}, \ \cpp{te}, \ \cmm{st}$}} node[below]{{\small $st = 0?$}} (r);
				\draw(r) edge[edgelabelabove] node[above]{{\small $\cpp{st}, \ \cmm{te}$}} node[below]{{\small $te=0?$}} (q');
			\end{tikzpicture}
		\end{center}
		\caption{Gadget for the rule $r := (q,\inc_i,q')$. 
		}
		\label{fig:case-one}
	\end{subfigure} 
	\hfill
	\begin{subfigure}[b]{0.45\textwidth}
		\begin{center}
			\tikzstyle{node}=[circle,draw=black,thick,minimum size=12mm,inner sep=0.75mm,font=\normalsize]
			\tikzstyle{edgelabelabove}=[sloped, above, align= center]
			\tikzstyle{edgelabelbelow}=[sloped, below, align= center]
			\begin{tikzpicture}[->,node distance = 2cm,scale=0.8, every node/.style={scale=0.8}]
				\node[node] (q) {$q$};
				\node[node, right = of q] (r) {{\small $r_{mid}$}};
				\node[node, right = of r] (q') {$q'$};
				
				\draw(q) edge[edgelabelabove] node[above]{{\small $\cpptwo{te}, \ \cmm{c_i}$}} node[below]{{\small $c_i = 0?$}} (r);
				\draw(r) edge[edgelabelabove] node[above]{{\small $\cpp{c_i}, \ \cmm{te}$}} node[below]{{\small $te = 0?$}} (q');
			\end{tikzpicture}
		\end{center}
		\caption{Gadget for the rule $r:= (q,\double_i,q')$. 
		}
		\label{fig:case-two}
	\end{subfigure}
	\vskip \baselineskip
	\begin{subfigure}[b]{0.45\textwidth}
		\begin{center}
			\tikzstyle{node}=[circle,draw=black,thick,minimum size=12mm,inner sep=0.75mm,font=\normalsize]
			\tikzstyle{edgelabelabove}=[sloped, above, align= center]
			\tikzstyle{edgelabelbelow}=[sloped, below, align= center]
			\begin{tikzpicture}[->,node distance = 2cm,scale=0.8, every node/.style={scale=0.8}]
				\node[node] (q) {$q$};
				\node[node, right = of q] (r) {{\small $r_{mid}$}};
				\node[node, right = of r] (q') {$q'$};
				
				\draw(q) edge[edgelabelabove] node{} (r);
				\draw(r) edge[edgelabelabove] node{} (q');
			\end{tikzpicture}
		\end{center}
		\caption{Gadget for the rule $r:= (q,\nop,q')$.
		}
		\label{fig:case-three}
	\end{subfigure}
	\hfill
	\begin{subfigure}[b]{0.45\textwidth}
		\begin{center}
			\tikzstyle{node}=[circle,draw=black,thick,minimum size=12mm,inner sep=0.75mm,font=\normalsize]
			\tikzstyle{edgelabelabove}=[sloped, above, align= center]
			\tikzstyle{edgelabelbelow}=[sloped, below, align= center]
			\begin{tikzpicture}[->,node distance = 2cm,scale=0.8, every node/.style={scale=0.8}]
				\node[node] (q) {$\fin$};
				\node[node, right = of q] (r) {$\fin'$};
				\node[node, right = of r] (q') {$\overline{\fin}$};
				
				\draw(q) edge[edgelabelabove] node[above]{{\small $\cmm{st}$}} node[below]{{\small $st = 0?, te = 0?$}} (r);
				\draw(r) edge[edgelabelabove] node[above]{{\small $\cmm{c_0}, \cmm{c_1}$}} node[below]{{\small $c_0 = 0?, c_1 = 0?$}} (q');
			\end{tikzpicture}
		\end{center}
		\caption{The ``finish'' gadget. 
		}
		\label{fig:finish}
	\end{subfigure}
\end{figure}


Note that for every rule $r$ of $\mach$, the corresponding 
gadget in $\cC_0$ has exactly two rules, where the first rule (from $q$ to $r_{mid}$) will be denoted by $r^b$ and the second rule (from $r_{mid}$ to $q'$) by $r^e$.
We would now like to show that the rules of $\mach$ are simulated by their corresponding gadgets. To this end, we first define a mapping $g$ from configurations of $\mach$
to configurations of $\cC_0$ as follows: 
If $C = (q,v_0,v_1)$, then $g(C)$ is the configuration of $\cC_0$ such that
$$\sta{g(C)} = q , \ g(C)(c_0) = v_0/2^m, \ g(C)(c_1) = v_1/2^m, \ g(C)(st) = 1/2^m, \ g(C)(te) = 0$$
We now have the following ``gadget simulation'' lemma, which solves Challenge (C1). 

\begin{restatable}[Gadget Simulation]{lemma}{gadgetsim}~\label{lem:gadget-simulation}
	Suppose $C$ is a configuration and $r$ is a rule of $\mach$.
	\begin{itemize}
		\item Soundness: If $C \act{r} C'$, then there exists $\alpha, \beta$ such that $g(C) \act{\alpha r^b, \beta r^e} g(C')$.
		\item Completeness: If $g(C) \act{\alpha r^b, \beta r^e} D$ for some $\alpha, \beta$ and $D$, then there exists $C'$ such that $D = g(C')$ and $C \act{r} C'$.
	\end{itemize}	
\end{restatable}

\begin{proof}[Proof sketch.]
	We have already discussed the case of increments in some detail before
	and so we will concentrate on when $r$ is a doubling 
	rule of the form $(q,\double_i,q')$.
	The soundness part can be easily obtained by setting $\alpha = g(C)(c_i)$ and $\beta = 2 \cdot g(C)(c_i)$. 
	For completeness, note that since $r^b$ has a zero-test on $c_i$,
	it must be that $\alpha = g(C)(c_i)$. Hence, after firing 
	$\alpha r^b$, the value of $te$ must be $2 \cdot g(C)(c_i)$.
	Now since $r^e$ has a zero-test on $te$, it must be that
	$\beta = 2 \cdot g(C)(c_i)$. So the net
	effect of firing $\alpha r^b, \beta r^e$ is to make the 
	value of $c_i$ to be $2 \cdot g(C)(c_i)$. 
	Hence, if we let $C'$ be such that $C \act{r} C'$ in $\mach$, it can be verified that $g(C') = D$. 
	For more details, see Subsection~\ref{app:gadget-simulation} of the appendix.
\end{proof}


\smallskip\noindent\textbf{The ``finish'' gadget. } Before, we solve Challenge (C2), we make a small modification to $\cC_0$. 
Recall that in $\mach$, we have a \emph{set of final configurations} 
given by $F := \{(\fin,n,n) : n \le 2^m\}$, whereas in a $\NCCM$, we are
allowed to specify only one final configuration. 
However, the $\NCCM$ $\cC_0$ only promises us that the initial configuration $c_{init}$
of $\mach$ can reach some configuration in $F$ in $m$ steps
iff $g(c_{init})$ can reach some configuration in the 
set $\{g(D) : D \in F\}$ in $2m$ steps. 
Hence, we need to make a modification to $\cC_0$ which allows us to
replace the set of configurations with a single final configuration.
To this end, we modify $\cC_0$ by adding the ``finish gadget'' from Figure~\ref{fig:finish}, where $\fin'$ and $\overline{\fin}$ are two fresh states and the first and the second rule are respectively denoted by $f^b$ and $f^e$. Let us call the resulting $\NCCM$ as $\cC_1$. 

Note that the effect of firing $f^b$ is to set the values of $st$ and
$te$ to 0. Further, if $f^e$ is fired, then $c_0$ and $c_1$ are decremented by the same amount and both of them are tested for zero.
This means that $f^e$ could be fired successfully only if the counter values of $c_0$ and $c_1$ at state $\fin'$ are the same
and the effect of firing $f^e$ is to set the values of $c_0$ and $c_1$ to 0.
This observation along with repeated applications of the Gadget Simulation lemma give us the following Simulation theorem.

\begin{restatable}[Simulation Theorem]{theorem}{simulation}\label{thm:simulation}
	The initial configuration $c_{init}$ of $\mach$ can reach a final configuration in $m$ steps iff $g(c_{init})$ can reach the configuration $(\overline{\fin},\mathbf{0})$
	in $2m+2$ steps in $\cC_1$.
\end{restatable}

The full proof of this theorem can be found in Subsection~\ref{app:simulation} of the appendix. 
We now move on to solving Challenge (C2).

\smallskip\noindent\textbf{Solving Challenge (C2). } 
Thanks to the Simulation theorem, the required reduction is almost over. 
As we had discussed before, the only remaining part is that since $g(c_{init})(st) = 1/2^m$ and $m$ is already given in binary, we cannot write down $g(c_{init})$ in polynomial time. \rst{See previous comment.} \bala{added a mention to discussion before}
To handle this challenge (Challenge (C2)), we construct an ``initialization'' gadget which starts from a ``small'' initial configuration and then ``sets up'' the configuration
$g(c_{init})$. 

The initialization gadget is shown in the Figure~\ref{fig:init}. The gadget shares the counters $st$ and $te$ with $\cC_1$ and has two new counters $x$ and $count$. 
Initially, the gadget will start in $in_0$ and will have the values $1, 0, 1/m$ and $0$ in $st, te, x$ and $count$ respectively. 
In each iteration of the gadget, the value of $st$ will be halved. The function of $x$ is to 
store the value $1/m$ and the function of $count$ is to count the number of executions of 
this gadget. Initially the value 
of $count$ is 0 and in every iteration its value will increase by $1/m$. Hence, if we finally require the value of $count$ to be 1, then we would have
executed this gadget precisely $m$ times, thereby setting the value of $st$ to $1/2^m$.

The following lemma 
,whose full proof could be found in Subsection~\ref{app:lemma-init} of the appendix, 
follows from an analysis of the initalization gadget, similar to the one for the Gadget Simulation lemma. 
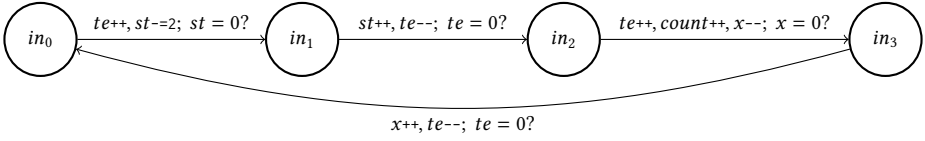
\begin{figure}
	\begin{center}
		\tikzstyle{node}=[circle,draw=black,thick,minimum size=12mm,inner sep=0.75mm,font=\normalsize]
		\tikzstyle{edgelabelabove}=[sloped, above, align= center]
		\tikzstyle{edgelabelbelow}=[sloped, below, align= center]
		\begin{tikzpicture}[->,node distance = 2.5cm,scale=0.8, every node/.style={scale=0.8}]
			\node[node] (in0) {$in_0$};
			\node[node, right = of in0] (in1) {$in_1$};
			\node[node, right = of in1] (in2) {$in_2$};
			\node[node, right = of in2, xshift=1cm] (in3) {$in_3$};
			
			\draw(in0) edge[edgelabelabove] node{$\cpp{te}, \cmmtwo{st}; \ st = 0?$} (in1);
			\draw(in1) edge[edgelabelabove] node{$\cpp{st}, \cmm{te}; \ te = 0?$} (in2);
			\draw(in2) edge[edgelabelabove] node{$\cpp{te}, \cpp{count}, \cmm{x}; \ x = 0?$} (in3);
			\draw(in3) edge[edgelabelbelow, bend left = 15] node{$\cpp{x}, \cmm{te}; \ te = 0?$} (in0);
		\end{tikzpicture}
	\end{center}
	\caption{The ``initialization'' gadget.} 
\label{fig:init}
\end{figure}

\begin{restatable}[The Initialization lemma]{lemma}{initlemma}\label{lem:init}
Suppose $C$ is a configuration of the initialization gadget such that $\sta{C} = in_0, C(te) = 0$ and $C(x) = 1/m$. Then we can execute one iteration of the 
gadget from $C$ to go to a configuration $C'$ if and only if $C'$ is the same as $C$ except that $C'(st) = C(st)/2$ and $C'(count) = C(count) + 1/m$.
\end{restatable}

We now construct our final $\NCCM$ $\cC$ as follows: We take the initialization gadget and the $\NCCM$ $\cC_1$ and we add a rule from $in_0$ to $\init$ which does not do anything to the
counters. Intuitively, we first execute the initialization gadget for some
steps and then pass the control flow to $\cC_1$.
We let $d_{init}$ be the configuration of $\cC$ whose state is $in_0$ and whose counter values are all 0, except for $d_{init}(x) = 1/m$ and $d_{init}(st) = 1$.
Then, we let $d_{fin}$ be the configuration of $\cC$ whose state is $\overline{\fin}$ and whose counter values are all 0, except for $d_{fin}(x) = 1/m$ and $d_{fin}(count) = 1$.
If we encode $d_{init}$ and $d_{fin}$ in binary, then they can be written down in polynomial time. Since $d_{fin}(count) = 1$, when
the control flow passes from the initialization gadget to $\cC_1$,
the value of $st$ must be $1/2^m$, which is exactly what we want. 

\begin{restatable}{theorem}{cOneToC}\label{thm:challenge2}
$g(c_{init})$ can reach the configuration $(\overline{\fin},\mathbf{0})$ in the $\NCCM$ $\cC_1$ in $2(m+1)$ steps if and only if $d_{init}$ can reach $d_{fin}$ in the
$\NCCM$ $\cC$ in $4m+1+2(m+1)$ steps.	
\end{restatable}

The full details behind this theorem can be found in Subsection~\ref{app:challenge2} of the appendix.
Combining this theorem with Theorem~\ref{thm:simulation}, proves the correctness of our reduction.



\section{From $\NCCM$ to $\QPVASS$}\label{sec:lower-bound-nccm-pvass}

We now move on to the next step in our reduction chain with the following problem called the reachability problem for $\QVASSRL$, defined
as follows:
Given a $\CVASS$ $\mach$, two configurations $c_{\mathit{init}}, c_{\mathit{fin}}$, and a number $m$ \emph{in binary}, whether one can
reach $c_{\mathit{fin}}$ from $c_{\mathit{init}}$ in exactly $m$ steps. 
\begin{theorem}
	The reachability problem for $\QVASSRL$ is $\NEXPTIME$-hard.
\end{theorem}

We prove this theorem by giving a reduction from the reachability problem for $\NCCM$. 
Fix a $\NCCM$ $\cC$, two of its configurations $c_{\mathit{init}}, c_{\mathit{fin}}$, and a number $m$.
Without loss of generality,
we assume that every rule in $\cC$ performs at least one zero-test.

\smallskip\noindent\textbf{Overview of the reduction.} We want to construct a $\CVASS$ $\mach$ that simulates $\cC$ for $m$ steps. The primary challenge that prevents us from doing this is the following:
\begin{itemize}
	\item[(D1)] How can we create gadgets to simulate exactly $m$ zero-tests of $\cC$ in $\mach$?
\end{itemize}

We circumvent this challenge as follows: We know that in a $\NCCM$, the value of every counter will always be in the range $[0,1]$. Hence, 
for every counter $x$, we introduce another counter $\bar{x}$, called the \emph{complementary counter} of $x$ and 
maintain the invariant $x + \bar{x} = 1$ throughout a run.
Then testing if the value of $x$ is 0, amounts to testing if the value of $\bar{x}$ is at least 1.
\newcommand{\geqone}{geq1}
This allows us to replace a zero-test with a greater than or equal to 1 (\geqone) test. 

The latter can be implemented as follows:
If $t$ and $t'$ are rules which decrement and increment $\bar{x}$ by 1 respectively and $C \act{1 t} C' \act{1 t'} C''$ is a run,
then we know that the value of $\bar{x}$ in $C$ is at least 1, which lets us implement a \geqone \ test. Note that for this to succeed, 
we require that both $t$ and $t'$ are fired completely, i.e.,
with fraction 1.

To sum this up, this means that if were to simulate a rule $r = (q,(w,s),q')$ of the $\NCCM$ $\cC$ in our new machine with the complementary counters, we need one rule to take care of the updates corresponding to $w$
and two rules to take care of \geqone \ tests corresponding to the zero tests in $s$, both of which
must be \emph{fired completely}.
Hence, simulating $m$ steps of $\cC$ in our new machine 
requires $3m$ steps, of which
exactly $2m$ steps must be \emph{fired completely}. 
This leads us to
\begin{itemize}
	\item[(D2)] How can we force the rules corresponding to \geqone \ tests to be fired completely, for exactly $2m$ times?
\end{itemize}
To solve this challenge, we introduce another counter $ctrl$, called the \emph{controlling counter}. 
We modify every rule corresponding to a \geqone \ test to also
increment the value of the counter $ctrl$ by 1. 
This means that, if $\rho$ is a run of $3m$ steps such that the value
of $ctrl$ after $\rho$ is exactly $2m$, then every rule corresponding to a \geqone
\ test must have been fired completely along the run $\rho$. 

\smallskip
\noindent\textbf{Formal construction.} Having given an informal overview of the reduction, we now proceed to the formal construction.
We are given an $\NCCM$ $\cC$ and a number $m$ in binary.
From the $\NCCM$ $\cC$, we will construct a $\CVASS$ $\mach$
as follows. For every counter $x$ of $\cC$, $\mach$ will have two counters $x$ and $\bar{x}$.
Every transition that increments $x$ will decrement $\bar{x}$ by the same amount and vice-versa, so that the sum of the values of $x$ and $\bar{x}$ will be equal to 1 throughout. Further, $\mach$
will have another counter $ctrl$, called the \emph{controlling counter}.

Suppose $r := (q,t,q')$ is a rule of $\cC$ such that $t = (w,s)$. 
Denote by $\bar{w}$ the vector such that $w(ctrl) = 0$ and
for every counter $x$ of $\cC$, $\bar{w}(x) = w(x)$ and $\bar{w}(\bar{x}) = -w(x)$. Then corresponding to the rule $r$ in $\cC$, $\mach$ will have the gadget in Figure~\ref{fig:cvass}, whose first, second
and third rules will be denoted by $r^b, r^m$ and $r^e$ respectively.

\begin{figure}
	\begin{center}
		\tikzstyle{node}=[circle,draw=black,thick,minimum size=12mm,inner sep=0.75mm,font=\normalsize]
		\tikzstyle{edgelabelabove}=[sloped, above, align= center]
		\tikzstyle{edgelabelbelow}=[sloped, below, align= center]
		\begin{tikzpicture}[->,node distance = 0.6cm,scale=0.8, every node/.style={scale=0.8}]
			\node[node] (q) {$q$};
			\node[node, right = of q] (r) {$r_{mid}$};
			\node[node, right = of r, xshift = 2.5cm] (r') {$r'_{mid}$};
			\node[node, right = of r', xshift = 2.5cm] (q') {$q'$};
			
			\draw(q) edge[edgelabelabove] node{$\bar{w}$} (r);
			\draw(r) edge[edgelabelabove] node{$\land_{x \in s}\ \cmm{\bar{x}}, \cpp{x}; \cpp{ctrl}$} (r');
			\draw(r') edge[edgelabelabove] node{$\land_{x \in s}\ \cpp{\bar{x}}, \cmm{x}; \cpp{ctrl}$} (q');	
		\end{tikzpicture}
	\end{center}
	\caption{Gadget for the rule $r := (q,t,q')$ with $t = (w,s)$.
	}
	\label{fig:cvass}
\end{figure}
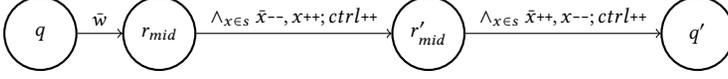

For any configuration $C$ of $\cC$, let $G(C)$ \ramen{In the previous section, $g$ was a function. Here it is a relation. I propose using $G$ or something else so that the reader is not confused. One coud also use notation like $C[ctrl \leftarrow \alpha]$ to indicate a configuration which only differs in $ctrl$.} \bala{done} denote the \emph{set of configurations} of $\mach$ such that $D \in G(C)$ iff $\sta{D} = \sta{C}$, $D(x) = C(x)$
and $D(\bar{x}) = 1-C(x)$ for every counter $x$ of $\cC$. Note that
any two configurations in $G(C)$ differ only in their value of the counter $ctrl$. For any number $\alpha$, let $G(C)_\alpha$ denote the unique
configuration in $G(C)$ whose $ctrl$ value is $\alpha$.
The following lemma is a consequence of the discussion given in the overview section.
, whose full proof could be found in Subsection~\ref{app:ctrl} of the appendix.

\begin{restatable}[Control counter simulation]{lemma}{ctrl}~\label{lem:ctrl}
	\begin{itemize}
		\item Soundness: If $C \act{\alpha r} C'$ in $\cC$, then
		for any $\zeta$, $G(C)_\zeta \act{\alpha r^b, r^m, r^e} G(C')_{\zeta+2}$. 
		\item Completeness: If $G(C)_\zeta \act{\alpha r^b, \beta r^m, \gamma r^e } D'$ for some $\alpha, \beta, \gamma, \zeta$ and $D$
		such that $D(ctrl) = \zeta + 2$, 
		then there exists $C'$ such that $D' = G(C')_{\zeta+2}$, $\beta = \gamma = 1$ and 
		$C \act{\alpha r} C'$.
	\end{itemize}
\end{restatable}

Repeated applications of the Control Counter Simulation lemma give us the following theorem, which completes our reduction.

\begin{restatable}{theorem}{nccmToCvass}\label{thm:nccm-to-cvass}
	$c_{init}$ can reach $c_{fin}$ in $\cC$ in $m$ steps iff $G(c_{init})_0$ can reach $G(c_{fin})_{2m}$ in $3m$ steps.	
\end{restatable}

The full proof of the above theorem can be found in Subsection~\ref{app:nccm-to-cvass} of the appendix.

\begin{example}
	Let us see a concrete application of this reduction on some example.
	To this end, consider the $\NCCM$ given in Figure~\ref{fig:inc-renamed}. Note
	that this is essentially a renamed version of the ``increment(i)'' gadget
	described in Figure~\ref{fig:case-one}. We consider this version here
	since it makes it easier to describe the effect of our reduction.
	The result of the application of the reduction on this $\NCCM$ is
	given in Figure~\ref{fig:cvass-inc-renamed}.
	
	Suppose for some $u, v \in [0,1]$ with $u+v \le 1$,
	we start in state $q_0$ in the $\NCCM$ given
	in Figure~\ref{fig:inc-renamed} with counter values $u, v$ and $0$
	for the counters $c, st$ and $te$, respectively. 
	From the argument given in the previous section, we know that if we 
	fire the $\NCCM$ in Figure~\ref{fig:inc-renamed} once, then
	we will reach the state $q_2$ with counter values $u + v, v$ and $0$
	for the counters $c, st$ and $te$, respectively. 
	
	Now, suppose we start in $q_0$ in the $\QVASSRL$ given in 
	Figure~\ref{fig:cvass-inc-renamed} with counter values
	$u, 1-u, v, 1-v, 0, 1$ and $0$ in $c, \overline{c}, st, \overline{st}, te, \overline{te}$ and $ctrl$,
	respectively. From the reduction, we know that if we fire the gadget in Figure~\ref{fig:cvass-inc-renamed}
	once, and reach the state $q_2$ with counter value $4$ for the controlling counter $ctrl$, 
	then the counter values for counters $c, \overline{c}, st, \overline{st}, te, \overline{te}$ are $u+v, 1 - u -v, v, 1-v, 0$ and $1$ respectively.
\end{example}

\begin{figure}
	\begin{center}
		\tikzstyle{node}=[circle,draw=black,thick,minimum size=12mm,inner sep=0.75mm,font=\normalsize]
		\tikzstyle{edgelabelabove}=[sloped, above, align= center]
		\tikzstyle{edgelabelbelow}=[sloped, below, align= center]
		\begin{tikzpicture}[->,node distance = 2cm,scale=0.8, every node/.style={scale=0.8}]
			\node[node] (q) {$q_0$};
			\node[node, right = of q] (r) {$q_1$};
			\node[node, right = of r] (q') {$q_2$};
			
			\draw(q) edge[edgelabelabove] node[above]{{\small $\cpp{c}, \ \cpp{te}, \ \cmm{st}$}} node[below]{{\small $st = 0?$}} (r);
			\draw(r) edge[edgelabelabove] node[above]{{\small $\cpp{st}, \ \cmm{te}$}} node[below]{{\small $te=0?$}} (q');
		\end{tikzpicture}
	\end{center}
	\caption{Renamed version of the ``increment(0)'' gadget from Figure~\ref{fig:case-one}.
		The rule from $q_0$ to $q_1$ shall be denoted  
		by $I$ and the rule from $q_1$ to $q_2$ shall be denoted by $J$.}
	\label{fig:inc-renamed}
\end{figure}
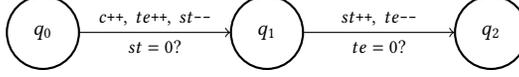 
\ramen{Why not continue to use $r^b$ notation in \cref{fig:inc-renamed}, especialy since $s$ is used for the subset for which we zero test in previous section?}
\balain[I am using this notation again in this section. Hence the renaming. However, I have now renamed it to $I$ and $J$, instead of $r$ and $s$.]

\begin{figure}
	\begin{center}
		\tikzstyle{node}=[circle,draw=black,thick,minimum size=12mm,inner sep=0.75mm,font=\normalsize]
		\tikzstyle{edgelabelabove}=[sloped, above, align= center]
		\tikzstyle{edgelabelbelow}=[sloped, below, align= center]
		\begin{tikzpicture}[->,node distance = 0.6cm,scale=0.8, every node/.style={scale=0.8}]
			\node[node] (q0) {$q_0$};
			\node[node, right = of q0, xshift=2cm] (r) {$I_{mid}$};
			\node[node, right = of r, xshift = 2cm] (r') {$I'_{mid}$};
			\node[node, below right = of r', xshift = 2cm, yshift = 0.5cm] (q1) {$q_1$};
			\node[node, below = of r'] (s) {$J_{mid}$};
			\node[node, left = of s, xshift = -2cm] (s') {$J'_{mid}$};
			\node[node, left = of s', xshift = -2cm] (q2) {$q_2$};
			
			\draw(q0) edge[edgelabelabove] node[above]{{\small $\cpp{c},\cpp{te},\cmm{st}$}} node[below]{{\small $\cmm{\overline{c}},\cmm{\overline{te}},\cpp{\overline{st}}$}} (r);
			\draw(r) edge[edgelabelabove] node{$\cmm{\overline{st}}, \cpp{st}, \cpp{ctrl}$} (r');
			\draw(r') edge[edgelabelabove] node{$\cpp{\overline{st}}, \cmm{st}, \cpp{ctrl}$} (q1);	
			\draw(q1) edge[edgelabelabove] node[above]{{\small $\cpp{st},\cmm{te}$}} node[below]{{\small $\cmm{\overline{st}},\cpp{\overline{te}}$}} (s);
			\draw(s) edge[edgelabelabove] node{$\cmm{\overline{te}}, \cpp{te}, \cpp{ctrl}$} (s');
			\draw(s') edge[edgelabelabove] node{$\cpp{\overline{te}}, \cmm{te}, \cpp{ctrl}$} (q2);	
		\end{tikzpicture}
	\end{center}
	\caption{Application of the reduction described in this section on the example $\NCCM$ given in Figure~\ref{fig:inc-renamed}.
	}
	\label{fig:cvass-inc-renamed}
\end{figure}
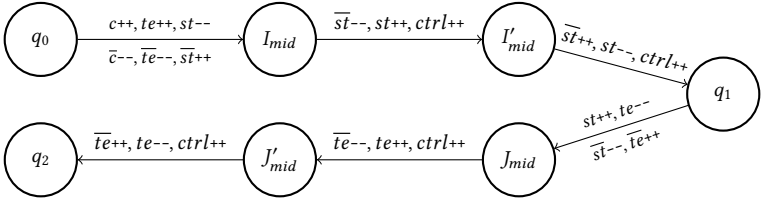


\subsection{Wrapping up} \label{subsec:cvass-to-cgvas}

We now provide the final steps to prove that reachability for 
$\QPVASS$ is $\NEXPTIME$-hard. To do this, we recall a well-known 
folklore fact about pushdown automata. It essentially states
that we can implement a binary counter in a PDA.

\begin{lemma}\label{lem:folklore}
	For any number $m$, in polynomial time in $\log(m)$, we
	can construct a PDA $P_m$ of bounded stack-height and two configurations $C$ and $C'$ such that there is
	exactly one run
	from $C$ to $C'$. Moreover, this run is of length exactly $m$.
\end{lemma}

\begin{proof}
	The essential idea is to use the stack to do a depth-first search of a binary tree of size $O(m)$. At each point, the PDA will only store at most $O(\log m)$ many entries in its stack, because
	the depth of the tree is $O(\log m)$. We now give a more precise construction.
	
	Note that when $m = 1$, $P_1$ can simply be taken to be a PDA with two states and a single transition
	between the first state and the second state which does nothing to the stack.
	Now, let us consider the case when $m > 1$ is a power of 2, i.e., $m = 2^k$ for some $k$.
	Consider the following PDA $P_m$ with $k$ stack symbols $S_1^k, S_2^k, \dots, S_k^k$.
	$P_m$ starts in the state $b_m$ with the empty stack. It then moves to 
	state $e_m$ while pushing $S_1^k$ onto the stack. The state $e_m$ has $k$ self-loop transitions
	as follows: For each $1 \le i < k$, the $i^{th}$ self-loop pops $S_i^k$ and pushes $S_{i+1}^k$ twice.
	Further, the $k^{th}$ self-loop simply pops $S_k^k$. 
	It can be easily verified that
	starting from state $b_m$ with the empty stack, there is exactly one path to the configuration
	whose state is $e_m$ and whose stack is empty. Moreover this path is of length exactly $m$.
	This is because the desired path is essentially the depth-first search traversal of a binary tree of size $m-1$,
	where the root is labelled by $S_1^k$ and each node at height $i$ is labelled by $S_{i+1}^k$.
	Due to the depth-first search traversal, the number of elements stored in the stack at any point during the run is $O(k)$.

	Now for the general case, suppose $m = \sum_{1 \le i \le n} 2^{k_i}$ for some $k_1 < k_2 < \dots < k_n \le \log(m)$. The desired PDA $P_m$ has $\sum_{1 \le i \le n} k_i$ stack symbols
	given by $S_1^{k_1},\dots,S_{k_1}^{k_1},S_1^{k_2},\dots,S_{k_2}^{k_2},\dots,S_1^{k_n},\dots,S_{k_n}^{k_n}$.
	Further, $P_m$ has $n+1$ states $b_m^1, \dots, b_m^n, b_m^{n+1}$. Initially, it starts in the state $b_m^1$
	with the empty stack. Then for each $1 \le i \le n$, it has a transition from $b_m^i$ to $b_m^{i+1}$
	which pushes $S_1^{k_i}$ onto the stack. Then, at state $b_m^{n+1}$, it has
	the following set of self-loops: For each $1 \le i \le n$ and each $1 \le j < k_i$, 
	it pops $S_j^{k_i}$ from the stack and pushes $S_{j+1}^{k_i}$ twice. Further for each $1 \le i \le n$,
	it pops $S_{k_i}^{k_i}$. It can now be easily verified that starting from state $b_m^1$ with the empty stack, there is exactly one path to the configuration
	whose state is $b_m^{n+1}$ and whose stack is empty and also that this path is of length exactly $m$.
\end{proof}

We now give a reduction from reachability for $\QVASSRL$ to 
reachability for $\QPVASS$.
Let $\mach = (Q,T,\Delta)$ be a $\CVASS$ such that $c_{init}$
and $c_{fin}$ are two of its configurations and let $m$ be a number,
encoded in binary. Construct the pair $(P_m,C,C')$ as given
by the Folklore lemma~\ref{lem:folklore}. We now take the usual cross product, i.e., the Cartesian product between $P_m$ and $\mach$,
to obtain a $\QPVASS$ $\cC$. (This operation is very similar to taking the cross product 
between a PDA and an NFA).
Intuitively, the PDA part of $\cC$ corresponds to simulating a binary
counter, counting till the value $m$ 
and the $\QVASS$ part of $\cC$ corresponds
to simulating the $\QVASS$ $\mach$.

Let $\bu$ (resp. $\bv$) be the configuration of $\cC$ such that
$\sta{\bu} = (\sta{C},\sta{c_{init}}), \stack(\bu) = \stack(C), \val(\bu) = \val(c_{init})$ (resp. $\sta{\bv} = (\sta{C'},\sta{c_{fin}}), \stack(\bv) = \stack(C'), \val(\bv) = \val(c_{fin})$).
%
By construction of $\cC$, 
$c_{init}$ can reach $c_{fin}$ in $\mach$ in $m$ steps
iff $\bu$ can reach $\bv$ in $\cC$.



\begin{theorem}
	Reachability in $\QPVASS$ is $\NEXPTIME$-hard.
\end{theorem}

	\section{Coverability, number of counters and encodings}\label{sec:lower-bound-discussion}

The chain of reductions from reachability in $\MCM$ to reachability in $\QPVASS$ prove that the latter is $\NEXPTIME$-hard. 
The reduction from $\MCM$ to $\NCCM$ was accomplished by using 6 counters,
and the reduction from $\NCCM$ to $\QVASSRL$ used $2x+1$ counters where $x$ is the number of counters of the $\NCCM$ instance.
Finally, the reduction from $\QVASSRL$ to $\QPVASS$ did not add any new counters. It follows that the lower bound already holds for $\QPVASS$
of dimension 13.

We can go one step further. Similar to reachability in $\QVASSRL$
we can define coverability in $\QVASSRL$, where want to cover
some configuration in a given number of steps.
Let us inspect the $\QVASSRL$ $\mach$ that we constructed
in Section 6. We claim that 
\begin{quote}
	$G(c_{init})_0$ can reach $G(c_{fin})_{2m}$ in $3m$ steps iff $G(c_{init})_0$ can cover $G(c_{fin})_{2m}$ in $3m$ steps.
\end{quote}

The left-to-right implication is trivial. For the other direction, notice that in any run of $3m$ steps in $\mach$ starting from $G(c_{init})_0$, the value of $ctrl$ can be
increased by at most $2m$. Further, for every counter $x \neq ctrl$, we maintain the invariant $x + \bar{x} = 1$ throughout.
It then follows that the only way to cover $G(c_{fin})_{2m}$ in $3m$
steps is by actually reaching $G(c_{fin})_{2m}$.
Hence, coverability in $\QVASSRL$ is also $\NEXPTIME$-hard.
Since the reduction in \cref{subsec:cvass-to-cgvas} preserves coverability, we obtain:
\begin{theorem}
	The coverability problem for 13-dimensional $\QPVASS$ is $\NEXPTIME$-hard. 
\end{theorem}



Let us now consider the encoding of the numbers that we use. 
It can be easily verified that in the final $\QPVASS$ instance that we construct using our chain of reductions from Sections~\ref{sec:lower-bound-mcm-hardness} till~\ref{sec:lower-bound-nccm-pvass}, all the numbers are fixed constants, except for the numbers appearing in the initial and final configurations, which are encoded in binary. 
Hence, the above theorem holds for 13-dimensional $\QPVASS$ where the numbers are encoded in binary. We show that it is possible to strengthen this result to unary-encoded numbers at the cost of increasing the number of counters by a constant. 
More specifically, in Section~\ref{app:discussion} of the appendix, we present an alternate reduction, which given an instance of reachability in $\NCCM$ over
$x$ counters produces an instance of coverability in $\QVASSRL$ over $10x+25$ counters where all numbers are encoded in unary. 
(We have already discussed the idea of this reduction in the \nameref{ingre:v}). Since the proof in Section~\ref{sec:lower-bound-mcm-nccm} shows that reachability in $\NCCM$ is $\NEXPTIME$-hard already over 6 counters, 
this will prove that coverability in $\QVASSRL$ over 85 counters where all numbers are encoded in unary is also $\NEXPTIME$-hard. Since the reduction given in Section~\ref{subsec:cvass-to-cgvas} from $\QVASSRL$ to $\QPVASS$ produces a $\QPVASS$ of bounded stack-height, does not add any new counters and does not change the encodings of the numbers, we can now conclude the following theorem.

\begin{restatable}{theorem}{discussionhardness}\label{thm:discussion-hardness}
	The coverability problem for $\QPVASS$ is $\NEXPTIME$-hard, already over $\QPVASS$es of dimension 85, bounded stack-height, and when all numbers are encoded in unary.
\end{restatable}

This hardness result is very strong, as it simultaneously achieves coverability, bounded stack, constant dimensions, and unary encodings. In contrast, in $\NEXPTIME$,
we can decide reachability of $\QPVASS$ over arbitary dimension, even when all the numbers are encoded in binary. 

Finally, 
the reduction from $\QVASSRL$ to $\QPVASS$ in~\cref{subsec:cvass-to-cgvas} only used the fact that for every $m$, (1)~there is a PDA of size $O(\log(m))$ which can ``count''
exactly till $m$ and (2)~we can take product of a PDA with a $\QVASS$.
For any model of computation that satisfies these two constraints, the corresponding reachability problem over continuous counters should also be $\NEXPTIME$-hard.
For instance, if we replace a stack in $\QPVASS$ with \emph{Boolean programs} to define \emph{Boolean programs with continuous counters} then their reachability and coverability
problems are also $\NEXPTIME$-hard. A similar result also holds when we replace the stack with a (discrete) one-counter machine which can only increment its counter and 
whose accepting condition is reaching a particular counter value given in binary. For both models, the reachability and coverability problems must also be in $\NEXPTIME$,
because the former can be converted into an exponentially bigger $\QVASS$, for which these problems are in $\NP$~\cite[Theorem 4.14]{blondinLogicsContinuousReachability2017}.

\section{Conclusion}

We have shown that the reachability problem for continuous pushdown VASS is $\NEXPTIME$-complete.
While our upper bound works for any arbitrary number of counters, our lower bound already
holds for the coverability problem for continuous pushdown VASS with constant number of counters, bounded stack-height
and when all numbers are encoded in unary. 

As part of future work, it might be interesting to study the complexity of coverability reachability for continuous pushdown VASS over low dimensions. It might also be interesting to study the coverability and reachability problems for extensions of continuous
pushdown VASS.
For instance, it is already known that reachability in continuous VASS in which the counters are allowed to be tested for zero is undecidable~\cite[Theorem 4.17]{blondinLogicsContinuousReachability2017}.
It might be interesting to see if this is also the case when the continuous counters are endowed with
operations such as resets or transfers. Finally, it would be nice to extend the decidability result
here
to other machine models, such as continuous VASS with higher-order stacks.
	
	\begin{acks}
		The authors are grateful to the anonymous reviewers for their helpful comments and for pointing out a small (and easily fixable) mistake in an earlier version.
		This research was sponsored in part by the \grantsponsor{DFG}{Deutsche Forschungsgemeinschaft}{https://www.dfg.de/} project \href{https://gepris.dfg.de/gepris/projekt/389792660}{\grantnum{DFG}{389792660}} TRR 248–CPEC.
		This project has received funding from the \grantsponsor{501100000781}{European Research Council (ERC)}{http://dx.doi.org/10.13039/501100000781} under
		the European Union’s Horizon 2020 research and innovation programme under grant
		agreement number \href{https://doi.org/10.3030/787367}{\grantnum{501100000781}{787367}} (PaVeS). 
		Funded by the European Union (ERC, FINABIS, \href{https://doi.org/10.3030/101077902}{\grantnum{501100000781}{101077902}})\marginpar{\includegraphics[width=1.3cm]{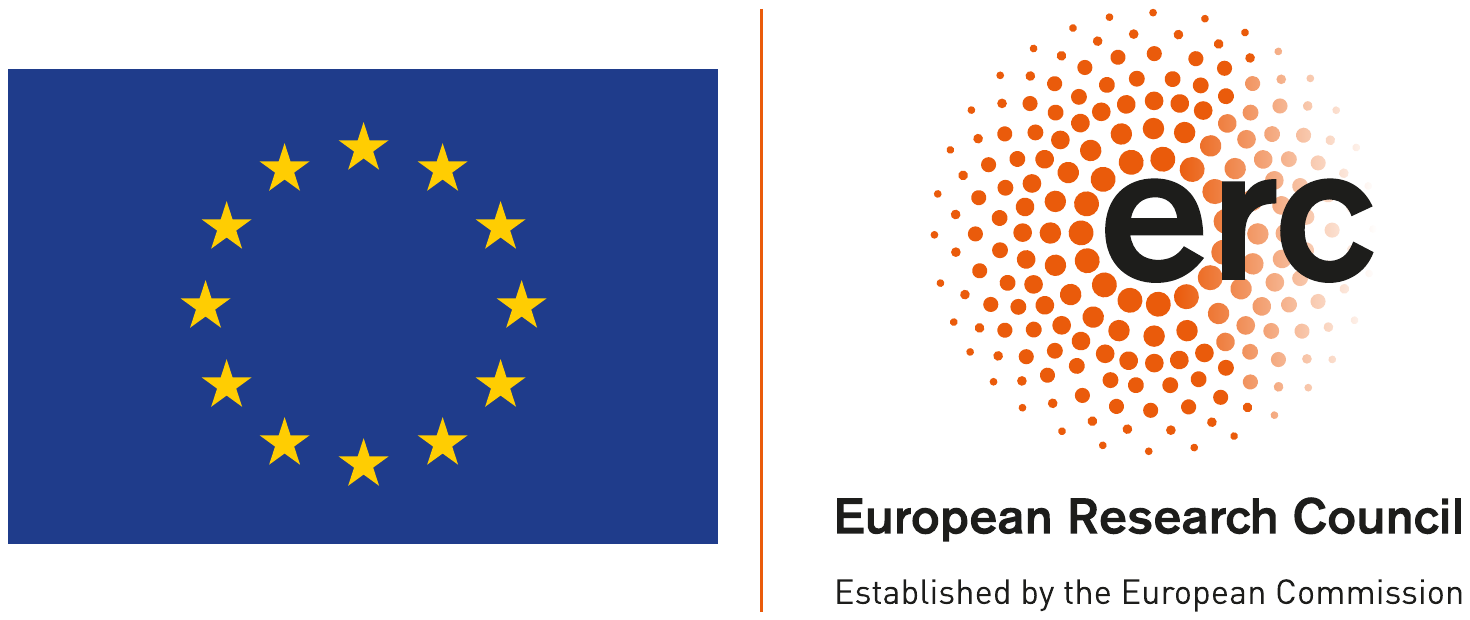}}. Views and opinions expressed are however those of the authors only and do not necessarily reflect those of the European Union or the European Research Council Executive Agency. Neither the European Union nor the granting authority can be held responsible for them.
	\end{acks}
	
	\label{beforebibliography}
	\newoutputstream{pages}
	\openoutputfile{main.pages.ctr}{pages}
	\addtostream{pages}{\getpagerefnumber{beforebibliography}}
	\closeoutputstream{pages}
	
	\bibliographystyle{ACM-Reference-Format}
	\bibliography{refs}
	
		\newpage
		\appendix

\section{Proofs for Section~\ref{sec:upper-bound}}\label{sec:appendix-upper-bound}
\subsection{Proof details of \cref{three-runs} [Proposition 4.5, \cite{blondinLogicsContinuousReachability2017}]}
\label{sec:appendix-upper-bound-three-runs}

Let us first sketch the proof. Recall that a $d$-dimensional $\CVASS$ $\cC=(Q,T,\Delta)$ has a finite set of states $Q$ and a finite set of vectors $T \subseteq \integ^d$ and transitions $\Delta \subseteq Q \times T \times Q$ which can be fired fractionally. 
	Proposition~4.5 in \cite{blondinLogicsContinuousReachability2017} states the corresponding result for $\CVASS$ with a finite state set, but actually works for an infinite state set and infinitely many transitions. This is because the proof works by concatenation and repeated fractional firing of the three runs. The lemma in \cite{blondinLogicsContinuousReachability2017} contains a fourth condition on the supports of the three runs which we express as letter-uniformity. We can reduce our case of the letter-uniform semigroup $K$ to an infinite $\CVASS$ $\cC_{\infty}=(Q_{\infty}, T, \Delta_{\infty})$ where $Q_{\infty}=\{ q_0\} \cup \bigcup_{w \in K} Q_w$ contains $|w|$ many states $Q_w =\{q_{w,1},\ldots, q_{w,n}\}$ for each $w=w_1\ldots w_n \in K$ along with a single distinguished state $q_0$, $\Delta_{\infty}$ comprises of transitions $q_0 \step{}{w_1} q_{w,1}, q_{w,1} \step{}{w_2} q_{w,2}, \ldots, q_{w,n} \step{}{w_n} q_0$ on  $\{q_0\} \cup Q_w$ for each $Q_w$, and $T=\Gamma$ where $K \subseteq \Gamma^*$ . Clearly there exists a cyclic run $(q_0,\mvec u) \step{\ratplus}{w} (q_0, \mvec v)$ iff $\mvec{u}\step{\ratplus}{K}\mvec{v}$. Applying the result from \cite{blondinLogicsContinuousReachability2017} to $\cC_{\infty}$ proves the lemma.

We showed that the result follows from the corresponding statement for infinite $\CVASS$, whose proof we outline. In particular, condition (C4) in the lemma below is satisfied in our case due to the letter-uniformity condition. 
\begin{lemma}[Proposition 4.5, \cite{blondinLogicsContinuousReachability2017} ] \label{lem:threerunInf}
		Let $(Q,T)$ be a $d$-dimensional $\CVASS$, $q_0 \in Q$ where $|Q|$ is infinite and $T \subseteq Q \times \integ^d \times Q$. Then we have $(q_0,\vec u) \step{\ratplus}{} (q_0,\vec v)$ if and only if  there	are $\mvec{u}',\mvec{v}'\in\Q^d$ such that: 
	\begin{itemize}
		\item[(C1)] $(q_0,\mvec{u})\step{\rat}{\pi}(q_0,\mvec{v})$,
		\item[(C2)] $(q_0,\mvec{u})\step{\ratplus}{\pi_f}(q_0,\mvec{v}')$, 
		\item[(C3)] $(q_0,\mvec{u}')\step{\ratplus}{\pi_b}(q_0,\mvec{v})$, and
		\item[(C4)] $\supp{\pi}=\supp{\pi_f}=\supp{\pi_b}$.
	\end{itemize}
\end{lemma}
As can be seen, a run from $\vec u$ to $\vec v$ immediately implies the four conditions. 
The proof of the converse direction follows from a series of lemmas which we describe below. 

The first lemma gives us \emph{scaling} properties of $\CVASS$ runs
\begin{lemma}[Lemma 4.2, \cite{blondinLogicsContinuousReachability2017}]
	\label{lem:CVASSscaling}
	\begin{itemize}
		\item[(S1)] $(p,\vec u) \step{\ratplus}{\pi} (q,\vec v)$ iff $(p,\alpha \vec u) \step{\ratplus}{\alpha \pi} (q,\alpha \vec v)$. 
		\item[(S2)] If $(p, \vec u) \step{\ratplus}{\alpha \pi} (q,\vec v)$, $(p, \vec u') \step{\ratplus}{\alpha' \pi} (q, \vec v')$ and $\alpha+\alpha' \leq 1$ then $(p,\alpha \vec u + \alpha' \vec u') \step{\ratplus}{(\alpha + \alpha') \pi} (q,\alpha \vec v +\alpha' \vec v')$.
		\item[(S3)] If $(p, \vec u) \step{\ratplus}{ \pi} (q, \vec v)$ then $(p, \vec u) \step{\ratplus}{\alpha \pi} (q,(1-\alpha)\vec u + \alpha\vec v )$.
		\item[(S4)]  $(p, \vec u) \step{\ratplus}{ \pi} (q, \vec v)$ in $\cV$ iff $(q, \vec v) \step{\ratplus}{ \pi^{-1}} (p, \vec u)$ in $\cV^{-1}$ where $\cV^{-1}$ is the dual $\CVASS$ where all transitions are reversed in both direction and sign. 
	\end{itemize}
\end{lemma}
The properties follow by simple induction on the length of runs. 

The second lemma states a set of sufficient conditions for the existence of a run. For any transition $t=(q,\delta(t),q') \in T$, let $\supp{\delta(t)}^- \subseteq [d]$ (resp. $\supp{\delta(t)}^+ \subseteq [d]$) be the set of indices $i$ where $t(i) <0$ (resp. $t(i) \geq 0$). 
\begin{lemma}
	\label{lem:runSufficiency}
	If for $q \in Q$, $\vec u, \vec v \in \ratplus^d$ and $\pi \in \addstar{T}$ the following conditions hold:
	\begin{itemize}
		\item[(1)] $(q, \vec u) \step{\rat}{ \pi} (q, \vec v)$,
		\item[(2)] $\supp{\delta(t)}^- \subseteq \supp{\vec u}$ for each $t \in \supp{\pi}$, and 
		\item[(3)] $\supp{\delta(t)}^+ \subseteq \supp{\vec v}$ for each $t \in \supp{\pi}$, 
	\end{itemize}
	then there is $\pi' \in \addstar{T}$ such that $(q, \vec u) \step{\ratplus}{ \pi'} (q, \vec v)$ and $\Psi(\pi)=\Psi(\pi')$
\end{lemma}
The key idea is to show that we can repeatedly fire a small fraction of $\pi$ from (1) is order to get the desired $\pi'$. Conditions (2) and (3) ensure that there exists a small enough fraction $\alpha$ for which, when $\alpha \pi$ is fired, then we obtain a run with $\ratplus$ semantics. Define 
\[neg(i)= \sum_{t \in T, i \in \supp{\delta(t)}^-} \Psi(\pi)(t)|\delta(t)(i)|,\]
\[pos(i)= \sum_{t \in T, i \in \supp{\delta(t)}^+} \Psi(\pi)(t)\delta(t)(i)\]
$neg(i)$ and $pos(i)$ capture the total negative positive effect on counter $i$ by the transitions in $\pi$.
Let
\[ \alpha_{neg}=\min(1,\min\{ \frac{\vec u(i)}{neg(i)} \mid i \in [d], neg(i) >0 \}),\]
\[ \alpha_{pos}=\min(1,\min\{ \frac{\vec v(i)}{pos(i)} \mid i \in [d], pos(i) >0 \}),\]
\[ n=\max(\lceil \frac{1}{\alpha_{neg}}\rceil,\lceil \frac{1}{\alpha_{pos}}\rceil,2), \quad \sigma=\frac{\pi}{n}, \quad \pi'=\sigma^n.\]
One can show that each subrun $\vec u_i \step{\ratplus}{\sigma} \vec u_{i+1}$ exists for $\vec u_i=\vec u + \frac{i}{n}\delta(\pi)$. For the case where $i=0$ (resp. $i=n-1$), this follows from the fact that $\frac{1}{n}$ is chosen smaller than $\alpha_{neg}$ (resp. $\alpha_{pos}$). For other values of $i$, one shows that $\vec u_i=\frac{n-1-i}{n-1} \vec u_0 + \frac{i}{n-1} \vec u_{n-1}$ and $\vec u_{i+1}=\frac{n-1-i}{n-1} \vec u_1 + \frac{i}{n-1} \vec u_{n}$ and applies scaling property (S2). Thus $(q, \vec u) \step{\ratplus}{ \pi'} (q, \vec v)$.

The next lemma observes that if there is a run $(p, \vec u) \step{\ratplus}{ \pi} (q, \vec v)$ then there is a run reaching $q$ which satisfies a support condition. 
\begin{lemma}[Lemma 4.4, \cite{blondinLogicsContinuousReachability2017}] \label{lem:suppCoverRun}
	Let $q \in Q$, $\vec u, \vec v \in \ratplus^d$ and $\pi \in \addstar{T}$. If $(p, \vec u) \step{\ratplus}{ \pi} (q, \vec v)$ then there is $\vec v'$ and $\pi' \in \addstar{T}$ such that
	\begin{itemize}
		\item $(p, \vec u) \step{\ratplus}{ \pi'} (q, \vec v')$,
		\item $path(\pi)=path(\pi')$, and
		\item $\supp{\delta(t)} \subseteq	\supp{\vec v'}$ for every $t \in \supp{\pi'}$. 
	\end{itemize}
\end{lemma}
It is easily observed that if $\pi=\alpha_1t_1\ldots \alpha_m t_m$, then the run $\pi'=\frac{\alpha_1}{2}t_1 \ldots \frac{\alpha_i}{2^i}\ldots \frac{\alpha_m}{2^m}t_m$ satisfies the conditions. By halving the firing fraction each time, we can ensure that if a transition removes tokens from a place that was positive in $\vec u$, we retain some tokens in this place in $\vec v'$. If a transition puts tokens in a place, we also retain some fraction of it in $\vec v'$.

We can now prove \cref{lem:threerunInf}. By using Lemma \ref{lem:suppCoverRun}, we obtain $\pi'_f$ from $\pi_f$ in (C2) such that $path(\pi'_f)=path(\pi_f)$, $(q, \vec u) \step{\rat}{ \pi'_f} (q, \vec x)$ and $\supp{\delta_t} \subseteq \supp{\vec x}$ for each $t \in \pi'_f$. Similarly we obtain $(q, \vec y) \step{\rat}{ \pi'_b} (q, \vec v)$ from $\pi_b$ with corresponding support conditions. Let $\pi$ from (C1) be $\alpha_1t_1\ldots \alpha_m t_m$, then for $j \in [m]$ define
\[ a_j = \frac{\Psi(\pi'_f)(t_j)}{\Psi(path(\pi)(t_j)},\]
\[ b_j = \frac{\Psi(\pi'_b)(t_j)}{\Psi(path(\pi)(t_j)},\]
\[\text{ let } 0 < \lambda \leq 1 \text{ small such that } \forall j \in [m], c_j=\alpha_j -\lambda(a_j+b_j) \geq 0, \text{ and }\]
\[ \sigma=c_1t_1\ldots c_mt_m\]
We can now fire this small fraction $\lambda$ of $\pi'_f$ and $\pi'_b$ in order to arrive at vectors which satisfy the conditions (2) and (3) of \cref{lem:runSufficiency}:
\[ (q, \vec u) \step{\ratplus}{ \lambda\pi'_f} (q, \vec x') \text{ where } \supp{\vec x} \subseteq \supp{\vec x'}, \text{ and}\]
\[ (q, \vec y') \step{\ratplus}{ \lambda\pi'_b} (q, \vec v) \text{ where } \supp{\vec y} \subseteq \supp{\vec y'}. \]
Furthermore, we have: 
\[ (q, \vec x') \step{\rat}{ \lambda\sigma} (q, \vec y'). \]
Applying \cref{lem:runSufficiency} to the above run, we get:
\[ (q, \vec x') \step{\ratplus}{ \sigma'} (q, \vec y'), \]
and the desired run $\pi'=\lambda \pi'_f \cdot \sigma' \cdot \lambda \pi'_b$ with $\supp{\pi'}=\supp{\pi}=\supp{\pi'}$ and $\Psi(\pi')=\Psi(\pi)$.

\section{Proofs for Section~\ref{sec:lower-bound-mcm-nccm}}

\subsection{Proof details of Lemma~\ref{lem:gadget-simulation}}\label{app:gadget-simulation}


	The lemma is immediate when $r$ is a nope rule of the form $(q,\nop,q')$. We will first prove it in the case when $r$ is an increment rule of the form $(q,\inc_i,q')$.
	For the soundness part, suppose $C \act{r} C'$. Then, notice that $g(C) \act{1/2^m \ r^b} D' \act{1/2^m \ r^e} g(C')$.
	To prove the completeness part, suppose $g(C) \act{\alpha r^b} D' \act{\beta r^e} D$ for some $\alpha, \beta, D$ and $D'$. 
	Since $r^b$ decrements $st$ by 1 and also has a zero-test on $st$, it follows that $g(C)(st) - \alpha = 0$ and so $\alpha = g(C)(st) = 1/2^m$.
	By construction of the rule $r^b$, it follows that the values of the counters $c_i, c_{1-i}, st$ and $te$ in $D'$
	are $g(C)(c_i) + 1/2^m, g(C)(c_{1-i}), 0$ and $1/2^m$ respectively. 
	Now, since $r^e$ decrements $te$ by 1 and also has a zero-test on $te$, it follows that $D'(te) - \beta = 0$ and so $\beta = D'(te) = 1/2^m$.
	By construction of the rule $r^e$, it follows that $\sta{D} = q'$ and the values of the counters $c_i, c_{1-i}, st$ and $te$ in $D$
	are $D(c_i), D(c_{1-i}), 1/2^m$ and $0$ respectively. Since $D(c_i) = g(C)(c_i) + 1/2^m$ and $D(c_{1-i}) = g(C)(c_{1-i})$, 
	it follows that if we let $C \act{r} C'$ in the machine $\mach$, then $g(C') = D$.

	For the case when $r$ is a doubling rule of the form $(q,\double_i,q')$, a similar argument applies, where
	$\alpha$ and $\beta$ (in both the soundness and the completeness part) will be $g(C)(c_i)$ and $2g(C)(c_i)$ respectively.

\subsection{Proof details of Theorem~\ref{thm:simulation}}\label{app:simulation}


	Suppose $c_{init}$ can reach a final configuration $c_{fin} = (\fin,n,n)$ for some number $n$ in $m$ steps. 
	By repeated applications of the soundness part of the Gadget Simulation lemma (Lemma~\ref{lem:gadget-simulation}), it follows
	that $g(c_{init})$ can reach $D := g(c_{fin})$ in $2m$ steps. Since $\sta{D} = \fin$ and $D(c_0) = D(c_1)$, it follows
	that if we set $\alpha = D(st), \beta = D(c_0)$, then by construction of $f^b$ and $f^e$ we have $D \act{\alpha f^b, \beta f^e} (\overline{q_f},\mathbf{0})$.

	Suppose $g(c_{init})$ can reach the configuration $(\overline{q_f},\mathbf{0})$ in $2(m+1)$ steps in $\cC_1$.
	By construction of $\cC_1$, any execution from $C := g(c_{init})$ to $C' := (\overline{q_f},\mathbf{0})$ must be of the form
	$C \act{\alpha_1 r_1^b, \beta_1 r_1^e} C_1 \act{\alpha_2 r_2^b, \beta_2 r_2^e} C_2 \dots C_{m-1} \act{\alpha_m r_m^b, \beta_m r_m^e} C_m \act{\alpha_{m+1} f^b, 
	\beta_{m+1} f^e} C'$ for some fractions $\alpha_1, \beta_1, \dots, \alpha_{m+1}, \beta_{m+1}$ and some rules $r_1, \dots, r_m$ of the machine $\mach$.
	By the completeness part of the Gadget Simulation lemma (Lemma~\ref{lem:gadget-simulation}), it follows that there exists configurations $D_1, D_2, \dots, D_m$
	such that each $C_i = g(D_i)$ and $c_{init} \act{r_1} D_1 \act{r_2} D_2 \dots D_{m-1} \act{r_m} D_m$. 
	Since $\alpha_{m+1} f^b$ is fireable from $C_m$, it must be the case that $\sta{C_m} = \fin$.
	By construction of the finish gadget, it
	follows that $C_m(c_0) = C_m(c_1)$. By construction of the mapping $g$, it follows that $D_m$ is of the form $(\fin,n,n)$ for some number $n$.
	Hence, $c_{init}$ can reach some final configuration in $m$ steps in $\mach$. 	

\subsection{Proof details of Lemma~\ref{lem:init}}\label{app:lemma-init}


	Suppose $C'$ is the same as $C$ except that $C'(st) = C(st)/2$ and $C'(count) = C(count) + 1/m$. Let $\alpha_1 = C(st)/2, \alpha_2 = C(st)/2, \alpha_3 = 1/m, \alpha_4 = 1/m$.
	It can be verified that $C \act{\alpha_1 i^1} C_1 \act{\alpha_2 i^2} C_2 \act{\alpha_3 i^3} C_4 \act{\alpha_4 i^4} C'$.
	
	Suppose $C \act{\alpha_1 i^1} C_1 \act{\alpha_2 i^2} C_2 \act{\alpha_3 i^3} C_4 \act{\alpha_4 i^4} C'$ is a run such that $C(te) = 0$ and $C(x) = 1/m$. Since $i^1$ decrements $st$ by 2 and also has a zero-test on $st$,
	it must be the case that $\alpha_1 = C(st)/2$. Hence, 
	\begin{equation*}
		C_1(st) = 0, C_1(te) = C(st)/2, C_1(x) = 1/m, C_1(count) = C(count)
	\end{equation*}
	
	Since $i^2$ decrements $te$ by 1 and also has a zero-test on $te$, it must be the case that $\alpha_2 = C_1(te) = C(st)/2$. Hence,
	\begin{equation*}
		C_2(st) = C(st)/2, C_2(te) = 0, C_2(x) = 1/m, C_2(count) = C(count)
	\end{equation*}

	Since $i^3$ decrements $x$ by 1 and also has a zero-test on $x$, it must be the case that $\alpha_3 = C_2(x) = 1/m$. Hence,
	\begin{equation*}
		C_3(st) = C(st)/2, C_3(te) = 1/m, C_3(x) = 0, C_3(count) = C(count) + 1/m
	\end{equation*}

	Since $i^4$ decrements $te$ by 1 and also has a zero-test on $te$, it must be the case that $\alpha_4 = C_3(te) = 1/m$. Hence,
	\begin{equation*}
		C_4(st) = C(st)/2, C_4(te) = 0, C_4(x) = 1/m, C_4(count) = C(count) + 1/m
	\end{equation*}
	Hence, the proof is complete.

\subsection{Proof details of Theorem~\ref{thm:challenge2}}\label{app:challenge2}


	Suppose $g(c_{init})$ can reach $(\overline{q_f},\mathbf{0})$ in $\cC_1$ in $2(m+1)$ steps by means of a run $\rho$.
	By applying the Initialization lemma $m$ times and then using the rule from $in_0$ to $\init$ which does nothing, 
	it follows that starting from $d_{init}$ we can reach the configuration $D$
	where $\sta{D} = \init$ and $D(c_0) = D(c_1) = 0, D(st) = 1/2^m, D(te) = 0, D(x) = 1/m$ and $D(count) = 1$.
	Since $\cC$ contains all the rules of $\cC_1$ and since $g(c_{init})$ and $D$ agree on their states and
	the counter values of $c_0, c_1, st$ and $te$, we can simply execute the run $\rho$ from the configuration $D$
	in $\cC$ to reach $d_{fin}$.
	
	Suppose $d_{init}$ can reach $d_{fin}$ in exactly $\ell := 4m+1+2(m+1)$ steps in $\cC$. Let 
	$D_0 := d_{init} \act{} D_1 \act{} D_2 \dots D_{\ell-1} \act{} d_{fin} := D_{\ell}$ be such a run. 
	By construction of the $\NCCM$ $\cc$, there must be exactly one index $j$ such that $\sta{D_j} = \init$ and $\sta{D_{j-1}} = in_0$
	and the rule fired between $D_{j-1}$ and $D_j$ does nothing to the counters.
	By construction of $\cC$, this means that the steps from $D_0$ to $D_{j-1}$ are steps in the initialization gadget
	and the steps from $D_{j-1}$ to $D_{\ell}$ are steps in $\cC_1$. Since no step in $\cC_1$ affects the value of the counter $count$
	and since $d_{fin}(count)$ is 1, it follows that $D_{j-1}(count) = 1$. Since the steps from $D_0$ to $D_{j-1}$
	are steps in the initialization gadget and since $\sta{D_0} = \sta{D_{j-1}} = in_0$, by repeated applications of the Initialization lemma,
	it follows that $j-1$ must be $4m$ and the values of the counters in $D_{j-1}$ satisfy
	\begin{equation*}
		D_{j-1}(c_0) = D_{j-1}(c_1) = 0, D_{j-1}(st) = 1/2^m, D_{j-1}(te) = 0, D_{j-1}(x) = 1/m, D_{j-1}(count) = 1
	\end{equation*}
	Since the counter values of $D_j$ and $D_{j-1}$ are the same, the same equation holds if we replace $D_{j-1}$ with $D_j$. 
	Notice that $\sta{D_j} = \init, D_j(c_0) = D_j(c_1) = D_j(te) = 0, D_j(st) = 1/2^m$. Since the run from $D_j$ to $d_{fin}$
	uses only the rules from $\cC_1$, it follows that executing the same sequence of rules also
	gives us a run from $g(c_{init})$ to $(\overline{q_f},\mathbf{0})$ in $\cC_1$ of length $\ell-j = \ell - 4m - 1 = 2(m+1)$.		

\section{Proofs for Section~\ref{sec:lower-bound-nccm-pvass}}

\subsection{Proof details of Lemma~\ref{lem:ctrl}}\label{app:ctrl}


	Let us fix a rule $r = (q,t,q')$ in $\cC$ with $t = (w,s)$.
	
	For the soundness claim, suppose $C \act{\alpha r} C'$ in $\cC$. It follows that from $G(C)_\zeta$ we can fire $r^b$ to reach a configuration $D$
	whose counter values are the same as $G(C')_\zeta$. By assumption, $C'(x) = 0$ for every $x \in s$. 
	It follows that from $G(C')_\zeta$ we can fire $r^m$ and then $r^e$ to reach the configuration $G(C')_{\zeta+2}$.
	
	For the completeness part, suppose $G(C)_\zeta \act{\alpha r^b} D^b \act{\beta r^m} D^m \act{\gamma r^e} D'$ for some $\alpha, \beta, \gamma, \zeta$ and $D^b, D^m, D'$ such that $D'(ctrl) = \zeta+2$. Let $D = G(C)_\zeta$.
	By construction of $r^b$, it follows that $D^b(c) = D(c) + \alpha \bar{w}$ for every counter $c$.
	By definition of $D$ and $\bar{w}$, it follows that 
	\begin{equation}\label{eq:one}
			D^b(c) \in [0,1] \text{ for every counter } c \neq ctrl.		
	\end{equation}

	Further, since the value of $ctrl$ has increased by 2 between $g(C)_\zeta$ and $D'$, by construction of the gadget, it follows that $\beta = \gamma = 1$. 
	By construction of $r^m$, this means that $D^b(\bar{x}) = 1$ for every $x \in S$ and so 
	\begin{equation}\label{eq:two}
		D^b(x) = 0 \text{ for every } x \in S.
	\end{equation}
	
	Combining equations~\ref{eq:one} and~\ref{eq:two} and using the definition of the rules $r$ and $r^b$, it follows that there exists $C'$ such that $C \act{\alpha r} C'$ in $\cC$ and
	$\valu{D^b} = \valu{G(C')_{\zeta}}$. Notice that since $\beta = \gamma = 1$, by construction of $r^m$ and $r^e$, it follows that $\valu{D^b}(c) = \valu{D'}(c)$ for every $c \neq ctrl$.
	Since $D'(ctrl) = \zeta+2$, we have that $\valu{D'} = \valu{G(C')}_{\zeta+2}$. Further, since $\sta{D'} = \sta{G(C')_{\zeta+2}}$, it follows that $D' = G(C')_{\zeta + 2}$.	

\subsection{Proof details of Lemma~\ref{thm:nccm-to-cvass}}\label{app:nccm-to-cvass}


	Suppose $c_{init}$ can reach $c_{fin}$ in $\cC$ in $m$ steps. Repeatedly applying the soundness part of the control counter simulation lemma (Lemma~\ref{lem:ctrl}) to every
	step of this run, it follows that $G(c_{init})_0$ can reach $G(c_{fin})_{2m}$ in $3m$ steps in $\mach$.
	
	Suppose $C_0 := G(c_{init})_0 \act{} C_1 \act{} C_2 \dots C_{3m-1} \act{} C_{3m} := G(c_{fin})_{2m}$ is a run in $\mach$. By construction of $\mach$ it follows that
	for every $0 \le i \le m-1$, there exists a rule $r_i$ of $\cC$ and fractions $\alpha_i, \beta_i, \gamma_i$ such that $C_{3i} \act{\alpha_i r^b} C_{3i+1} \act{\beta_i r^m} C_{3i+2}
	\act{\gamma_i r^e} C_{3i+3}$. By construction of $r^b, r^m$ and $r^e$, we have $C_{3i+3}(ctrl) = C_{3i}(ctrl) + \beta_i + \gamma_i$. 
	Hence, $2m = C_{3m}(ctrl) = C_0(ctrl) + \sum_{0 \le i \le m-1} (\beta_i + \gamma_i)$. Since $C_0(ctrl) = 0$ and since each $\beta_i$ and $\gamma_i$ is a fraction,
	it follows that for every $i$, $\beta_i = \gamma_i = 1$. Hence, $C_{3i+3}(ctrl) = C_{3i}(ctrl) + 2$. Repeatedly applying the completeness part of the control counter
	simulation lemma (Lemma~\ref{lem:ctrl}), it follows that there is a run from $c_{init}$ to $c_{fin}$ using the sequence $\alpha_1 r_1, \alpha_2 r_2, \dots, \alpha_m r_m$,
	thereby completing the proof.

		\section{Proof details for Theorem~\ref{thm:discussion-hardness}}\label{app:discussion}

In this section, we present the proof details behind Theorem~\ref{thm:discussion-hardness}. 
To prove this, we first show that coverability in $\QVASSRL$ is $\NEXPTIME$-hard even for a constant number of counters and even when all the numbers are encoded in unary. We do this in four stages.
The first two stages recall the chain of reductions starting from Bounded-MPCP and observe some minute details about those reductions.
The next two stages exploit these details in order to prove the desired hardness result.


\subsection*{Stage 1: Hardness of structured instances of reachability in $\NCCM$. }

We first prove that reachability in $\NCCM$ is $\NEXPTIME$-hard even under some special constraints on the instances. An instance $\langle \cC,c_{init},c_{fin},m \rangle$ of the reachability problem for $\NCCM$ is said to be 
\emph{structured} if all the \emph{denominators} of the initial and the final 
configurations are powers of 2. 
We now show that reachability in $\NCCM$ is 
$\NEXPTIME$-hard even for structured instances. 
To do this, we first trace the chain of reductions starting from Bounded-PCP.

In~\cite[Section 6.1]{AiswaryaMS22}, it was shown that Bounded-PCP is $\NEXPTIME$-hard
even when the bound $\ell$ is restricted to be a power of 2. By slightly modifying the
reduction given in the proof of Theorem~\ref{thm:mcm} (by adding more nope instructions
if necessary), we can then show that reachability in $\MCM$ is $\NEXPTIME$-hard
even when the bound on the length is restricted to be a power of 2.
Examining the reduction from reachability in $\MCM$ to reachability in $\NCCM$ given in Section~\ref{sec:lower-bound-mcm-nccm}, we notice that all the numerators of the initial and the 
final configurations constructed by the reduction are all either 0 or 1 and all the
denominators are either $1$ or $m$, where $m$ is the bound on the length of the reachability in the given $\MCM$. Hence, this immediately implies that reachability in $\NCCM$
is $\NEXPTIME$-hard for structured instances.

\subsection*{Stage 2: Hardness of super-structured instances of coverability in $\QVASSRL$. }

Next we prove that coverability in $\QVASSRL$ is hard even with some special 
constraints on its instances. An instance $\langle \mach,c_{init},c_{fin},m \rangle$ of the coverability problem for $\QVASSRL$ is said to be 
\emph{structured} if all the \emph{denominators} of the initial and the final 
configurations are powers of 2. It is said to be \emph{super-structured} if it is structured and all the values of the initial and the final configurations are \emph{at most 1}. We now show that coverability in $\QVASSRL$ is
$\NEXPTIME$-hard even when restricted to super-structured instances.

Consider the reduction from reachability in $\NCCM$ to coverability in $\QVASSRL$
given in Section~\ref{sec:lower-bound-nccm-pvass}. If the given instance of reachability in $\NCCM$
is structured, then the reduction also outputs a structured instance of coverability in $\QVASSRL$. We now give a reduction from structured instances of coverability in $\QVASSRL$ to super-structured instances of coverability in $\QVASSRL$.

Let $\langle \mach,c_{init},c_{fin},m \rangle$ be a structured instance of coverability in $\QVASSRL$.
Let $2^k$ be the smallest power of two which is strictly bigger than all the numbers
appearing in $c_{init}, c_{fin}$ and $m$. (Notice that $2^k$ can be at most twice
the maximum number appearing in any one of these terms).
To $\mach$, add two new counters $b$ and $\overline{b}$ and replace each rule $r := (q,t,q')$ with the gadget in Figure~\ref{fig:full-to-super}.

\begin{figure}
	\begin{center}
		\tikzstyle{node}=[circle,draw=black,thick,minimum size=12mm,inner sep=0.75mm,font=\normalsize]
		\tikzstyle{edgelabelabove}=[sloped, above, align= center]
		\tikzstyle{edgelabelbelow}=[sloped, below, align= center]
		\begin{tikzpicture}[->,node distance = 2cm,scale=0.8, every node/.style={scale=0.8}]
			\node[node] (q) {$q$};
			\node[node, right = of q] (r) {$r_{mid}$};
			\node[node, right = of r] (q') {$q'$};
			
			\draw(q) edge[edgelabelabove] node{$t, \cmm{b}, \cpp{\overline{b}}$} (r);
			\draw(r) edge[edgelabelabove] node{$\cpp{b}, \cmm{\overline{b}}$} (q');
		\end{tikzpicture}
	\end{center}
	\caption{Gadget for the rule $r := (q,t,q')$. Assume that the 
		values of $b$ and $\overline{b}$ are $1/2^k$ and 0 to begin with respectively.
		The first rule ($r^b$) from $q$ to $r_{mid}$ does everything that $r$ does and in addition decrements $b$ and increments $\overline{b}$. The second rule ($r^e$) from $r_{mid}$ to $q'$ aims to bring the values of $b$ and $\overline{b}$ to $1/2^k$ and 0 respectively.}
	\label{fig:full-to-super}
\end{figure}
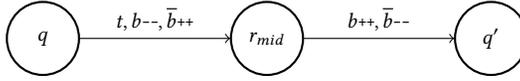 

Call the resulting $\CVASS$ as $\mach_1$. For any configuration $C$ of $\mach$, let
$H(C)$ denote the \emph{set of configurations} of $\mach_1$ such that $D \in H(C)$
iff $\sta{D} = \sta{C}, D(x) = C(x)/2^k$ for every counter $x$ of $\mach$ and
$D(b) + D(\overline{b}) = 1/2^k$. Note that any two configurations in $H(C)$
differ only by their values of $b$ and $\overline{b}$. Further, fixing any one of these
two values determines the value of the other as well. For any $\zeta \le 1/2^k$, let 
$H(C)_\zeta$ denote the unique configuration in $H(C)$ whose $b$ value is $\zeta$.
The following lemma is a simple consequence of the construction of the gadgets in $\mach_1$.

\begin{lemma}[Fractional Simulation Lemma]
	Suppose $C$ is a configuration and $r$ is a rule of $\mach$.
	\begin{itemize}
		\item (Soundness): If $C \act{\alpha r} C'$ in $\mach$, then for any
		$\zeta \le 1/2^k$, $H(C)_\zeta \act{\alpha/2^k r^b, (1-\alpha)/2^k r^e} H(C')_\zeta$.
		\item (Completeness): If $H(C)_\zeta \act{\alpha r^b, \beta r^e} D$ for 
		some $\alpha, \beta, D$ and some $0 < \zeta \le 1/2^k$ then $\alpha, \beta \le 1/2^k$ and
		there exists $C',\zeta'$ such that $0 < \zeta' \le 1/2^k, D = H(C')_{\zeta'}$ and $C \act{(\alpha \cdot 2^k) r} C'$.
	\end{itemize}
\end{lemma}

Because of the above lemma, we immediately get
\begin{theorem}
	$c_{init}$ can cover $c_{fin}$ in $\mach$ in exactly $m$ steps if and only if $H(c_{init})_{1/2^k}$ can
	cover $H(c_{fin})_{1/2^k}$ in $\mach_1$ in exactly $2m$ steps. 
\end{theorem}

Notice that $\langle \mach_1,H(c_{init})_{1/2^k},H(c_{fin})_{1/2^k},2m \rangle$ is a super-structured 
instance of coverability in $\QVASSRL$, which completes the required claim. 

\subsection*{Stage 3: Hardness of coverability in $\QVASSRL$ with unary encodings using amplifiers.}

In the next stage, assuming the existence of objects called \emph{amplifiers}, we prove our main result, that the coverability problem for $\QVASSRL$ is $\NEXPTIME$-hard even for a constant number of counters and even
when all the numbers are encoded in unary. We proceed to first define the notion of an amplifier.

Let $p_1, p_2, \dots, p_m$ be natural numbers, all of them encoded in binary using exactly $n$ bits, and let $k$ be a number such that
each $p_i \le 2^k$. Assume the following theorem.
\begin{theorem}[$\CVASS$ Amplifiers Theorem]~\label{thm:amplifier}
	In polynomial time in $n, m$ and $k$, we can construct a $\CVASS$ $Amp(p_1,\dots,p_m,k)$ and
	two configurations $c$ and $c'$ satisfying the following properties:
	\begin{itemize}
		\item $Amp(p_1,\dots,p_m,k)$ has $2m+5$ counters, which are called $y,y_1,\dots,y_{m+2}, \overline{y_1},\dots,\overline{y_{m+2}}$.
		\item The control flow graph of $Amp(p_1,\dots,p_m,k)$ is a path of length exactly $3(2k+4nm)$ which begins at $\sta{c}$ and 
		ends at $\sta{c'}$.
		\item Every value of $c$ and the value $c'(y)$ is encoded in unary.
		\item $c'(y_i) = p_i/2^k$ for every $i \le m$, $c'(\overline{y_i}) = 1 - c'(y_i)$ for every $i \le m+2$.
		\item There is exactly one run from $c$ which can reach a configuration $d$ such that $d(y) \ge c'(y)$.
		Further, this run is of length exactly $3(2k+4nm)$ and the configuration reached at the end of this run is $c'$ itself.
	\end{itemize}	
Such a tuple $\langle Amp(p_1,\dots,p_m,k),c,c' \rangle$ is called a $\CVASS$ amplifier for $(p_1,\dots,p_m,k)$.
\end{theorem}

Here the control flow graph of a $\CVASS$ $\mach$ is the graph obtained by taking the set of states of $\mach$ as its vertices
and connecting any two vertices by an edge if there is a rule between them.

The reason behind calling such a $\CVASS$ an amplifier is that starting from a configuration encoded in \emph{unary}, it is able
to reach a configuration encoded in \emph{binary}. 

Assuming the above theorem, we now reduce coverability in $\QVASSRL$ over super-structured instances to coverability in $\QVASSRL$ with unary encodings,
which will prove the required $\NEXPTIME$-hardness result. In the next subsection, we will prove the $\CVASS$ Amplifiers theorem, which
will complete our overall proof.

Let $\langle \mach,c_{init},c_{fin},\ell\rangle$ be a super-structured instance of coverability in $\QVASSRL$ with counters $x_1,\dots,x_m$. (Note that $c_{init}$ and $c_{fin}$ are encoded
in binary). Without loss of generality, we can assume that there are no incoming rules (resp. outgoing rules) to (resp. from) $\sta{c_{init}}$ (resp. $\sta{c_{fin}}$). By definition of a super-structured instance, the denominator of each value in $c_{init}$ and $c_{fin}$ is a power of 2.
By multiplying each numerator with a sufficient power of 2, we may assume that the denominator of each value in the initial and final configurations
is the same $2^k$ for some number $k$. (Note that this is a polynomial-time transformation).
Let $c_{init}(x_i) = \alpha_i/2^k$ and $c_{fin}(x_i) = \beta_i/2^k$ for each $i$. By padding with sufficient zeros, we can assume that each $\alpha_i$ and $\beta_i$
is described the same number of bits (say $n$). Hence, we have $2m$ numerators each of which is less than $2^k$ and each of which are describable by exactly 
$n$ bits.

The essential idea behind the reduction is to simply plugin $Amp(\alpha_1,\dots,\alpha_m,k)$ and $Amp(\beta_1,\dots,\beta_m,k)$ before
and after $\mach$. We now expand on this idea. By Theorem~\ref{thm:amplifier}, in polynomial time in the size
of the given super-structured instance we can construct amplifiers $\langle \mach_{d},d_{init},d_{fin} \rangle$ and $\langle \mach_{e},e_{init},e_{fin} \rangle$
for $(\alpha_1,\dots,\alpha_m,k)$ and $(\beta_1,\dots,\beta_m,k)$ respectively.
Let $y,y_1,\dots,y_{m+2},\overline{y_1},\dots,\overline{y_{m+2}}$ be the counters of $\mach_d$ and 
$z,z_1,\dots,z_{m+2},\overline{z_1},\dots,\overline{z_{m+2}}$ be the counters of $\mach_e$ respectively.
Without loss of generality, we can assume that the sets of counters $\{x_1,\dots,x_m\}, \{y,y_1,\dots,\overline{y_{m+2}}\}, \{z,z_1,\dots,\overline{z_{m+2}}\}$ are disjoint.

\paragraph*{The construction. } 
Let $q_1, q_2, q_3, q_4, q_5, q_6$ be the states of $d_{init}, d_{fin}, c_{init}, c_{fin}, e_{init}, e_{fin}$ respectively.
We now construct a $\CVASS$ $\cC$ as given in Figure~\ref{fig:final}. 
We use the following shorthand notation in that figure as follows:
Between the states $q_2$ and $q_3$ we have a rule labelled by $T^{inc}(x,y,\overline{y})$.
This is shorthand for $m+2$ rules of the form 
$$q_2 \act{\cpp{x_1},\cmm{y_1},\cpp{\overline{y_1}}} p_1 
\act{\cpp{x_2},\cmm{y_2},\cpp{\overline{y_2}}} p_2 \dots p_{m-1} \act{\cpp{x_m},\cmm{y_m},\cpp{\overline{y_m}}} p_m \act{\cmm{y_{m+1}},\cpp{\overline{y_{m+1}}}}
	p_{m+1} \act{\cmm{y_{m+2}},\cpp{\overline{y_{m+2}}}} q_3$$
where $p_1,\dots,p_{m+1}$ are new states. Similarly, from $q_6$ to $end$ we have
a rule labelled by $T^{dec}(x,z,\overline{z})$. This is also shorthand for $m+2$ rules of the same kind
as before, obtained by introducing fresh states $p_1',\dots,p_{m+1}'$ and the following rules,
$$q_6 \act{\cmm{x_1},\cmm{z_1},\cpp{\overline{z_1}}} p_1' 
\act{\cmm{x_2},\cmm{z_2},\cpp{\overline{z_2}}} p_2' \dots p_{m-1}' \act{\cmm{x_m},\cmm{z_m},\cpp{\overline{z_m}}} p_m' \act{\cmm{z_{m+1}},\cpp{\overline{z_{m+1}}}}
p_{m+1}' \act{\cmm{z_{m+2}},\cpp{\overline{z_{m+2}}}} end$$

Note that in the first set of rules the $x_i$'s were incremented, whereas in the second set they are decremented.

\begin{figure}
	\begin{center}
		\tikzstyle{node}=[circle,draw=black,thick,minimum size=7mm,inner sep=0.75mm,font=\normalsize]
		\tikzstyle{edgelabelabove}=[sloped, above, align= center]
		\tikzstyle{edgelabelbelow}=[sloped, below, align= center]
		\begin{tikzpicture}[->,node distance = 0.2cm,scale=0.8, every node/.style={scale=0.8}]
			\node[node] (start1) {{\small $q_1$}};
			\node[node, draw = none, right = of start1] (dots1) {$\dots \dots$};
			\node[node, right = of dots1] (end1) {{\small $q_2$}};
			\node[node, draw = none, below = of start1, xshift = 1.5cm, yshift = -0.2cm] (A1) {{\Large $\mach_{d}$}};
			\draw[dashed] (start1) ++(-0.8,1) rectangle (3,-1);
			
			\node[node, right = of end1, xshift = 2.25cm] (start2) {{\small $q_3$}};
			\node[node, draw = none, right = of start2] (dots2) {$\dots \dots$};
			\node[node, right = of dots2] (end2) {{\small $q_4$}};
			\node[node, draw = none, below = of start2, xshift = 1.5cm, yshift = -0.4cm] (A2) {{\Large $\mach$}};
			\draw[dashed] (start2) ++(-0.6,1) rectangle (8.2,-1);
			
			\node[node, right = of end2, xshift=1cm] (startl) {{\small $q_5$}};
			\node[node, draw = none, right = of startl] (dotsl) {$\dots \dots$};
			\node[node, right = of dotsl] (endl) {{\small $q_6$}};
			\node[node, draw = none, below = of startl, xshift = 1.5cm, yshift = -0.2cm] (Al) {{\Large $\mach_{e}$}};
			\draw[dashed] (startl) ++(-0.8,1) rectangle (13.6,-1);

			\node[node, right = of endl, xshift = 2cm] (end) {{\small $end$}};
			
			\draw(end1) edge[edgelabelabove] node{$T^{inc}(x,y,\overline{y})$} (start2);
			\draw(end2) edge (startl);
			\draw(endl) edge[edgelabelabove] node{$T^{dec}(x,z,\overline{z})$} (end);
		\end{tikzpicture}
	\end{center}
	\caption{Construction of the final $\CVASS$ $\cC$.}
	\label{fig:final}
\end{figure}

\newcommand{\boldy}{\mathbf{y}}
Note that the set of counters of $\cC$ is the union of the set of counters of $\mach_{d}, \mach$ and $\mach_{e}$
which we will respectively denote by $\boldy,\bx,\bz$ respectively.

\subsubsection*{Observations about $\cC$} Before we state the initial and the final configurations of our reduction and prove its correctness, let us observe
some facts about the constructed $\CVASS$ $\cC$. 

\paragraph*{Fact 1: } Let $c$ be any configuration of $\cC$ whose state is $q_1$ such that $c(w) = d_{init}(w)$ for every counter $w \in \boldy$.
Let $c'$ be the configuration of $\cC$ whose state is $q_2$ and whose counter values are the same as
that of $c$, except that $c'(w) = d_{fin}(w)$ for every counter $w \in \boldy$.
Then, there is a run from $c$ to $c'$ of length exactly $3(2k + 4nm)$. Further, any run from $c$ to a configuration $c''$ such that
$c''(y) \ge c'(y)$ must necessarily visit $c'$ after exactly $3(2k+4nm)$ steps.

\paragraph*{Proof of Fact 1: } This is simply the assumption on the amplifier $\langle \mach_d, d_{init}, d_{fin} \rangle $ recast in terms of the 
$\CVASS$ $\cC$. By construction of $\cC$ and by the fact that $\langle \mach_d, d_{init}, d_{fin} \rangle$ is a $\CVASS$ amplifier, it follows that
there is a run from $c$ to $c'$ of length exactly $3(2k+4nm)$ in $\cC$. Further, suppose there is a run from $c$ to some configuration
$c''$ such that $c''(y) \ge c'(y)$. Since the counter $y$ is never affected after leaving the state $q_2$ to fire the rules
corresponding to $T^{inc}(x,y,\overline{y})$, we can assume that the supposed run from $c$ to $c''$ consists only of those rules belonging to $\mach_d$ in $\cC$.
By definition of an amplifier, the required claim then follows.

\paragraph*{Fact 2: } Let $c$ be any configuration of $\cC$ whose state is $q_2$ such that $c(w) = d_{fin}(w)$ for every counter $w \in \boldy$.
Let $c'$ be the configuration of $\cC$ whose state is $q_3$ and whose counter values are the same as that of $c$, except that
$c'(y_i) = 0$ and $c'(\overline{y_i}) = 1$ for every $1 \le i \le m+2$ and $c'(x_i) = c(x_i) + \alpha_i/2^k$ for every $1 \le i \le m$. 
Then, there is a run from $c$ to $c'$ of length exactly $m+2$. Further, any run from $c$ to a configuration $c''$ such that
$c''(y_i) \ge c'(y_i)$ and $c''(\overline{y_i}) \ge c'(\overline{y_i})$ for every $1 \le i \le m+2$ must necessarily visit
$c'$ after exactly $m+2$ steps.

\paragraph*{Proof of Fact 2: } By construction of the rules of the gadget $T^{inc}(x,y,\overline{y})$, it follows that there is a run
from $c$ to $c'$ of length exactly $m+2$. Further, suppose there is a run from $c$ to some configuration $c''$
such that $c''(y_i) \ge 0$ and $c''(\overline{y_i}) \ge 1$ for every $1 \le i \le m+2$. 
By definition of an amplifier, it follows that
the only outgoing rules from $q_2$ in $\cC$ are the edges corresponding to the rules in $T^{inc}(x,y,\overline{y})$.
By assumption on the configuration $c$, we have 
$c(y_i) = \alpha_i/2^k$ for every $1 \le i \le m$ and $c(\overline{y_i}) = 1 - c(y_i)$ for every $1 \le i \le m+2$. 
Since the counters $\{y_i,\overline{y_i} : 1 \le i \le m+2\}$ are never affected after the rules in $T^{inc}(x,y,\overline{y})$, 
we can assume that supposed run from $c$ to $c''$ consists only of those rules in $T^{inc}(x,y,\overline{y})$.
By construction of the rules in $T^{inc}(x,y,\overline{y})$ and by the definition of $c$ and $c'$, it follows
that this run must necessarily visit $c'$ after exactly $m+2$ steps.

\paragraph*{Fact 3: } Let $c$ be any configuration of $\cC$ whose state is $q_5$ such that $c(w) = e_{init}(w)$ for every counter $w \in \bz$. 
Let $c'$ be the configuration of $\cC$ whose state is $q_6$ and whose counter values are the same as that of $c$, except
that $c'(w) = e_{fin}(w)$ for every counter $w \in \bz$. 
Then there is a run from $c$ to $c'$ of length exactly $3(2k+4nm)$. 
Further, any run from $c$ to a configuration $c''$ such that $c''(z) \ge c'(z)$ must necessarily visit 
$c'$ after exactly $3(2k+4nm)$ steps.

\paragraph*{Proof of Fact 3: } The proof is similar to the proof of Fact 1.

\paragraph*{Fact 4: } Let $c$ be any configuration of $\cC$ whose state is $q_6$ such that $c(w) = e_{fin}(w)$ for every counter $w \in \bz$.
If $c(w) \ge c_{fin}(w)$ for every $w \in \bx$, there is a run from $c$ to $c'$ of length exactly $m+2$ where $c'$
is the configuration of $\cC$ whose state is $end$ and whose counter values are the same as that of $c$, except that
$c'(z_i) = 0$ and $c'(\overline{z_i}) = 1$ for every $1 \le i \le m+2$ and $c'(x_i) = c(x_i) - \beta_i/2^k$ for every $1 \le i \le m$. 
Further, any run from $c$ to a configuration $c''$ such that $c''(z_i) \ge 0$ and $c''(\overline{z_i}) \ge 1$ for every $1 \le i \le m+2$ must necessarily mean that $c(w) \ge c_{fin}(w)$ for every $w \in \bx$
and also that the run must necessarily visit $c'$ after exactly $m+2$ steps.

\paragraph*{Proof of Fact 4: } The proof is similar to the proof of Fact 2.

\subsubsection*{Correctness of the reduction.}

Now we prove the correctness of our reduction. Let $f_{init}$ be the configuration of $\cC$ whose state is $q_1$ such that $f_{init}(w) = d_{init}(w)$ for every $w \in \boldy$,
$f_{init}(w) = e_{init}(w)$ for every $w \in \bz$ and $f_{init}(w) = 0$ for every $w \in \bx$. Let $f_{fin}$ be the configuration of $\cC$ whose
state is $end$ such that $f_{fin}(w) = 0$ for every $w \in \{y_1,\dots,y_{m+2},
z_1,\dots,z_{m+2}\}$, $f_{fin}(w) = 1$ for every $w \in \{\overline{y_1},\dots,\overline{y_{m+2}},\overline{z_1},\dots,\overline{z_{m+2}}\}$,
$f_{fin}(y) = d_{fin}(y), f_{fin}(z) = e_{fin}(z)$ and $f_{fin}(w) = 0$ for every $w \in \bx$. Note that by the assumption on the amplifiers,
both $f_{init}$ and $f_{fin}$ can be encoded in unary. Set $\ell_\cC = 3(2k+4nm) + (m+2) + \ell + 1 + 3(2k+4nm) + (m+2)$. 
We now claim that,

\begin{theorem}
	$f_{init}$ can cover $f_{fin}$ in $\cC$ in $\ell_{\cC}$ steps in $\cC$ iff $c_{init}$ can cover $c_{fin}$ in $\ell$ steps in $\mach$.
\end{theorem}

\begin{proof}
	Suppose $c_{init}$ can reach a configuration $c$ which covers $c_{fin}$ in exactly $\ell$ steps in $\mach$ by a run $\rho$. 
	First let us define the following configurations:
	\begin{itemize}
		\item $f_2$ is the same configuration as $f_{init}$ except that its state is $q_2$ and $f_2(w) = d_{fin}(w)$ for every $w \in \boldy$.
		\item $f_3$ is the same configuration as $f_2$ except that its state is $q_3$ and $f_3(y_i) = 0, f_3(\overline{y_i}) = 1$ for every $1 \le i \le m+2$
		and $f_3(x_i) = c_{init}(x_i)$ for every $1 \le i \le m$.
		\item $f_4$ is the same configuration as $f_3$ except that its state is $q_4$ and $f_4(x_i) = c(x_i)$ for every $1 \le i \le m$.
		\item $f_5$ is the same configuration as $f_4$ except that its state is $q_5$.
		\item $f_6$ is the same configuration as $f_5$ except that its state is $q_6$ and $f_6(w) = e_{fin}(w)$ for every $w \in \boldy$.
		\item $f_7$ is the same configuration as $f_6$ except that its state is $end$ and $f_7(z_i) = 0, f_7(\overline{z_i}) = 1$ for every $1 \le i \le m+2$
		and $f_7(x_i) = f_6(x_i) - \beta_i/2^k$ for every $1 \le i \le m$.
	\end{itemize}
	
	By Facts 1, 2, 3 and 4, we have that
	$f_{init}$ can reach $f_2$ in exactly $3(2k+4nm)$ steps, $f_2$ can reach $f_3$ in exactly $m+2$ steps, $f_5$ can reach $f_6$ in exactly
	$3(2k+4nm)$ steps and $f_6$ can reach $f_7$ in exactly $m+2$ steps. Further, from the construction it is clear that $f_4$ can reach
	$f_5$ in exactly one step and also that $f_7$ covers $f_{fin}$. By assumption there is a run $\rho$ from $c_{init}$ to $c$ of exactly $m$
	steps in $\mach$. By construction of $f_3$ and $f_4$ it can be observed that the same run $\rho$ is also a run from $f_3$ to $f_4$
	in $\cC$. It then follows that we have a run from $f_{init}$ to $f_7$ of length exactly $\ell_C$.

	Suppose $f_{init}$ can reach a configuration $f$ which covers $f_{fin}$ in exactly $\ell_\cC$ steps in $\cC$ by a run $\rho$.
	First let us define the following configurations:
	\begin{itemize}
		\item $f_2$ is the same configuration as $f_{init}$ except that its state is $q_2$ and $f_2(w) = d_{fin}(w)$ for every $w \in \boldy$.
		\item $f_3$ is the same configuration as $f_2$ except that its state is $q_3$ and $f_3(y_i) = 0, f_3(\overline{y_i}) = 1$ for every $1 \le i \le m+2$
		and $f_3(x_i) = c_{init}(x_i)$ for every $1 \le i \le m$.
		\item $f_4$ is the same configuration as $f_3$ except that its state is $q_4$ and $f_4(x_i) = f(x_i) + \beta_i/2^k$ for every $1 \le i \le m$.
		\item $f_5$ is the same configuration as $f_4$ except that its state is $q_5$.
		\item $f_6$ is the same configuration as $f_5$ except that its state is $q_6$ and $f_6(w) = e_{fin}(w)$ for every $w \in \boldy$.
	\end{itemize}
	
	By Facts 1 and 2, it must be the case that after $3(2k+4nm)$ steps of $\rho$, $f_2$ is reached and after $m+2$ steps from there,
	$f_3$ is reached. Now let us look at the first points in the run $\rho$ when the state $q_4$ is reached
	and let us call the configuration at that point as $c_4$. 
	By assumption on $\mach$, there are no outgoing rules from $q_4$ in $\mach$
	and so the only way to move out of $q_4$ in $\cC$ is to take the rule to $q_5$ to reach a configuration $c_5$ whose
	counter values are the same as that of $q_4$. By construction of $\cC$, no counter in $\bz$ is affected
	before reaching the state $q_5$ and so this means that $c_5(w) = f_{init}(w) = e_{init}(w)$ for any $w \in \bz$.
	Since $f(z) \ge f_{fin}(z) = e_{fin}(z)$ and since we have a run from $c_5$ to $f$,
	by Fact 3, it must be the case that starting from $c_5$, after $3(2k+4nm)$ steps of $\rho$, we reach a 
	configuration $c_6$ which is the same as $c_5$ except that its state is $q_6$ and $c_6(w) = e_{fin}(w)$ for any $w \in \bz$.
	By Fact 4, it must be the case that starting from $c_6$, after $m+2$ steps of $\rho$, we reach a configuration
	$c_7$ whose state is $end$ and whose counter values are the same as that of $c_6$, except that
	$c_7(z_i) = 0$ and $c_7(\overline{z_i}) = 1$ for every $1 \le i \le m+2$ and $c_7(x_i) = c_6(x_i) - \beta_i/2^k$ for every $1 \le i \le m$. 
	Since there are no outgoing rules from $end$ in $\cC$, it follows that $c_7 = f$.
	This in turn implies that $c_6 = f_6, c_5 = f_5$ and $c_4 = f_4$. This further implies that the path from $f_3$ to $f_4$ in $\rho$
	is a path of length exactly $\ell$ consisting only of those rules from $\mach$. By definition of $f_3$ and $f_4$, this 
	run is also a run of length exactly $\ell$ in $\mach$ from the configuration $c_{init}$ to the configuration $c$ whose state is $q_4$ and $c(x_i) = f(x_i) + \beta_i/2^k$ for every $i$. Since $c$ covers $c_{fin}$, it follows that there is a run of length exactly $\ell$ in $\mach$ from $c_{init}$
	which covers $c_{fin}$.
\end{proof}


Finally, we make an observation on the number of counters that we have used for the reduction. Inspecting the arguments given in 
Stages 1 and 2, we already see that coverability in $\QVASSRL$ for super-structured instances is already $\NEXPTIME$-hard 
for 15 counters. The reduction given in Stage 3 gives us an instance of coverability in $\QVASSRL$ over unary encodings with
$m+2(2m+5)$ counters, where $m$ is the number of counters in the given super-structured instance of coverability in $\QVASSRL$.
Assuming Theorem~\ref{thm:amplifier}, it then follows that

\begin{theorem}
	Coverability in $\QVASSRL$ over unary encodings is $\NEXPTIME$-hard, already over $\CVASS$es of dimension 85.
\end{theorem}

Since the reduction from $\QVASSRL$ to $\QPVASS$ given in Subsection~\ref{subsec:cvass-to-cgvas} preserves coverability and produces a $\QPVASS$ of bounded stack-height, we get~\cref{thm:discussion-hardness}.

\subsection*{Stage 4: Constructing $\CVASS$ amplifiers. }

\newcommand{\boldp}{\mathbf{p}}
All that remains is to prove Theorem~\ref{thm:amplifier}, i.e., construct $\CVASS$ amplifiers. Let $p_1, p_2, \dots, p_m$ be natural numbers, all of them encoded in binary using exactly $n$ bits and let $k$ be a number such that each $p_i \le 2^k$. Let $\boldp$ be the vector $(p_1,\dots,p_m)$.
We want to construct a $\CVASS$ amplifier for $(\boldp,k)$ in time polynomial in $n, m$ and $k$. 
We do this in the following way. Instead of constructing a $\CVASS$ which is an amplifier for $(\boldp,k)$, we
first construct a $\NCCM$ which is an amplifier for $(\boldp,k)$. Then we apply the reduction from Section~\ref{sec:lower-bound-nccm-pvass} on this
$\NCCM$ amplifier and show that the resulting $\CVASS$ is also an amplifier for $(\boldp,k)$.

To do this, we first introduce gadgets which can do basic operations on counters such as addition, doubling and halving.
We have already seen all of these gadgets as part of the reduction given in Section~\ref{sec:lower-bound-mcm-nccm}. 

\subsubsection*{Gadgets for addition, doubling and halving.}

Consider the three $\NCCM$'s given in Figures~\ref{fig:add},~\ref{fig:doub},~\ref{fig:halve}, which
we call the addition, subtraction, doubling and halving gadgets
respectively. In these gadgets $x, st$ and $te$ are some three counters.
The addition and subtraction gadgets are parameterized by two counters $x$ and $st$
and the doubling and halving gadgets are parameterized by a single counter $x$.
The addition and the doubling gadgets are the same as the ones that appeared as 
gadgets for the $\inc$ and $\double$ rules respectively in Section~\ref{sec:lower-bound-mcm-nccm}.
The halving gadget is a simplification of the initialization gadget used in Section~\ref{sec:lower-bound-mcm-nccm}.

A configuration $C$ of any one of these gadgets $\mach_g$ is said to be \emph{good} if
\begin{itemize}
	\item $\mach_g$ is the addition gadget and $C(x) + C(st) \leq 1$.
	\item $\mach_g$ is the doubling gadget and $2C(x) \leq 1$.
	\item $\mach_g$ is the halving gadget and $C$ is any configuration.
\end{itemize}

The following lemma is immediate from the construction of these gadgets.
\begin{lemma}[The gadget lemma]\label{lem:gadget}
	Suppose $C$ is a good configuration of one of these gadgets $\mach_g$ such that
	$\sta{C} = q$ and $C(te) = 0$. Then, there is exactly one run
	of length two in any of these gadgets from $C$ which reaches a configuration
	$C'$ which is the same as $C$ except that $\sta{C'} = q'$ and
	\begin{itemize}
		\item $C'(x) = C(x) + C(st)$ if $\mach_g$ is the addition gadget.
		\item $C'(x) = 2C(x)$ if $\mach_g$ is the doubling gadget.
		\item $C'(x) = C(x)/2$ if $\mach_g$ is the halving gadget.
	\end{itemize}
\end{lemma}

\begin{figure}[h]
\begin{center}
	\tikzstyle{node}=[circle,draw=black,thick,minimum size=12mm,inner sep=0.75mm,font=\normalsize]
	\tikzstyle{edgelabelabove}=[sloped, above, align= center]
	\tikzstyle{edgelabelbelow}=[sloped, below, align= center]
	\begin{tikzpicture}[->,node distance = 3cm,scale=0.8, every node/.style={scale=0.8}]
		\node[node] (q) {$q$};
		\node[node, right = of q] (r) {$add$};
		\node[node, right = of r] (q') {$q'$};
		
		\draw(q) edge[edgelabelabove] node{$\cpp{x}, \cpp{te}, \cmm{st}; \ st = 0?$} (r);
		\draw(r) edge[edgelabelabove] node{$\cpp{st}, \cmm{te}; \ te = 0?$} (q');
	\end{tikzpicture}
\end{center}
	\caption{The addition gadget $Add(x,st)$. }
\label{fig:add}
\end{figure}
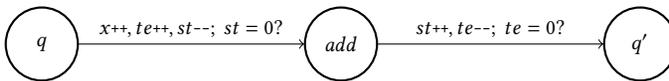 

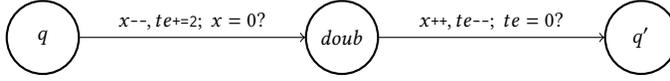
\begin{figure}[h]
	\begin{center}
		\tikzstyle{node}=[circle,draw=black,thick,minimum size=12mm,inner sep=0.75mm,font=\normalsize]
		\tikzstyle{edgelabelabove}=[sloped, above, align= center]
		\tikzstyle{edgelabelbelow}=[sloped, below, align= center]
		\begin{tikzpicture}[->,node distance = 3cm,scale=0.8, every node/.style={scale=0.8}]
			\node[node] (q) {$q$};
			\node[node, right = of q] (r) {$doub$};
			\node[node, right = of r] (q') {$q'$};
			
			\draw(q) edge[edgelabelabove] node{$\cmm{x}, \cpptwo{te}; \ x = 0?$} (r);
			\draw(r) edge[edgelabelabove] node{$\cpp{x}, \cmm{te}; \ te = 0?$} (q');
		\end{tikzpicture}
	\end{center}
	\caption{The doubling gadget $Doub(x)$. }
	\label{fig:doub}
\end{figure}

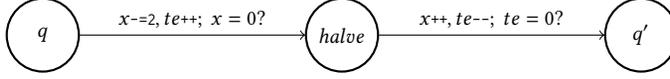
\begin{figure}[h]
	\begin{center}
		\tikzstyle{node}=[circle,draw=black,thick,minimum size=12mm,inner sep=0.75mm,font=\normalsize]
		\tikzstyle{edgelabelabove}=[sloped, above, align= center]
		\tikzstyle{edgelabelbelow}=[sloped, below, align= center]
		\begin{tikzpicture}[->,node distance = 3cm,scale=0.8, every node/.style={scale=0.8}]
			\node[node] (q) {$q$};
			\node[node, right = of q] (r) {$halve$};
			\node[node, right = of r] (q') {$q'$};
			
			\draw(q) edge[edgelabelabove] node{$\cmmtwo{x}, \cpp{te}; \ x = 0?$} (r);
			\draw(r) edge[edgelabelabove] node{$\cpp{x}, \cmm{te}; \ te = 0?$} (q');
		\end{tikzpicture}
	\end{center}
	\caption{The halving gadget $Halve(x)$. }
	\label{fig:halve}
\end{figure}

\subsubsection*{Constructing $\NCCM$ amplifiers.}

Given an $\NCCM$ $\mach$, its control flow graph is the graph obtained by taking the set of states of $\mach$ as its vertices
and connecting any two vertices by an edge if there is a rule between them.
We now define the notion of an $\NCCM$ amplifier for the tuple $(\boldp,k)$. 

\begin{definition}[$\NCCM$ amplifier.]
	An $\NCCM$ amplifier for the tuple $(\boldp,k)$ is a triple $\langle \cC,c_{init},c_{fin} \rangle$ such that $\cC$ is an $\NCCM$ and $c_{init}, c_{fin}$ are two of
	its configurations satisfying the following properties:
	\begin{itemize}
		\item $\cC$ has $m+2$ counters $x_1,\dots,x_{m},st,te$.
		\item The control flow graph of $\cC$ is a path of length exactly $2k+4nm$. 
		\item $c_{init}$ is 0 on every counter, except on $st$ where it takes the value 1. 
		\item $c_{fin}(x_i) = p_i/2^k$ for each $i$, $c_{fin}(st) = 1/2^k$ and $c_{fin}(te) = 0$.
		\item Starting from $c_{init}$ there is exactly one run of length $2k+4nm$. Further, this run ends at $c_{fin}$.
	\end{itemize}
\end{definition}

We now show that
\begin{theorem}
	In time polynomial in $n,m$ and $k$ we can construct a $\NCCM$ amplifier for $(\boldp,k)$.
\end{theorem}

\begin{proof}
Before we proceed with the construction, we make a small remark on notation.
In the sequel, we will describe $\NCCM$'s using statements of the form
"We move from state $q$ to $q'$ by adding the counter $st$ to the counter $x$".
This statement simply means that we attach the addition gadget $Add(x,st)$ in between
the state $q$ and $q'$. 
We employ a similar notation for doubling and halving as well.

%

Now we move on the construction of the $\NCCM$ amplifier. For the sake of simplicity, let us first handle the case when $\boldp$ is just 
a single number $p$. Hence, we want to construct an $\NCCM$ with counters $x_1, st$ and $te$. 
Let the binary encoding of $p$ in least significant bit notation be $b_1,b_2,\dots,b_n$.

We construct the $\NCCM$ as follows. 
First, we halve the value of $st$ for $k$ times, by introducing $k+1$ states $p_0,p_1,\dots,p_k$ and
by moving from $p_i$ to $p_{i+1}$ by halving the value of $st$. Note that since each halving operation
actually requires 2 rules, we can do this part by using exactly $2k$ rules.
Hence, if we initially began
with value 1 in $st$, after executing these $k$ rules, we would have reached the value $1/2^k$ in $st$.
Then if the bit $b_1$ is 0 we move from $p_k$ to a state $q_1$ by adding the value of
$st$ to $x_1$ two times and if the bit $b_1$ is 1, we move from $p_k$ to $q_1$
by adding the value of $st$ to $x_1$ exactly once. 
Hence, if $b_1$ is 0 we would have reached
$q_1$ with the value $2/2^k$ in the counter $x_1$ and otherwise we would have reached
$q_1$ with the value $1/2^k$ in the counter $x_1$. From here on, we introduce states
$q_2, q_3, \dots, q_n$ and for each $1 \le i \le n-1$, we move from $q_i$ to $q_{i+1}$
in the following manner: If the bit $b_{i+1}$ is 0, we move from $q_i$ to $q_{i+1}$
by doubling the value of the counter $x_1$ and if the bit $b_{i+1}$ is 1, we move
from $q_i$ to $q_{i+1}$ by first doubling the value of the counter $x_1$ and then
adding the value of $st$ to $x_1$. By introducing appropriate dummy rules which do nothing,
we can be assured that moving from $p_k$ to $q_1$ and moving from each $q_i$ to $q_{i+1}$ takes
exactly 4 rules. 

Note that the values of $c_{init}$ and $c_{fin}$
are fixed by the constraints on the definition of the amplifiers and we just need to specify their states in the $\NCCM$.
We set the initial state to be $p_0$ and the final state to be $q_n$.
By repeatedly applying the Gadget lemma (Lemma~\ref{lem:gadget}), we can then easily show
that starting from $c_{init}$ there is exactly one run of length $2k + 4n$ and further
that this run ends at $c_{fin}$.

This construction can be generalized in a straightforward manner to the case when $m > 1$. 
Essentially, we do the same process as above, where
we first halve the counter $st$ for $k$ times to arrive at the value $1/2^k$ in $st$ (by using $2k$ rules), 
then we set $x_1$ to $p_1/2^k$, $x_2$ to $p_2/2^k$ and so on till $x_m$ is set to $p_m/2^k$, all of them by the same procedure defined above (by using $4n$ rules for each $x_i$, thereby resulting in $4nm$ rules totally).

Finally, note that the control flow graph of the constructed $\NCCM$ is a path of length exactly $2k+4nm$.
\end{proof}

\subsubsection*{Constructing $\CVASS$ amplifiers. }

Now we construct $\CVASS$ amplifiers, i.e., the original amplifiers that we wanted to construct. 
Let $\langle \cC,c_{init},c_{fin} \rangle$ be a $\NCCM$ amplifier for the tuple $(\boldp,k)$ where the counters of $\cC$ are $x_1,\dots,x_m,st,te$.
Apply the reduction from Section~\ref{sec:lower-bound-nccm-pvass} on the $\NCCM$ reachability instance
$\langle \cC,c_{init},c_{fin},2k+4nm \rangle$ to get a $\QVASSRL$ instance
$\langle \mach,d_{init},d_{fin},3(2k+4nm) \rangle$. We claim that $\langle \mach,d_{init},d_{fin} \rangle$ is a $\CVASS$ amplifier for $(\boldp,k)$.

First, $\mach$ has $2m+5$ counters $x_1,\dots,x_m,st,te,\overline{x_1},\overline{x_2},\dots,\overline{x_m},\overline{st},\overline{te},ctrl$.
Next, since $c_{init}$ is encoded in unary so is $d_{init}$. Further, $d_{fin}(x_i) = c_{fin}(x_i) = p_i/2^k$ and for 
every $x \in \{x_1,\dots,x_m,st,te\}$, we have $d_{fin}(\overline{x}) = 1 - d_{fin}(x)$. Finally, $c_{fin}(ctrl)$ is simply $2(2k+4nm)$
which can also be encoded in unary. Since $\cC$ is an amplifier, it follows that the control flow graph of $\cC$ is a path of 
length exactly $2k+4nm$. Inspecting the structure of $\mach$, we notice that the control flow graph of $\mach$ is a path
of length exactly $3(2k+4nm)$
Hence, the first four requirements of being a $\CVASS$ amplifier are satisfied.

Now we come to the fifth requirement. Since there is a run from $c_{init}$ to $c_{fin}$ of length exactly $2k+4nm$,
by the correctness of the reduction, there is a run from $d_{init}$ to $d_{fin}$ of length exactly $3(2k+4nm)$.
Now, suppose there is some other run $\rho$ from $d_{init}$ to a configuration $d$ such that $d(ctrl) \ge d_{fin}(ctrl)$.
Inspecting the structure of $\mach$, we notice that the control flow graph of $\mach$ is a path
of length exactly $3(2k+4nm)$ such that only every second and third rule increases the counter $ctrl$ by one.
Since $\rho$ ends in a configuration whose $ctrl$ value is at least $2(2k + 4nm)$, it follows that
$\rho$ is a run of length exactly $3(2k+4nm)$
and also that $d(ctrl) = d_{fin}(ctrl) = 2(2k+4nm)$. 

By slicing the run $\rho$ into rules of three each, applying the completeness part of the control counter simulation lemma repeatedly
and using the fact that $\cC$ is an amplifier, we can then show that $d = d_{fin}$, thereby satisfying the fifth
requirement of being a $\CVASS$ amplifier.

	\newoutputstream{todos}
	\openoutputfile{main.todos.ctr}{todos}
	\addtostream{todos}{\arabic{@todonotes@numberoftodonotes}}
	\closeoutputstream{todos}

\end{document}